\documentclass[a4paper,UKenglish,autoref]{lipics-v2019}

\usepackage[dvipsnames]{xcolor}
\hypersetup{
    colorlinks=true,
    linkcolor=Red,
    linkbordercolor=BrickRed,
    citecolor=PineGreen,
    citebordercolor=OliveGreen,
    urlcolor=MidnightBlue}

\usepackage[noend]{algpseudocode}
\usepackage{algorithm}
\usepackage{algorithmicx}
\usepackage{thm-restate}
\usepackage{mathtools}
\usepackage[colorinlistoftodos]{todonotes}
\usepackage[T1]{fontenc}
\usepackage{relsize}
\usepackage{footmisc}
\usetikzlibrary{
  shapes.multipart,
  matrix,
  positioning,
  shapes.callouts,
  shapes.arrows,
  calc}

\def\bo{\mathcal{O}}
\def\Bo{\mathlarger{\mathcal{O}}}

\DeclarePairedDelimiter\floor{\lfloor}{\rfloor}
\newcommand{\func}[1]{\texttt{\textsc{#1}}}
\def\poly{\mathrm{poly}}

\def\filled{\textup{\textsf{filled}}}
\def\unfilled{\textup{\textsf{unfilled}}}

\def\ZERO{\mathsf{0}}
\def\ONE{\mathsf{1}}
\def\TRUE{\mathsf{true}}
\def\FALSE{\mathsf{false}}
\def\PHI{\phi}
\def\LAST{\mathbf{last}}

\def\ADJ{\mathbf{A}}

\usepackage{pifont}

\newcommand{\BD}{$\textsf{??}$}
\newcommand{\UBD}{{\Large \text{\ding{55}}}}

\algrenewcomment[1]{\hfill \textcolor{Brown}{\(\triangleright\) #1}}

\usepackage{bm}
\usepackage{wrapfig, framed, caption}
\usepackage{soul}
\sethlcolor{Red!20}

\renewcommand{\vec}[1]{\mathbf{#1}}
\graphicspath{{./images/}}

\theoremstyle{plain}
\newtheorem{invariant}[theorem]{Invariant}

\title{Local Access to Huge Random Objects through Partial Sampling}

\date{}
\author{Amartya Shankha Biswas}{CSAIL, MIT}{asbiswas@mit.edu}{https://orcid.org/0000-0002-5068-1524}{MIT Presidential Fellowship}
\author{Ronitt Rubinfeld}{CSAIL, MIT}{ronitt@csail.mit.edu}{}{NSF grants CCF-1650733, CCF-1733808, IIS-1741137 and CCF-1740751}
\author{Anak Yodpinyanee}{CSAIL, MIT}{anak@csail.mit.edu}{https://orcid.org/0000-0002-7572-2003}
{NSF grants CCF-1650733, CCF-1733808, IIS-1741137 and DPST scholarship, Royal Thai Government}

\authorrunning{A. Shankha Biswas, R. Rubinfeld, A. Yodpinyanee}

\Copyright{Amartya Shankha Biswas, Ronitt Rubinfeld, Anak Yodpinyanee}



\nolinenumbers

\begin{document}

\maketitle

\begin{abstract}

Consider an algorithm performing a computation on a \emph{huge random object} (for example a random graph or a ``long'' random walk).
Is it necessary to generate the entire object prior to the computation,
or is it possible to provide query access to the object and sample it incrementally ``on-the-fly'' (as requested by the algorithm)?
Such an \emph{implementation} should emulate the random object by answering queries in a manner consistent with
an instance of the random object sampled from the true distribution (or close to it).
This paradigm is useful when the algorithm is sub-linear and thus, sampling the entire object up front would ruin its efficiency.

Our first set of results focus on undirected graphs with independent edge probabilities,
i.e. each edge is chosen as an independent Bernoulli random variable.
We provide a general implementation for this model under certain assumptions.
Then, we use this to obtain the first efficient local implementations for the Erd\"os-R\'enyi $G(n,p)$ model for \emph{all} values of $p$,
and the Stochastic Block model.
As in previous local-access implementations for random graphs, we support \func{Vertex-Pair} and \func{Next-Neighbor} queries.
In addition, we introduce a new \func{Random-Neighbor} query.
Next, we give the first local-access implementation for \func{All-Neighbors} queries in the (sparse and directed) Kleinberg's Small-World model.
Our implementations require no pre-processing time, and answer each query using $ \mathcal{O}(\poly(\log n)) $ time, random bits, and additional space.

Next, we show how to implement random Catalan objects, specifically focusing on Dyck paths
(balanced random walks on the integer line that are always non-negative).
Here, we support \func{Height} queries to find the location of the walk,
and \func{First-Return} queries to find the time when the walk returns to a specified location.
This in turn can be used to implement \func{Next-Neighbor} queries on random rooted ordered trees,
and \func{Matching-Bracket} queries on random well bracketed expressions (the Dyck language).

Finally, we introduce two features to define a new model that:
(1) allows multiple independent (and even simultaneous) instantiations of the same implementation,
to be consistent with each other without the need for communication,
(2) allows us to generate a richer class of random objects that do not have a succinct description.
Specifically, we study uniformly random \emph{valid} $q$-colorings of an input graph $G$ with maximum degree $\Delta$.
This is in contrast to prior work in the area, where the relevant random objects are defined as a distribution with $\mathcal O(1)$ parameters
(for example, $n$ and $p$ in the $G(n,p)$ model).
The distribution over valid colorings is instead specified via a ``huge'' input (the underlying graph $G$),
that is far too large to be read by a sub-linear time algorithm.
Instead, our implementation accesses $G$ through local neighborhood probes,
and is able to answer queries to the color of any given vertex in sub-linear time for $q\ge 9\Delta$,
in a manner that is consistent with a specific random valid coloring of $G$.
Furthermore, the implementation is memory-less, and can maintain consistency with non-communicating copies of itself.

\end{abstract}

\newpage

\begingroup
  \tableofcontents
\endgroup
\newpage

\pagenumbering{arabic} 
\section{Introduction}
Consider an algorithm performing a computation on a \emph{huge random object} (for example a huge random graph or a ``long'' random walk).
Is it necessary to generate the entire object prior to the computation,
or is it possible to provide local query access to the object and generate it incrementally ``on-the-fly'' (as requested by the algorithm)?
Such an \emph{implementation} would ideally emulate the random object by answering appropriate queries,
in a manner that is consistent with a specific instance of the random object, sampled from the true distribution.
This paradigm is useful when we wish to simulate a sub-linear algorithm on a random object,
since sampling the entire object up front would ruin its efficiency.

In this work, we focus on generating huge random objects in a number of new settings,
including basic random graph models that were not previously considered, Catalan objects, and random colorings of graphs.
One emerging theme that we develop further is to provide access to random objects through more complex yet natural queries.
For example, consider an implementation for Erd\"os-R\'enyi $G(n,p)$ random graphs,
where the simplest query \func{Vertex-Pair}$(u,v)$ would ask about the existence of edge $(u,v)$,
which can be answered trivially by flipping a coin with bias $p$ (thus revealing a single entry in the adjacency matrix).
However, many applications involving non-dense graphs would benefit from \emph{adjacency list} access,
which we provide through \func{Next-Neighbor} and \func{Random-Neighbor} queries\footnote{
\label{foot:graph_queries}
\func{Vertex-Pair}$(u,v)$ returns whether $u$ and $v$ are adjacent, \func{Next-Neighbor}$(v)$ returns a new neighbor of $v$ each time
it is invoked (until none is left), and \func{Random-Neighbor}$(v)$ returns a uniform random neighbor of $v$ (if $v$ is not isolated).}

The problem of sampling partial information about huge random objects was pioneered in \cite{huge,huge_old,huge_journal},
through the implementation of \emph{indistinguishable} pseudo-random objects of exponential size.
Further work in \cite{sparse,reut} considers the implementation of different random graph models.
\cite{reut} introduced the study of $\func{Next-Neighbor}$ queries which provide efficient access to the adjacency list representation.
In addition to supporting \func{Vertex-Pair} and \func{Next-Neighbor},
we also introduce and implement a new \func{Random-Neighbor} query for undirected graphs with independent edge probabilities.


Finally, we define a new model that allows us to generate a richer class of random objects that do not have a succinct description.
Specifically, we study uniformly random \emph{valid} $q$-colorings of an input graph $G$ with max degree $\Delta$.
This is in contrast to prior work in the area, where random objects are defined as a distribution with $\mathcal O(1)$ parameters
(for example, $n$ and $p$ in the $G(n,p)$ model).
The distribution over valid colorings is instead specified via a ``huge'' input (the underlying graph $G$),
that is too large to be read by a sub-linear time algorithm.
Instead, our implementation accesses $G$ using local neighborhood probes.

This new model can be compared to \emph{Local Computation Algorithms}, which also implement query access to a consistent valid solution,
and read their input using local probes.
Inspired by this connection, we extend our model to support \emph{memory-less} implementations.
This allows multiple independent (possibly simultaneous) instantiations to agree on the same random object, without any communication.
We show how to implement local access to color of any given node in a random coloring, using sub-linear resources.
Unlike LCAs which can generate an \emph{arbitrary} valid solution, our model requires a uniformly random one.

\paragraph*{Random Graphs (Section~\ref{sec:undirected} and Section~\ref{sec:small_world})}%
\label{par:random_graphs}
Random graphs are one of the most well studied types of random object.
We consider the problem of implementing local access to a number of fundamental random graph models,
through natural queries, including \func{Vertex-Pair}, \func{Next-Neighbor},
and a newly introduced \func{Random-neighbor} query\footref{foot:graph_queries}, using polylogarithmic resources per query.
Our results on random graphs are summarized in Table~\ref{table:graph_results}.

\subparagraph*{Undirected Random Graphs with Independent Edge Probabilities:}
\label{par:undirected_random_graphs_with_independent_edge_probabilities}
We implement the aforementioned queries for the generic class of \emph{undirected graphs} with \emph{independent edge probabilities}
$\left\{ p_{uv} \right\}_{u,v\in V}$, where $p_{uv}$ denotes the probability that there is an edge between $u$ and $v$.
Under reasonable assumptions on the ability to compute certain values pertaining to consecutive edge probabilities,
our implementations support \func{Vertex-Pair}, \func{Next-Neighbor}, and \func{Random-Neighbor} queries\footref{foot:graph_queries},
using $\mathcal{O}(\poly(\log n))$ resources.
Note that in this setting, $\func{Vertex-Pair}$ queries are trivial by themselves,
since the existence of an edge depends on an independent random variable.
However, maintaining all three types of queries simultaneously is much harder.
As in \cite{reut} (and unlike many of the implementations presented in \cite{huge_old,huge}), our techniques allow unlimited queries,
and the generated random objects are sampled from the true distribution (rather than just being indistinguishable).
In particular, our construction yields local-access implementations for the Erd\"{o}s-R\'{e}nyi $G(n,p)$ model (for \emph{all} values of $p$),
and the Stochastic Block model with random community assignment.

\begin{table}[htpb]
\centering
\renewcommand{\arraystretch}{1}
\begin{tabular}{| >{\centering\arraybackslash}m{120pt} || >{\centering\arraybackslash}m{50pt} | >{\centering\arraybackslash}m{60pt} | >{\centering\arraybackslash}m{70pt} | >{\centering\arraybackslash}m{60pt} |}
    \hline
    \textbf{Model}              & \func{Vertex-Pair} & \func{Next-Neighbor} & \func{Random-Neighbor} & \func{All-Neighbors} \\      \hline \hline
    $G(n,p)$ with $p=\frac{\log^{\mathcal O(1)} n}{n}$
                                & \cite{sparse}      & \cite{sparse}        & \cite{sparse}          & \cite{sparse}        \\[5pt] \cline{1-5}
    $G(n,p)$ for arbitrary $p$  & This paper         & This paper           & This paper             & \UBD                 \\[5pt] \cline{1-5}
    Stochastic Block Model with $\log^{\mathcal O(1)} n$ communities
                                & This paper         & This paper           & This paper             & \UBD                 \\[5pt] \cline{1-5}
    BA Preferential Attachment  & \cite{reut}        & \cite{reut}          & \BD                    & \UBD                 \\[5pt] \cline{1-5}
    Small world model           & This paper         & This paper           & This paper             & This paper           \\[5pt] \cline{1-5}
    Random ordered rooted trees & This paper         & This paper           & This paper             & This paper           \\[5pt] \cline{1-5}
    \end{tabular}
    \vspace{0.7em}
    \caption{\textbf{(Local access implementations for random graphs):}
        All the implementations in this table use polylogarithimc time, additional space and random bits per query.
        A \UBD\ in the \func{All-Neighbors} column indicates that the graphs in this model may have un-bounded degree,
        and it is therefore impossible to sample \func{All-Neighbors} efficiently.
        A \BD\ entry indicates a sampling problem with no known efficient solution.}
    \label{table:graph_results}
\end{table}

\vspace{-1.5em}

While \func{Vertex-Pair} and \func{Next-Neighbor} queries 
have been considered in prior work \cite{reut, huge_old, sparse},
we provide the first implementations of these in non-sparse random graph models.
Prior results for implementing queries to $G(n,p)$ focused on the sparse case where $p = \log^{\mathcal O(1)} n/n$ \cite{sparse}.
The dense case where $p = \Theta(1)$ is also relatively simple because most of the adjacency matrix is filled,
and neighbor queries can be answered by performing $\Theta(1)$ \func{Vertex-Pair} queries until an edge is found.
The case of general $p$ is more involved, and was not considered previously.
For example, when $p = 1/\sqrt{n}$, each vertex has high degree $\mathcal O(\sqrt{n})$ but most of the adjacency matrix is empty,
thus making it difficult to generate a neighbor efficiently.
\func{Next-Neighbor} queries were introduced in \cite{reut}
in order to access the neighborhoods of vertices in non-sparse graphs in \emph{lexicographic order}.
This query however, does not allow us to efficiently explore graphs in some natural ways, such as via random walks,
since the initially returned neighbors are biased by the lexicographic ordering.
We address this deficiency by introducing a new \func{Random-Neighbor} query,
which would be useful, for instance, in sub-linear algorithms that employ random walk processes.
We provide the first implementation of all three queries in \emph{non-sparse graphs} as follows:

\begin{restatable}{theorem}{UndirectedGrand}
\label{thm:grand}
Given a random graph model defined by the probability matrix $\{ p_{uv}\}_{u,v\in [n]}$,
and assuming that we can compute the the quantities $\prod_{u=a}^b (1-p_{vu})$ and $\sum_{u=a}^b p_{vu}$ in $\mathcal O(\log^{\mathcal O(1)} n)$ time,
there exists an implementation for this model that supports local access through \func{Random-Neighbor}, \func{Vertex-Pair},
and \func{Next-Neighbor} queries using $\mathcal O(\log^{\mathcal O(1)} n)$ running time, additional space, and random bits per query.
\end{restatable}

We show that these assumptions can be realized in Erd\"os-R\'enyi random graphs and the Stochastic Block Model.
SBM presents additional challenges for assigning community labels to vertices (Section~\ref{sec:applications_to_random_graph_models}).

\begin{restatable}{corollary}{ERGrand}
There exists an algorithm that implements local access to a Erd\"{o}s-R\'{e}nyi $G(n,p)$ random graph,
through \func{Vertex-Pair}, \func{Next-Neighbor}, and \func{Random-Neighbor} queries,
while using $\bo(\log^3 n)$ time, $\bo(\log^3 n)$ random bits, and $\bo(\log^2 n)$ additional space per query with high probability.
\end{restatable}

\begin{restatable}{corollary}{SBMGrand}\label{cor:sbm-construct}
There exists an algorithm that implements local access to a random graph from the $n$-vertex Stochastic Block Model
with $r$ randomly-assigned communities,
through \func{Vertex-Pair}, \func{Next-Neighbor}, and \func{Random-Neighbor} queries,
using $O(r\,\poly(\log n))$ time, random bits, and space per query w.h.p.
\end{restatable}

We remark that while we are able to generate Erd\"{o}s-R\'{e}nyi random graphs on-the-fly supporting all three types of queries,
our construction still only requires $\widetilde{\mathcal O}(m+n)$ time and space to generate a complete $G(n,p)$ graph,
which is optimal up to logarithmic factors.

\begin{restatable}{corollary}{EROptimal}
\label{thm:er-optimal}
The final algorithm in Section~\ref{sec:undirected} can generate a complete random graph
from the Erd\"{o}s-R\'{e}nyi $G(n,p)$ model using overall
$\widetilde{\mathcal{O}}(n+m)$ time, random bits and space, which is $\widetilde{\bo}(pn^2)$ in expectation.
This is optimal up to $ \mathcal{O}(\poly(\log n))$ factors.
\end{restatable}

\subparagraph*{Directed Random Graphs - The Small World Model (Section~\ref{sec:small_world}):}
\label{par:directed_random_graphs}
We consider local-access implementations for directed graphs through Kleinberg's Small World model,
where the probabilities are based on distances in a 2-dimensional grid.
This model was introduced in \cite{kleinberg} to capture the geographical structure of real world networks,
in addition to reproducing observed properties of short routing paths and low diameter.
We implement \func{All-Neighbors} queries for the Small World model, using $\mathcal{O}(\poly(\log n))$ time, space and random bits.
Since such graphs are sparse, the other queries follow directly.
\begin{restatable}{theorem}{DirectedGrand}\label{DirectedGrand}
There exists an algorithm that implements \func{All-Neighbors} queries for a random graph from Kleinberg's Small World model,
where probability of including each directed non-grid edge $(u,v)$ in the graph is $c/(\func{dist}(u,v))^2$,
where the range of allowed values of $c$ is $\log^{\pm\mathcal O(1)} n$ and $\func{dist}$ denotes the Manhattan distance,
using $O(\log^2 n)$ time and random words per query with high probability.
\end{restatable}

\paragraph*{Catalan Objects (Section~\ref{sec:catalan_objects})}%
\label{par:intro_catalan_objects}
Catalan objects capture a well studied combinatorial property and admit numerous interpretations that include
well bracketed expressions, rooted trees and binary trees, Dyck languages etc.
In this paper, we focus on the Dyck path interpretation and implement local access to a uniformly random instance of a Dyck path.
Since Dyck paths have natural bijections to other Catalan objects,
we can use our Dyck Path implementation to obtain implementations of these random Catalan objects.

A Dyck path is defined as a sequence of $n$ upward $(+1)$ steps, and $n$ downward $(-1)$ steps,
with the added constraint that the \emph{sum of any prefix of the sequence is non-negative}.
A \emph{random} Dyck path may be viewed as a constrained one dimensional balanced random walk.
A natural query on Dyck paths is $\func{Height}(t)$ which returns the position of the walk after $t$ steps
(equivalently, sum of the first $t$ sequence elements).
$\func{Height}$ queries correspond to $\func{Depth}$ queries on rooted trees and bracketed expressions
(Section~\ref{sec:bijections_to_other_catalan_objects}).

We also introduce and support \func{First-Return} queries, where $\func{First-Return}(t)$ returns
the first time when the random walk returns to the same \emph{height} as it was at time $t$, as long as $\func{Height}(t+1)=\func{Height}(t)+1$
(Section~\ref{sec:bijections_to_other_catalan_objects} presents the rationale for this definition).
This allows us to support more involved queries that are widely used;
for example, \func{First-Return} queries are equivalent to finding the next child of a node in a random rooted tree,
and also to finding a matching closing bracket in random bracketed expressions.

\func{Height} queries for unconstrained random walks follow trivially from the implementation of interval summable functions in \cite{huge, histogram}.
However, the added non-negativity constraint introduces intricate non-local dependencies on the distribution of positions.
We show how to overcome these challenges, and support both queries using $\mathcal O(\poly(\log n))$ resources.

\begin{restatable}{theorem}{CatalanGrand}
\label{thm:catalan_main}
There is an algorithm using $\mathcal O(\log^{\mathcal O(1)} n)$ resources per query that provides access to
a uniformly random Dyck path of length $2n$ by implementing the following queries:
\begin{itemize}
    \item \func{Height}$(t)$ returns the position of the walk after $t$ steps.
    \item \func{First-Return}$(t)$: If $\func{Height}(t+1) > \func{Height}(t)$, then $\func{First-Return}(t)$ returns the smallest index $t'$,
    such that $t'>t$ and $\func{Height}(t') = \func{Height}(t)$ (i.e. the first time the Dyck path returns to the same height).
    Otherwise, if $\func{Height}(t+1) < \func{Height}(t)$, then $\func{First-Return}(t)$ is not defined.
\end{itemize}
\end{restatable}

\paragraph*{Random Coloring of Graphs - A New Model (Section~\ref{sec:random_coloring_of_a_graph}):}%
\label{par:random_coloring_of_graphs}
So far, all the results in this area have focused on random objects with a small description size;
for instance, the $G(n, p)$ model is described with only two parameters $n$ and $p$.
We introduce a new model for implementing random objects with \emph{huge input description};
that is, the distribution is specified as a uniformly random solution to a huge combinatorial problem.
The challenge is that now our algorithms cannot read the entire description in sub-linear time.

In this model, we implement query access to random $q$-colorings of a given huge graph $G$ with maximum degree $\Delta$.
A random coloring is generated by proposing $\mathcal O(n\log n)$ color updates and accepting the ones that do not create a conflict (Glauber dynamics).
This is an inherently sequential process with the acceptance of a particular proposal depending on all preceding neighboring proposals.
Moreover, unlike the previously considered random objects, this one has no succinct representation, and we can only uncover the proper distribution
by probing the underlying graph (in the manner of \emph{local computation algorithms} \cite{LCA, LCA_space_efficient}).
Unlike LCAs which can return an \emph{arbitrary} valid solution, we also have to make sure that we return a solution from the correct distribution.
We are able to construct an efficient oracle that returns the final color of a vertex using only a sub-linear number of probes when $q\ge 12\Delta$.

This implementation also has the feature that multiple independent instances of the algorithm having access to the same random bits,
will respond to queries in a manner consistent with each other; they will generate exactly the same coloring, regardless of the queries asked.
Since these implementations are memory-less, the resulting coloring is \emph{oblivious} to the order of queries,
and only depends on common random bits,

\begin{restatable}{theorem}{ColoringGrand}
\label{thm:coloring_grand}
Given neighborhood query access to a graph $G$ with $n$ nodes, maximum degree $\Delta$, and $q=3\alpha\Delta \ge 9\Delta$ colors,
we can generate the color of any given node from a distribution of color assignments that is $\epsilon$-close (in $L_1$ distance)
to the uniform distribution over all valid colorings of $G$, in a consistent manner,
using only $\mathcal O((n/\epsilon)^{4/\alpha}\Delta\log (n/\epsilon))$ time, probes, and public random bits per query.
\end{restatable}
This run-time is sub-linear when $\alpha > 4 \iff q > 12\Delta$.

\subsection{Related Work}
\label{sec:related_work}
The problem of computing local information of huge random objects was pioneered in \cite{huge_old,huge}.
Further work of \cite{sparse} considers the implementation of sparse $G(n,p)$ graphs from the Erd\"{o}s-R\'{e}nyi model \cite{er},
with $p = O(\poly(\log n)/n)$, by implementing \func{All-Neighbors} queries.
While these implementations use polylogarithmic resources over their entire execution,
they generate graphs that are only guaranteed to \emph{appear random} to algorithms that inspect a \emph{limited portion} of the generated graph.

In \cite{reut}, the authors construct an oracle for the implementation of recursive trees, and BA preferential attachment graphs.
Unlike prior work, their implementation allows for an arbitrary number of queries.
Although the graphs in this model are generated via a sequential process,
the oracle is able to locally generate arbitrary portions of it and answer queries in polylogarithmic time.
Though preferential attachment graphs are sparse, they contain vertices of high degree,
thus \cite{reut} provides access to the adjacency list through \func{Next-Neighbor} queries.
For additional related work, see Section~\ref{sec:additional_related_work}.

\subsection{Applications}
\label{sec:applications}
One motivation for implementing local access to huge random objects is so that we can simulate sub-linear time algorithms on them.
The standard paradigm of generating the entire (input) object a priori is wasteful,
because sub-linear algorithms only inspect a small fraction of the generated input.
For example, the greedy routing algorithm on Kleinberg's small world networks \cite{kleinberg}
only uses $\mathcal O(\log^2 n)$ probes to the underlying network.
Using our implementation, one can execute this algorithm on a random small world instance in $\mathcal O(\poly(\log n))$ time,
without incurring the $\mathcal O(n)$ prior-sampling overhead, by generating only those parts of the graph that are accessed by the routing algorithm.

Local access implementations may also be used to design parallel generators for random objects.
The model in Definition~\ref{def:local_access_LCA} is particularly suited to parallelization;
different processors/machines can generate parts of the random object independently,
using a \emph{read-only} shared memory containing public random bits.

\section{Model and Overview of our Techniques}
\label{sec:overview_of_our_techniques}
We begin by formalizing our model of \emph{local-access implementations}, inspired by \cite{reut},
but we add some aspects that were not addressed in earlier models.
First, we define families of random objects.

\begin{definition}
\label{def:parametrized_random_object}
A \textbf{random object family} maps a description $\Pi$ to a distribution $\mathsf X^{\Pi}$ over the set $\mathbb X^{\Pi}$.
\end{definition}

For example, the \emph{family} of Erd\"{o}s-R\'{e}nyi graphs maps $\Pi = (n, p)$ to a distribution over $\mathbb X^{\Pi}$,
which is the set of all possible $n$-vertex graphs,
where the probability assigned to any graph containing $m$ edges in the distribution $\mathsf X^{\Pi}$ is exactly $p^m\cdot (1-p)^{\binom{n}{2}-m}$.

\begin{definition}
\label{def:local_access}
Given a \textbf{random object family} $\{(\mathsf X^{\Pi}, \mathbb X^{\Pi})\}$ parameterized by $\Pi$, a \emph{local access implementation}
of a family of query functions $\langle F_1, F_2,\cdots \rangle$ where $F_i: \mathbb X^{\Pi}\rightarrow \{0,1\}$,
provides an oracle $\mathcal M$ with an internal state for storing a partially generated random object.
Given a description $\Pi$ and a query $F_i$, the oracle returns the value $\mathcal M(\Pi, F_i)$,
and updates its internal state, while satisfying the following:
\begin{itemize}
    \item \textbf{Consistency:}
    There must be a single $X\in \mathbb X^{\Pi}$, such that for all queries $F_i$ presented to the oracle,
    the returned value $\mathcal M(\Pi,F_i)$ equals the true value $F_i(X)$.
    \item \textbf{Distribution equivalence:}
    The random object $X\in \mathbb X^{\Pi}$ consistent with the responses $\{ F_i(X)\}$ must be sampled from some distribution $\hat{\mathsf{X}}^{\Pi}$
    that is $\epsilon$-close to the desired distribution $\mathsf{X}^{\Pi}$ in $L_1$-distance.
    In this work, we focus on supporting $\epsilon = n^{-c}$ for any desired constant $c>0$.
    \item \textbf{Performance:}
    The computation time, and random bits required to answer a single query must be sub-linear in $|X|$ with high probability,
    without any initialization overhead.
\end{itemize}
In particular, we allow queries to be made adversarially and non-deterministically.
The adversary has full knowledge of the algorithm's behavior and its past random bits.
\end{definition}

For example, in the $G(n,p)$ family with description $\Pi = (n, p)$,
we can define \func{Vertex-Pair} query functions $\{ F_{(u,v)}\}_{u,v\in [n]}$.
So, given a graph $G\in \mathbb X^{\Pi}$, the query $F_{(u,v)}(G) = 1$ if and only if $(u,v)\in G$.

In prior work \cite{reut, huge, sparse} as well as some of our results, the input description $\Pi$ is of small size (typically $\mathcal O(\log n)$),
and the oracle $\mathcal M$ can read all of $\Pi$ (for example, $\Pi = (n, p)$ in $G(n,p)$).

\subparagraph*{{\Large Distributions with Huge Description Size:}}
\label{par:distributions_with_huge_description_size}
We initiate the study of \textbf{random object families} where the description $\Pi$ is too large to be read by a sub-linear time algorithm.
In this setting, the oracle from Definition~\ref{def:local_access} cannot read the entire input $\Pi$, and instead accesses it through local probes.
For instance, consider the \textbf{random object family} that maps a graph $G$ to the uniform distribution over valid colorings of $G$.
Here, the description $\Pi$ includes the entire graph $G$, which is too large to be read by a sub-linear time algorithm.
In this case, the oracle can query the underlying graph using neighborhood probes.
The number of such probes used to answer a single query must be sub-linear in the input size.

\paragraph*{Supporting Independent Query Oracles: Memory-less Implementations}
\label{par:supporting_independent_query_oracles_memory_less_implementations}
The model in Definition~\ref{def:local_access} only supports sequential queries,
since the response to a future query may depend on the changes in internal state caused by past queries.
In some applications, we may want to have multiple independent query oracles whose responses are all consistent with each other.
One way to achieve this is to restrict our attention to \emph{memory-less} implementations; ones without any internal state.
An important implication of being memory-less is that the responses to each query is oblivious to the order of queries being asked.
In fact, the lack of internal state implies that independent implementations that use the same random bits and the same input description
must respond to queries in the same way.
Thus, instead of using the internal state to maintain consistency, memory-less implementations are given access to the same public random oracle.


For the problem of sampling a random graph coloring,
we present an implementation that is memory-less and also accesses the input description through local probes,
as elaborated in the following model:

\begin{definition}
\label{def:local_access_LCA}
Given a \textbf{random object family} $\{(\mathsf X^{\Pi}, \mathbb X^{\Pi})\}$ parameterized by input $\Pi$,
a \emph{local access implementation} of a family of query functions $\langle F_1, F_2,\cdots \rangle$,
provides an oracle $\mathcal M$ with the following properties.
$\mathcal M$ has query access to the input description $\Pi$, and a tape of public random bits $\vec R$.
Upon being queried with $\Pi$ and $F_i$, the oracle uses sub-linear resources to return the value $\mathcal M(\Pi,\vec R,F_i)$,
which must equal $F_i(X)$ for a specific $X\in\mathbb X^{\Pi}$, where the choice of $X$ depends only on $\vec R$,
and the distribution of $X$ (over $\vec R$) is $\frac1{n^c}$-close to the distribution over $\mathbb X^{\Pi}$, for any given constant $c$.
Thus, different instances of $\mathcal M$ with the same description $\Pi$ and the same random bits $\vec R$,
must agree on the choice of $X$ that is consistent with all answered queries regardless of the order and content of queries that were actually asked.
\end{definition}

We can contrast Definition~\ref{def:local_access_LCA} with the one for \emph{Local Computation Algorithms} \cite{LCA, LCA_space_efficient}
which also allow query access to \emph{some} valid solution by reading the input through local probes.
The additional challenge in our setting is that we also have to make sure that we return a uniformly random solution, rather than an arbitrary one.
We also note that the memory-less property may be achieved for small description size \textbf{random object families}.
For instance, our implementation for the directed small world model admits such a memory-less implementation using public random bits.

\subsection{Undirected Graphs}
\label{sec:undirected_graphs}
We consider the generic class of \emph{undirected graphs} with {\em independent edge probabilities} $\left\{ p_{uv} \right\}_{u,v\in V}$,
(where $p_{uv}$ denotes the probability that there is an edge between $u$ and $v$),
from which the results can be applied to Erd\"os-R\'enyi random graphs and the Stochastic Block Model.
Throughout, we identify our vertices via their unique IDs from $1$ to $n$.
In this model, \func{Vertex-Pair} queries by themselves can be implemented trivially,
since the existence of any edge $(u,v)$ is an independent Bernoulli random variable,
but it becomes harder to maintain consistency when implementing them in conjunction with the other queries.
Inspired by \cite{reut}, we provide an implementation of $\func{next-neighbor}$ queries,
which return the neighbors of any given vertex one by one in lexicographic order.
Finally, we introduce a new query: \func{Random-Neighbor} that returns a uniformly random neighbor of any given vertex.
This would be useful for any algorithm that performs random walks.
$\func{Random-Neighbor}$ queries in non-sparse graphs present particularly interesting challenges that are outlined below.

\paragraph*{\func{Next-Neighbor} Queries}
\label{par:next_neighbor_queries}
In Erd\"os-R\'enyi graphs, the (lexicographically) next neighbor of a vertex can be recovered by generating consecutive entries
of the adjacency matrix until a neighbor is found, which takes roughly $\Omega(1/p_{uv})$ time.
For small edge probabilities $p_{uv} = o(1)$, this implementation is inefficient, and we show how to improve the runtime to $\mathcal O(\poly(\log n))$.
Our main technique is to sample the number of ``non-neighbors'' preceding the next neighbor.
To do this, we assume that we can estimate the ``skip'' probabilities $F(v,a,b)=\prod^{b}_{u=a} (1-p_{vu})$,
where $F(v,a,b)$ is the probability that $v$ has no neighbors in the range $[a,b]$.
We later show how to compute this quantity efficiently for $G(n,p)$ and the Stochastic Block Model.

This strategy of \emph{skip-sampling} is also used in \cite{reut}.
However, in our work, the main difficulty arises from the fact that our graph is undirected,
and thus we must ``inform'' all (potentially $\Theta(n)$) non-neighbors once we decide on the query vertex's next neighbor.
Concretely, if $u'$ is sampled as the next neighbor of $v$ after its previous neighbor $u$,
we must maintain consistency in subsequent steps by ensuring that none of the vertices in the range $(u,u')$ return $v$ as a neighbor.
This update will become even more complicated as we handle \func{Random-Neighbor} queries, where we may generate non-neighbors at random locations.

In Section~\ref{sec:ER-rand}, we present a randomized implementation (Algorithm~\ref{alg:oblivious-coin-toss})
that supports \func{Next-Neighbor} queries efficiently, but has a complicated performance analysis.
We remark that this approach may be extended to support \func{Vertex-Pair} queries (Section~\ref{sec:reroll-cont}) with superior performance
(if we do not support \func{Random-Neighbor} queries),
and also to provide deterministic resource usage guarantee (Section~\ref{sec:ER-det}).

\paragraph*{\func{Random-Neighbor} Queries}
\label{par:random_neighbor_queries}
We implement \func{Random-Neighbor} queries (Section~\ref{sec:blocks}) using $\poly(\log n)$ resources.
The ability to do so is surprising since:
(1) Sampling the degree of vertex, may not be viable in \emph{sub-linear} time, because this quantity alone
imposes dependence on the existence of \emph{all} of its potential incident edges and consequently on the rest of the graph (since it is undirected).
Thus, our implementation needs to return a random neighbor, with probability reciprocal to the query vertex's degree,
without resorting to \emph{determining} its degree.
(2) Even without committing to the degrees, answers to \func{Random-Neighbor} queries
affect the conditional probabilities of the remaining adjacencies in a global and non-trivial manner.\footnote{
\label{conditional} Consider a $G(n,p)$ graph with small $p$, say $p = 1/\sqrt n$,
such that vertices will have $\tilde{\mathcal{O}}(\sqrt n)$ neighbors with high probability.
After $\tilde{\mathcal{O}}(\sqrt n)$ \func{Random-Neighbor} queries, we will have uncovered all the neighbors (w.h.p.),
so that the conditional probability of the remaining $\Theta(n)$ edges should now be close to zero.}

We formulate an approach which samples many consecutive edges simultaneously,
in such a way that the conditional probabilities of the unsampled edges remain independent and ``well-behaved'' during subsequent queries.
For each vertex $v$, we divide the potential neighbors of $v$ into contiguous ranges $\{B^{(i)}_v\}$ called blocks,
so that each $B^{(i)}_v$ contains $\Theta(1)$ neighbors in expectation (i.e. $\sum_{u\in B_i} p_{vu} = \Theta(1)$).
The subroutine of \func{Next-Neighbor} is applied to sample the neighbors within a block in expected $\mathcal{\widetilde O}(1)$ time.
We can now find a neighbor of $v$ by picking a random neighbor from a random block,
but this introduces a bias because all blocks may not have the same number of neighbors.
We remove this bias by rejecting samples with probability proportional to the number of neighbors in the block.

\subsubsection{Applications to Random Graph Models}
\label{sec:applications_to_random_graph_models}
We now consider the application of our construction above to actual random graph models,
where we must realize the assumption that $\prod^{b}_{u=a} (1-p_{vu})$ and $\sum^{b}_{u=a} p_{vu}$ can be computed efficiently.
For the Erd\"{o}s-R\'{e}nyi $G(n,p)$ model, these quantities have simple closed-form expressions.
Thus, we obtain implementations of $\func{Vertex-Pair}$, $\func{Next-Neighbor}$, and $\func{Random-Neighbor}$ queries,
using polylogarithmic resources per query, for \emph{arbitrary} values of $p$,
We remark that, while $\Omega(n+m) = \Omega(p n^2)$ time and space is clearly necessary to generate and represent a full random graph,
our implementation supports local-access via all three types of queries, and yet can generate a full graph in $\widetilde{O}(n+m)$ time and space
(Corollary~\ref{thm:er-optimal}).

We also generalize our construction to implement the Stochastic Block Model.
In this model, the vertex set is partitioned into $r$ \emph{communities} $\left\{ C_1, \ldots, C_r \right\}$.
The probability that an edge exists between $u\in C_i$ and $v \in C_j$ is $p_{ij}$.
A naive solution would be to simply assign communities to contiguous (by index) blocks of vertices,
which would easily allow us to calculate the relevant sums/products of probabilities on continuous ranges of indices,
with some additional case analysis to check when we are at a community boundary.
However, this setup is unrealistic, and not particularly useful in the context of the Stochastic Block model.
As communities in the observed data are generally unknown a priori,
and significant research has been devoted to designing efficient algorithms for community detection and recovery, these studies generally consider
the \emph{random community assignment} condition for the purpose of designing and analyzing algorithms \cite{mossel2015reconstruction}.
Thus, we construct implementations where the community assignments are sampled from some given distribution $\mathsf R$,
or from a collection of specified community sizes $\langle |C_1|, \ldots, |C_r|\rangle$.
The main difficulty here is to obtain a uniformly sampled assignment of vertices to communities on-the-fly.

Since the probabilities of potential edges now depend on the communities of their endpoints,
we can't obtain closed form expressions for the relevant sums and products of probabilities.
However, we observe that it suffices to efficiently count the number of vertices of each community in any range of contiguous vertex indices.
We then design a data structure extending a construction of \cite{huge}, which maintains these counts for ranges of vertices,
and determines the partition of their counts only on an as-needed basis.
This extension results in an efficient technique to sample counts from the \emph{multivariate hypergeometric distribution}
(Section~\ref{sec:multivariate_hypergeometric_sampling}) which may be of independent interest.
For $r$ communities, this yields an implementation with $ \mathcal{O}(r\cdot \poly(\log n))$ overhead in required resources for each operation.

\subsection{Directed Graphs}
\label{sec:directed_graphs}
Lastly, we consider Kleinberg's Small World model (\cite{kleinberg, klein}) in Section~\ref{sec:small_world}.
While small world models are proposed to capture properties such as small shortest-path
distances and large clustering coefficients \cite{watts1998collective},
this important special case of Kleinberg's model, defined on two-dimensional grids, demonstrates the underlying geographical structures of networks.

In this model, each vertex is identified via its 2D coordinate $v = (v_x, v_y) \in [\sqrt{n}]^2$.
Defining the Manhattan distance as $\func{dist}(u,v)=|u_x-v_x|+|u_y-v_y|$,
the probability that each directed edge $(u,v)$ exists is $c/(\func{dist}(u,v))^{2}$.
A common choice for $c$ is given by normalizing the distribution such that the expected out-degree of each vertex is 1 ($c = \Theta(1/\log n)$).
We can also support a range of values of $c=\log^{\pm\Theta(1)}n$.
Since the degree of each vertex in this model is $\mathcal O(\log n)$ with high probability, we implement \func{All-Neighbor} queries,
which in turn can emulate \func{Vertex-Pair}, \func{Next-Neighbor} and \func{Random-Neighbor} queries.
In contrast to our previous cases, this model imposes an underlying two-dimensional structure of the vertex set,
which governs the distance function as well as complicates the individual edge probabilities.

We implement \func{All-Neighbors} queries in the small world model by listing all neighbors from closest to furthest away from the queried vertex,
using $\poly(\log n)$ resources per query.
The main challenge is to sample for the next closest neighbor, when the probabilities are a function of the Manhattan distance on the lattice.
Rather than sampling for a neighbor directly, we partition the nodes based on their distance from $v$ (there are $\Theta(d)$ vertices at distance $d$).
We first choose the next smallest distance partition with a neighbor using rejection sampling techniques
(Lemma~\ref{lem:rejection_sampling}) on the appropriate distribution.
In the second step, we generate all the neighbors within that partition using \emph{skip-sampling}.

\subsection{Random Catalan Objects}
\label{sec:overview_catalan_objects}
Many important combinatorial objects can be interpreted as Catalan objects.
One such interpretation is a Dyck path; a one dimensional random walk on the line with $n$ up and $n$ down steps, starting from the origin,
with the constraint that the height is always non-negative.
We implement $\func{Height}(t)$, which returns the position of the walk at time $t$,
and \func{First-Return}$(t)$, which returns the first time when the random walk returns to the same position as it was at time $t$.
These queries are natural for several types of Catalan objects.
As noted previously, we can use standard bijections to translate the Dyck path query implementations into
natural queries for \emph{bracketed expressions} and \emph{ordered rooted trees}.
Specifically, \func{Height} values in Dyck paths are equivalent to \emph{depth} in bracket expressions and trees.
The \func{First-Return} queries are more involved, and are equivalent to finding the \emph{matching bracket} in bracket expressions,
and alternately to finding the \emph{next child} of a node in an ordered rooted tree (see Section~\ref{sec:bijections_to_other_catalan_objects}).

Over the course of the execution, our algorithm will determine the height of a random Dyck path at many different positions $\{ x_1, x_2,\cdots, x_m\}$
(with $x_i<x_{i+1}$), both directly as a result of user given $\func{Height}$ queries, and indirectly through recursive calls to $\func{Height}$.
These positions divide the sequence into contiguous \emph{intervals} $[x_i,x_{i+1}]$,
where the height of the endpoints $y_i, y_{i+1}$ have been determined, but none of the intermediate heights are known.
An important observation is that the unknown section of the path within an \emph{interval}
is entirely determined by the positions and heights of the endpoints, and in particular is completely independent of all other \emph{intervals}.
So, each interval $[x_i,x_{i+1}]$ along with corresponding heights $\{y_i,y_{i+1}\}$,
represents a generalized Dyck problem with $U$ up steps, $D$ down steps,
and a constraint that the path never dips more than $y_i$ units below the starting height.

\paragraph*{$\func{Height}$ Queries}
\label{par:height_queries}
General $\func{Height}(x)$ queries can then be answered by recursively halving the \emph{interval} containing $x$,
by repeatedly sampling the height at the midpoint, until the height of position $x$ is sampled.
We start by implementing a subroutine that given an \emph{interval} $[x_i,x_{i+1}]$ of length $2B$, containing $2U$ up and $2D$ down steps,
determines the number of up steps $U'=U+d$ assigned to the first half of the \emph{interval}
(we parameterize $U'$ with $d$ in order to make the analysis cleaner).
Note that this is equivalent to answering the query $\func{Height}(x_i+B)$.
This is done by sampling the parameter $d$ from a distribution $\{ p_d\}$ where $p_d \equiv S_{left}(d)\cdot S_{right}(d)/S_{total}$.
Here, $S_{left}(d)$ (respectively $S_{right}(d)$) is the number of possible paths in the left (resp. right) half of the \emph{interval} when
$U+d$ up steps and $D-d$ down steps are assigned to the first half, and $S_{total}$ is the number of possible paths in the original $2B$-interval.

The problem of determining the number of up steps in the first half of the \emph{interval} was solved for the unconstrained case
(where the sequence is just a random permutation of up and down steps) in \cite{huge}.
Adding the non-negativity constraint introduces further difficulties as the distribution over $d$ has a CDF that is difficult to compute.
We construct a different distribution $\{q_d\}$ that approximates $\{p_d\}$ pointwise to a factor of $\mathcal O(\log n)$
and has an efficiently computable CDF.
This allows us to sample from $\{q_d\}$ and leverage rejection sampling techniques
(see Lemma~\ref{lem:rejection_sampling} in Section~\ref{sec:basic_tools_for_efficient_sampling}) to obtain samples from $\{p_d\}$.

\paragraph*{$\func{First-Return}$ Queries}
\label{par:_first-return_queries}
Note that \func{First-Return}$(x)$ is only of interest when the first step after position $x$ is upwards (i.e. $\func{Height}(x+1)>\func{Height}(x)$),
since this situation is important for complex queries to random bracketed expressions and random rooted trees.
Section~\ref{sec:bijections_to_other_catalan_objects} details the rationale for this definition based on bijections between these objects.
\func{First-Return} queries are challenging because we need to find the \emph{interval} containing the first return to $\func{Height}(x)$.
Since there could be up to $\Theta(n)$ intervals, it is inefficient to iterate through all of them.
To circumvent this problem, we allow each interval $[x_i,x_{i+1}]$ to sample and maintain its own boundary constraint $k_i$
instead of using the global non-negativity constraint.
This implies that the path within the interval $[x_i,x_{i+1}]$ never reaches the height $y_i-k_i$ or lower.
Additionally, we maintain a crucial invariant that states that this boundary is achieved by the endpoint of lower height i.e. $\min(y_i,y_{i+1})$.
If the invariant holds, we can find the interval containing $\func{First-Return}(x)$ by finding the smallest determined position $x_j>x$
whose sampled height $y_j \le \func{Height}(x)$, and considering the interval $[x_{j-1},x_j]$ preceding $x_j$.
We use an interval tree to update and query for the known heights.

Unfortunately, every $\func{Height}$ query creates new intervals by sub-dividing existing ones, potentially breaking the invariant.
We re-establish the invariant for $[x_i,x_{i+1}]$ by generating a ``mandatory boundary'' $h$
(a boundary constraint with the added restriction that some position within the interval \emph{must} touch the boundary),
and then sampling a position $x_{mid}\in [x_i,x_{i+1}]$ such that $\func{Height}(x_{mid}) = h$ (Figure~\ref{fig:dyck_invariant_preserve}).
This creates new intervals $[x_i,x_{mid}]$ and $[x_{mid},x_{i+1}]$, both of which have a boundary constraint at $h$.

The first step of sampling the \emph{mandatory boundary} is performed
by binary searching on the possible boundary locations using an appropriate CDF (Section~\ref{sec:sampling_the_lowest_achievable_height}).
To find an intermediate position touching this boundary, we parameterize the position with $d$,
and find the distribution $\{p_d\}$ associated with the various possibilities.
Since we cannot directly sample from this complicated distribution, we define a \emph{piecewise continuous} probability distribution
$\hat q(\delta)$ that approximates $p_{\floor\delta}$ (Section~\ref{sec:sampling_first_position_touching_mandatory_boundary}).
We then use this to define a discrete distribution $\{q_d\}$ where $q_d = \int_d^{d+1}\hat q(\delta)$,
where we can efficiently compute the CDF of $\{q_d\}$ by integrating the piecewise continuous $\hat q(\delta)$.
The challenge here is to construct an appropriate $\hat q(\delta)$ that only has $\mathcal O(\log^{\mathcal O(1)} n)$ continuous pieces.
This allows us to again use the rejection sampling technique (Lemma~\ref{lem:rejection_sampling}) to indirectly obtain a sample from $\{p_d\}$.

\subsection{Random Coloring of a Graph}
\label{sec:overview_random_coloring_of_a_graph}
Finally, we introduce a new model (Definition~\ref{def:local_access_LCA}) for implementing huge random objects,
where the distribution is specified as a random solution to a huge combinatorial problem.
In this new setting, we will implement local query access to random $q$-colorings of a given huge graph $G$ of size $n$ with maximum degree $\Delta$.
Since the implementation has to run in sub-linear time, it is not possible to read the entire input $G$ during a single query execution.

\paragraph*{$\func{Color}$ Queries}
\label{par:color_queries}
Given a graph $G$ with maximum degree $\Delta$, and the number of colors $q\ge 9\Delta$,
we are able to construct an efficient implementation for $\func{Color}(v)$ that returns the final color of $v$
in a uniformly random $q$-coloring of $G$ using only a sub-linear number of probes.
Random colorings of a graph are sampled using $\mathcal O(n\log n)$ iterations of a Markov chain \cite{glauber_survey}.
Each step of the chain proposes a random color update for a random vertex, and accepts the update if it does not create a conflict.
This is an inherently sequential process, with the acceptance of a particular proposal depending on all preceding neighboring proposals.

To make the runtime analysis simpler, we use a modified version of Glauber Dynamics that proceeds in $\mathcal O(\log n)$ epochs.
In each epoch, all of the $n$ vertices propose a random color and update themselves if their proposals do not conflict with any of their neighbors.
This Markov chain was presented in \cite{what_local} for distributed graph coloring, and mixes in $\mathcal O(\log n)$ epochs when $q\ge 9\Delta$.
In order to implement the query $\func{Color}(v)$, it suffices to implement $\func{Accepted}(v,t)$
that indicates whether the proposal for $v$ was accepted in the $t^{th}$ epoch.
This depends on the prior colors of the potentially $\Delta$ neighbors of $v$.
Determining the prior colors of all the neighbors $w$ using recursive calls would result in
$\Delta$ invocations of $\func{Accepted}(w, t-1)$ (at the preceding epoch $t-1$).
Naively, this gives a bound of $\Delta^t$ on the number of invocations.
We can prune the recursions by only considering neighbors who proposed the color $c$ during \emph{some} past epoch.
This reduces the expected number of recursive calls to $t\Delta/q$,
since there are $t\Delta$ potential proposals and each one is $c$ with probability $1/q$.
If $q$ is larger than $t\Delta$, the number of recursive calls is less than $1$,
which gives a sub-linear bound on the total number of resulting invocations.
Since the number of epochs $t$ can be as large as $\Theta(\log n)$, this strategy will only work when $q = \Omega(\Delta\log n)$.

The improvement to $q = \Omega(\Delta)$ follows from the observation that for a neighbor $w$ that proposed color $c$ at epoch $t'$,
the recursive call corresponding to $w$ can directly jump to epoch $t'$.
If the conflicting color $c$ was indeed accepted at that epoch,
we then step \emph{forwards} through epochs $t'+1, t'+2,\cdots$, to check whether $c$ was overwritten by some future accepted proposal.
This strategy dramatically reduces the recursive sub-problem size (given by the epoch number $t'$),
and furthermore we show that we do not have to step through too many future epochs in order to check whether $c$ was overwritten.
This allows us to bound the runtime by $\widetilde{\mathcal O}\left(t\Delta (n/\epsilon)^{6.12\Delta/q}\right)$,
where the overall coloring is sampled from a distribution that is $\epsilon$-close to uniform (see Definition~\ref{def:local_access_LCA}).

One requirement for our strategy is the ability to access a \emph{valid initial coloring} (the initial state of the Markov Chain)
through local probes, in addition to local probes to the underlying graph structure.
This assumption can be removed by using a result of \cite{coloring_initialize},
that presents an LCA for $\Delta+1$ graph coloring using $\Delta^{\mathcal O(1)}\log n$ graph probes.
Alternately, we can assign random initial colors to the vertices, which may result in an \emph{invalid} final coloring.
However, the Markov Chain will transform the initial invalid coloring into a valid one with high probability.

\subsection{Basic Tools for Efficient Sampling}
\label{sec:basic_tools_for_efficient_sampling}
In this section, we describe the main techniques used to sample from a distribution $\{ p_i\}_{i\in [n]}$,
which differs based on the type of access to $\{p_i\}$ provided to the algorithm.
If the algorithm is given cumulative distribution function (CDF) access to $\{p_i\}$,
then it is well known that via $\mathcal O(\log n)$ CDF evaluations, one can sample according
to a distribution that is at most $n^{-c}$ far from $\{p_i\}$ in $L_1$ distance.

Sampling can be more challenging when when we can only access the probability distribution function (PDF) of $\{p_i\}$,
The approach that we use in this work is to construct an auxiliary distribution $\{q_i\}$ such that:
(1) $\{ q_i\}$ has an efficiently computable CDF, and
(2) $q_i$ approximates $p_i$ pointwise to within a polylogarithmic multiplicative factor for ``most'' of the support of $\{ p_i\}$.
The following Lemma inspired by \cite{huge} formalizes this concept.
\begin{lemma}
\label{lem:rejection_sampling}
Let $\{p_i\}_{i\in [n]}$ and $\{q_i\}_{i\in [n]}$ be distributions on $[n]$ satisfying the following conditions:
\begin{enumerate}
    \item There is a $\log^{\mathcal O(1)}n$ time algorithm to approximate $p_i$ and $q_i$
    up to a multiplicative $\left(1\pm \frac{1}{n^c}\right)$ factor.
    \item We can generate an index $i$ according to a distribution $\{\hat q_i\}$,
    where $\hat q_i$ is a $\left(1\pm \frac{1}{n^c}\right)$ multiplicative approximation to $q_i$.
    \item There exists a $poly(\log n)$-time recognizable set $S$ such that
    \begin{itemize}
        \item $1-\sum\limits_{i\in S} p_i < \frac 1{n^c}$
        \item For every $i\in S$, it holds that $p_i\le \log^{\mathcal{O}(1)} n\cdot q_i$
    \end{itemize}
\end{enumerate}
Then, with high probability we can use only $\log^{\mathcal O(1)}n$ samples from $\{\hat q_i\}$ to generate an index $i$
according to a distribution that is $ \mathcal O\left(\frac{1}{n^c}\right)$-close to $\{p_i\}$ in $L_1$ distance.
\end{lemma}
\begin{proof}
We begin by setting an upper bound $U = \log^{\mathcal O(1)}n$ on $p_i/q_i$ for all $i\in S$.
The sampling proceeds in iterations, such that in each iteration we obtain an index $i$ according to the distribution $\{ \hat q_i\}$.
If $i\not\in S$, this index is returned with probability $\tilde p_i/U\tilde q_i$,
where $\tilde p_i$ and $\tilde q_i$ are the $\left( 1\pm \frac{1}{n^c} \right)$ multiplicative approximations to $p_i$ and $q_i$,
Otherwise, we repeat this process until some output is returned.

The probability of returning index $i\in S$ in a particular iteration is $\hat q_i\cdot (\tilde p_i/U \tilde q_i)$,
which is in turn a $\left( 1\pm \frac{1}{n^c} \right)$ multiplicative approximation to $p_i/U$.
Hence, the probability of success in a single iteration is roughly $1/U$,
and therefore we only need $\mathcal O(U\log n) = \log^{\mathcal O(1)}n$ iterations (and the same number of samples from $\{ \hat q_i\}$)
in order to succeed with high probability.
The resulting distribution of indices approximates $\{ p_i\}$ pointwise on the domain $S$, up to a factor of $\left( 1\pm \frac{1}{n^c} \right)$.
Since the remainder of the domain contains at most $ \frac{1}{n^c}$ probability mass,
the output distribution is $\mathcal O\left(\frac{1}{n^c}\right)$ close to $\{ p_i\}$ in $L_1$ distance.
\end{proof}

\section{Local-Access Implementations for Random Undirected Graphs}
\label{sec:undirected}

In this section, we provide an efficient local access implementations for random undirected graphs
when the probabilities $p_{uv}=\mathbb P[(u,v)\in E]$ are given, and we can efficiently approximate the following quantities:
(1) the probability that there is no edge between a vertex $u$ and a range of consecutive vertices $[a,b]$, namely $\prod_{u=a}^b (1-p_{vu})$, and
(2) the sum of the edge probabilities (i.e., the expected number of edges) between $u$ and vertices from $[a,b]$, namely $\sum_{u=a}^b p_{vu}$.
In Section~\ref{sec:applications}, we provide subroutines for computing these values for the Erd\"{o}s-R\'{e}nyi model and the Stochastic Block model.
We also begin by assuming perfect-precision arithmetic, which we relax in Section~\ref{sec:remove-perfect}.

First, consider the adjacency matrix $\ADJ$ of $G$, where each entry $\ADJ[u][v]$ can exist in three possible states:
$\ADJ[u][v] = \ONE$ or $\ZERO$ if the algorithm has determined that $\{u,v\}\in E$ or $\{u,v\} \notin E$ respectively,
and $\ADJ[u][v] = \PHI$ if whether $\{u,v\}\in E$ or not will be determined by future random choices
(in fact, the marginal probability of $\mathbb P[(u,v)\in E]$ conditioned on all prior samples is still $p_{uv}$).
Our implementation also maintains the vector $\LAST$ (used in the same sense as \cite{reut}),
where $\LAST[v]$ records the neighbor of $v$ returned in the last call \func{Next-Neighbor}$(v)$, or $\LAST[v]=0$ if no such call has been invoked.
All cells of $\ADJ$ and $\LAST$ are initialized to $\PHI$ and $0$, respectively.

We use the Bernoulli random variable $X_{uv} \sim \mathsf{Bern}(p_{uv})$ when sampling the value of $\ADJ[u][v]=\phi$.
For the sake of analysis, we will frequently view our random process as if the \emph{entire} table of random variables $X_{uv}$ has been sampled
\emph{up-front}, and the algorithm simply ``uncovers'' these variables instead of making coin-flips.
Thus, every cell $\ADJ[u][v]$ is originally $\PHI$, but will eventually take the value $X_{uv}$.

\subparagraph*{Obstacles for maintaining $\ADJ$ explicitly:}
Consider a naive implemention that fills out the cells of $\ADJ$ one-by-one as required by each query;
equivalently, we perform $\func{Vertex-Pair}$ queries on successive vertices until a neighbor is found.
There are two problems with this approach.
Firstly, the algorithm only finds a neighbor, for a \func{Random-Neighbor} or \func{Next-Neighbor} query, with probability $p_{uv}$:
for $G(n,p)$ this requires $1/p$ iterations, which is already infeasible for $p = o(1/\poly(\log n))$.
Secondly, the algorithm may generate a large number of non-neighbors in the process, possibly in random or arbitrary locations.

\subsection{\func{Next-Neighbor} Queries via Run-of-$\ZERO$'s Sampling}\label{sec:ER-rand}

We implement \func{Next-Neighbor}$(v)$ by sampling for the first index $u > \LAST[v]$ such that $X_{vu}=\ONE$,
from a sequence of Bernoulli RVs $\{X_{v, u}\}_{u > \LAST[v]}$.
To do so, we sample a consecutive ``run'' of $\ZERO$'s with probability $\prod_{u=\LAST[v]+1}^{u'} (1-p_{vu})$:
this is the probability that there is no edge between a vertex $v$ and any $u \in (\LAST[v],u']$, which can be computed efficiently by our assumption.
The problem is that, some entries $\ADJ[v][u]$'s in this run may have already been determined (to be $\ONE$ or $\ZERO$)
by queries $\func{Next-Neighbor}(u)$ for $u > \LAST[v]$.
To mitigate this issue, we give a succinct data structure that determines the value of $\ADJ[v][u]$ for $u >\LAST[v]$ and,
more generally, captures the state $\ADJ$, in Section~\ref{sec:nn-ds}.
Using this data structure, we ensure that our sampled run does not skip over any $\ONE$.
Next, for the sampled index $u$ of the first occurrence of $\ONE$,
we check against this data structure to see if $\ADJ[v][u]$ is already assigned to $\ZERO$, in which case we re-sample for a new candidate $u' > u$.
Section~\ref{sec:nn-correctness} discusses the subtlety of this issue.

We note that we do not yet try to handle other types of queries here yet.
We also do not formally bound the number of re-sampling iterations of this approach here, because the argument is not needed by our final algorithm.
Yet, we remark that $O(\log n)$ iterations suffice with high probability, even if the queries are adversarial.
This method can be extended to support \func{Vertex-Pair} queries (but unfortunately not \func{Random-Neighbor} queries).
See Section~\ref{sec:reroll-cont} for full details.

\subsubsection{Data structure}\label{sec:nn-ds}

From the definition of $X_{uv}$, \func{Next-Neighbor}$(v)$ is given by $\min\{u > \LAST[v]: X_{vu} = \ONE\}$.
Let $P_v = \{u: \ADJ[v][u]=1\}$ be the set of known neighbors of $v$, and $w_v = \min \{(P_v \cap (\LAST[v], n])\}$ be its first known neighbor
not yet reported by a $\func{Next-Neighbor}(v)$ query, or equivalently, the next occurrence of $\ONE$ in $v$'s row on $\ADJ$ after $\LAST[v]$.
If there is no known neighbor of $v$ after $\LAST[v]$, we set $w_v = n+1$.
Consequently, $\ADJ[v][u] \in \{\PHI,\ZERO\}$ for all $u \in (\LAST[v], w_v)$,
so \func{Next-Neighbor}$(v)$ is either the index $u$ of the first occurrence of $X_{vu} = \ONE$ in this range, or $w_v$ if no such index exists.

We keep track of $\LAST[v]$ in a dictionary, to avoid any initialization overhead.
Each $P_v$ is maintained as an ordered set, which is also instantiated when it becomes non-empty.
When \func{Next-Neighbor}$(v)$ returns $u$, we add $v$ to $P_u$ and $u$ to $P_v$.
We do not attempt to maintain $\ADJ$ explicitly, as updating it requires replacing up to $\Theta(n)$ $\PHI$'s to $\ZERO$'s
for a single \func{Next-Neighbor} query in the worst case.
Instead, we argue that $\LAST$ and $P_v$'s provide a succinct representation of $\ADJ$ via the following observation.

\begin{restatable}{lemma}{cond}\label{lem:cond-0}
The data structures $\LAST$ and $P_v$'s together provide a succinct representation of $\ADJ$ when only \func{Next-Neighbor} queries are allowed. In particular, $\ADJ[v][u]=\ONE$ if and only if $u \in P_v$. Otherwise, $\ADJ[v][u]=\ZERO$ when $u <\LAST[v]$ or $v < \LAST[u]$. In all remaining cases, $\ADJ[v][u]=\PHI$.
\end{restatable}
\begin{proof}
The condition for $\ADJ[v][u]=\ONE$ clearly holds by constuction.
Otherwise, observe that $\ADJ[v][u]$ becomes \emph{decided} (i.e. its value is changed from $\PHI$ to $\ZERO$)
during the first call to \func{Next-Neighbor}$(v)$ that returns a value $u' > u$ thereby setting $\LAST[v] = u' \implies u<\LAST[v]$, or vice versa.
\end{proof}

\subsubsection{Queries and Updates}
\label{sec:nn-correctness}
\begin{wrapfigure}[15]{R}{0.5\textwidth}
\vspace{-3.0em}
\begin{framed}
    \renewcommand\figurename{Algorithm}
    \caption{Sampling \func{Next-Neighbor}}
    \label{alg:oblivious-coin-toss}
    \begin{algorithmic}[1]
        \Procedure{Next-Neighbor}{$v$}
            \State{$u \gets \LAST[v]$}
            \State{$w_v \gets \min \{(P_v \cap (u, n]) \cup \{n+1\}\}$}
            \Repeat
                \State{\textbf{sample} $F\sim\mathsf{F}(v,u,w_v)$}
                \State{$u \gets F$}
            \Until{$u = w_v$ or $\LAST[u] < v$}
            \If{$u \neq w_v$}
                \State{$P_v \gets P_v \cap \{u\}$}
                \State{$P_u \gets P_u \cap \{v\}$}
            \EndIf
            \State{$\LAST[v] \gets u$}
            \State \Return $u$
        \EndProcedure
    \end{algorithmic}
\end{framed}
\end{wrapfigure}
We now present Algorithm~\ref{alg:oblivious-coin-toss}, and discuss the correctness of its sampling process.
The argument here is rather subtle and relies on viewing the process as an ``uncovering'' of the table of RVs $X_{uv}$
(introduced in Section~\ref{sec:undirected}).
Consider the following strategy to find \func{Next-Neighbor}$(v)$ in the range $(\LAST[v],w_v)$.
Suppose that we generate a sequence of $w_v-\LAST[v]-1$ independent coin-tosses,
where the $i^{th}$ coin $C_{vu}$ corresponding to $u = \LAST[v]+i$ has bias $p_{vu}$, regardless of whether $X_{vu}$ is decided or not.
Then, we use the sequence $\langle C_{vu} \rangle$ to assign values to \emph{undecided} random variables $X_{vu}$.
The main observation here is that, the \emph{decided} random variables $X_{vu} = \ZERO$ do not need coin-flips,
and the corresponding coin result $C_{vu}$ can be discarded.
Thus, we generate coin-flips until we encounter some $u$ satisfying both $C_{vu} = \ONE$ and $\ADJ[v][u] = \PHI$.

Let $\mathsf{F}(v,a,b)$ denote the probability distribution of the occurrence $u$ of the first coin-flip $C_{vu} = \ONE$ among the neighbors in $(a, b)$.
More specifically, $F\sim\mathsf{F}(v,a,b)$ represents the event that $C_{v,a+1}=\cdots=C_{v,F-1}=\ZERO$ and $C_{v,F}=\ONE$,
which happens with probability $\mathbb P[F=f]=\prod_{u=a+1}^{f-1} (1-p_{vu}) \cdot p_{vf}$.
For convenience, let $F = b$ denote the event where all $C_{vu}=\ZERO$. Our algorithm samples $F_1\sim\mathsf{F}(v,\LAST[v],w_v)$ to find the first occurrence of $C_{v,F_1} = \ONE$, then samples $F_2\sim\mathsf{F}(v,F_1,w_v)$ to find the second occurrence $C_{v,F_2} = \ONE$, and so on.
These values $\{F_i\}$ are iterated as $u$ in Algorithm~\ref{alg:oblivious-coin-toss}.
This process generates $u$ satisfying $C_{vu}=1$ in increasing order,
until we find one that also satisfies $\ADJ[u][v]=\phi$ (this outcome is captured by the condition $\LAST[u] < v$),
or until the next generated $u$ is equal to $w_v$.
Note that once the process terminates at some $u$, we make no implications on the results of any uninspected coin-flips after $C_{vu}$.

\subparagraph*{Obstacles for extending beyond \func{Next-Neighbor} queries:}
There are two main issues that prevent this method from supporting \func{Random-Neighbor} queries.
Firstly, while one might consider applying \func{Next-Neighbor} starting from some random location $u$, to find the minimum $u' \geq u$
where $\ADJ[v][u']=\ONE$, the probability of choosing $u'$ will depend on the probabilities $p_{vu}$'s, and is generally not uniform.
Secondly, in Section~\ref{sec:nn-ds}, we observe that $\LAST[v]$ and $P_v$ together provide a succinct representation of $\ADJ[v][u] = \ZERO$
only for contiguous cells $\ADJ[v][u]$ where $u \leq \LAST[v]$ or $v \leq \LAST[u]$: they cannot handle $\ZERO$ anywhere else.
Unfortunately, in order to support \func{Random-Neighbor} queries, we would need to assign $\ADJ[v][u]$ to $\ZERO$ in random locations
beyond $\LAST[v]$ or $\LAST[u]$.
This cannot be done by the current data structure.
Specifically, to speed-up the sampling process for small $p_{vu}$'s, we must generate many random non-neighbors at once,
but we cannot afford to spend time linear in the number of $\ZERO$'s to update our data structure.
We remedy these issues via the following approach.

\subsection{Final Implementation Using Blocks}
\label{sec:blocks}
We begin this section by focusing first on \func{Random-Neighbor} queries, then extend the construction to the remaining queries.
In order to handle \func{Random-Neighbor}$(v)$, we divide the neighbors of $v$ into \emph{blocks}
$\vec B_v = \{ B^{(1)}_v, B^{(2)}_v, \ldots, B^{(i)}_v, \ldots\}$,
so that each block contains, in expectation, roughly the same number of neighbors of $v$.
We implement \func{Random-Neighbor}$(v)$ by randomly selecting a block $B^{(i)}_v$,
filling in entries $\ADJ[v][u]$ for $u \in B^{(i)}_v$ with $\ONE$'s and $\ZERO$'s, and then reporting a random neighbor from this block.
As the block size may be large when the probabilities are small, instead of using a linear scan,
our \func{Fill} subroutine will be implemented using the ``run-of-$\ZERO$s'' sampling from  Algorithm~\ref{alg:oblivious-coin-toss}
(see Section~\ref{sec:ER-rand}).
Since the number of iterations required by this subroutine is roughly proportional to the number of neighbors,
we choose to allocate a constant number of neighbors in expectation to each block:
with constant probability the block contains some neighbors, and with high probability it has at most $O(\log n)$ neighbors.

As the actual number of neighbors appearing in each block will be different,
we balance out the discrepancies by performing \emph{rejection sampling}.
This equalizes the probability of choosing any neighbor implicitly without knowledge of $\deg(v)$.
Using the fact that the maximum number of neighbors in any block is $\Bo(\log n)$,
we show not only that the probability of success in the rejection sampling process is at least $1/\poly(\log n)$,
but the number of iterations required by \func{Next-Neighbor} is also bounded by $\poly(\log n)$, achieving the overall $\poly(\log n)$ complexities.
Here, we will extensively rely on the assumption that the expected number of neighbors for consecutive vertices,
$\sum_{u=a}^b p_{vu}$, can be approximated efficiently.

\subsubsection{Partitioning and Filling the Blocks}
\label{sec:block_partitioning_and_filling}
We fix a sufficiently large constant $L$, and assign the vertex $u$ to the $\lceil\sum^{u}_{i=1} p_{vi}/L\rceil^\textrm{th}$ block of $v$.
Essentially, each block represents a contiguous range of vertices, where the expected number of neighbors of $v$ in the block is $\approx L$
(for example, in $G(n,p)$, each block contains $\approx L/p$ vertices).
We define $\Gamma^{(i)}(v) = \Gamma(v) \cap B^{(i)}_v$, the neighbors appearing in block $B^{(i)}_v$.
Our construction ensures that $L-1 < \mathbb E \left[|\Gamma^{(i)}(v)|\right] < L+1$ for every $i < |\vec B_v|$
(i.e., the condition holds for all blocks except possibly the last one).

Now, we show that with high probability, all the block sizes $|\Gamma^{(i)}(v)|=\mathcal{O}(\log n)$, and at least a $1/3$-fraction of the blocks are non-empty (i.e., $|\Gamma^{(i)}(v)|>0$), via the following lemmas (proven in Section~\ref{sec:undirected_omitted}).

\begin{restatable}{lemma}{MaxBlockSize}
\label{lem:MaxBlockSize}
With high probability, the number of neighbors in every block, $|\Gamma^{(i)}(v)|$, is at most $ \mathcal{O}(\log n)$.
\end{restatable}

\begin{restatable}{lemma}{EmptyBlock}
\label{lem:EmptyBlock}
With high probability, for every $v$ such that $|\vec B_v| = \Omega(\log n)$ (i.e., $\mathbb E = \Omega(\log n)$), at least a $1/3$-fraction of the blocks $\{B^{(i)}_v\}_{i\in[|\vec B_v|]}$ are non-empty.
\end{restatable}

We consider blocks to be in two possible states -- \filled~or \unfilled. Initially, all blocks are considered \unfilled.
In our algorithm we will maintain, for each block $B^{(i)}_v$, the set $P^{(i)}_v$ of known neighbors of $u$ in block $B^{(i)}_v$;
this is a refinement of the set $P_v$ in Section~\ref{sec:ER-rand}.
We define the behaviors of the procedure $\func{Fill}(v,i)$ as follows.
When invoked on an unfilled block $B^{(i)}_v$, $\func{Fill}(v,i)$ decides whether each vertex $u \in B^{(i)}_v$ is a neighbor of $v$
(implicitly setting $\ADJ[v][u]$ to $\ONE$ or $\ZERO$) unless $X_{vu}$ is already decided; in other words, update $P_v^{(i)}$ to $\Gamma^{(i)}(v)$.
Then $B^{(i)}_v$ is marked as \filled.
We postpone the description of our implementation of $\func{Fill}$ to Section~\ref{sec:fill_implement}, instead using it as a black box.

\subsubsection{Putting it all together: \func{Random-Neighbor} queries}
\label{sec:random_neighbor}
\begin{wrapfigure}[10]{R}{0.47\textwidth}
\vspace{-0.75em}
\begin{framed}
    \renewcommand\figurename{Algorithm}
    \caption{Block sampling.}
    \label{alg:random}
    \begin{algorithmic}
        \Procedure{Random-Neighbor}{$v$}
            \While{True}
                \State{\textbf{sample} $B^{(i)}_v \thicksim_{\mathcal U} \vec B_v$ u.a.r.}
                \If {$B^{(i)}_v$ is not \emph{filled}}
                    \State$\func{Fill}\left( v,i\right)$
                \EndIf
                \State{\textbf{with probability} $\frac{|P_v^{(i)}|}{M}$}
                    \State{\hspace{\algorithmicindent}\Return $u\thicksim_{\mathcal U} P_v^{(i)}$ u.a.r }
            \EndWhile
        \EndProcedure
    \end{algorithmic}
\end{framed}
\end{wrapfigure}

Consider Algorithm~\ref{alg:random} for sampling a random neighbor via rejection sampling.
For simplicity, throughout the analysis, we assume $|\vec B_v| = \Omega(\log n)$;
otherwise, invoke $\func{Fill}(v,i)$ for all $i \in [|\vec B_v|]$ to obtain the entire neighbor list $\Gamma(v)$.

To obtain a random neighbor, we first choose a block $B^{(i)}_v$ uniformly at random, and invoke $\func{Fill}(v,i)$ if the block is \emph{unfilled}.
Then, we \emph{accept} the sampled block for generating our random neighbor with probability proportional to $|P_v^{(i)}|$.
More specifically, if $M = \Theta(\log n)$ is an upper bound on the maximum number of neighbors in any block (see Lemma~\ref{lem:MaxBlockSize}),
we accept block $B^{(i)}_v$ with probability $|P_v^{(i)}|/M$, which is well-defined (i.e., does not exceed $1$) with high probability.
Note that if $P_v^{(i)} = \emptyset$, we sample another block.
If we choose to accept $B^{(i)}_v$, we return a random neighbor from $P_v^{(i)}$.
Otherwise, \emph{reject} this block and repeat the process again.

Since the returned vertex is always a member of $P_v^{(i)}$, a valid neighbor is always returned.
We now show that the algorithm correctly samples a uniformly random neighbor
and bound the number of iterations required for the rejection sampling process.
\begin{restatable}{lemma}{rand-gen-correct}
\label{lem:rand_gen_correct}
Algorithm~\ref{alg:random} returns a uniformly random neighbor of vertex $v$.
\end{restatable}
\begin{proof}
It suffices to show that the probability that any neighbor in $\Gamma(v)$ is return with uniform positive probability, within the same iteration.
Fixing a single iteration and consider a vertex $u\in P_v^{(i)}$, we compute the probability that $u$ is accepted.
The probability that $B^{(i)}_vi$ is chosen is $1/|\vec B_v|$, the probability that $B^{(i)}_v$ is accepted is $|P_v^{(i)}|/M$,
and the probability that $u$ is chosen among $P_v^{(i)}$ is $1/|P_v^{(i)}|$.
Hence, the overall probability of returning $u$ in a single iteration
of the loop is $1/(|\vec B_v|\cdot M)$, which is positive and independent of $u$.
Therefore, each vertex is returned with the same probability.
\end{proof}

\begin{restatable}{lemma}{rand-gen-fast}
\label{lem:rand_gen_fast}
Algorithm~\ref{alg:random} terminates in $\mathcal{O}(\log n)$ iterations in expectation, or $\mathcal{O}(\log^2 n)$ iterations w.h.p.
\end{restatable}
\begin{proof}
Using Lemma~\ref{lem:EmptyBlock}, we conclude that
the probability of choosing a non-empty block is at least $1/3$.
Since $M = \Theta(\log n)$ by Lemma~\ref{lem:MaxBlockSize}, the success probability of each iteration is at least $1/(3M)=\Omega(1/\log n)$,
Thus, the number of iterations required is $O(\log^2 n)$ with high probability.
\end{proof}

\subsection{Implementation of \func{Fill}}
\label{sec:fill_implement}

\begin{wrapfigure}[13]{R}{0.45\textwidth}
\vspace{-3em}
\begin{framed}
    \renewcommand\figurename{Algorithm}
    \caption{Filling a block}
    \label{alg:fill}
    \begin{algorithmic}
    \Procedure{Fill}{$v,i$}
    \State{$(a,b) \gets B^{(i)}_v$}
            \While{$a < b$}
                \State{\textbf{sample} $u\sim\mathsf{F}(v,a,b)$}
                \State{$B_u^{(j)} \gets$ block containing $v$}
                \If{$B_u^{(j)}$ is not \filled}
                    \State{$P_v^{(i)}\gets P_v^{(i)}\cup\{u\}$}
                    \State{$P_u^{(j)}\gets P_u^{(j)}\cup\{v\}$}
                \EndIf
                \State{$a \gets u$}
            \EndWhile
            \State{\textbf{mark} $B_u^{(j)}$ as \filled}
    \EndProcedure
    \end{algorithmic}
\end{framed}
\end{wrapfigure}

Lastly, we describe the implementation of the $\func{Fill}$ procedure, employing the approach of skipping non-neighbors, as developed for Algorithm~\ref{alg:oblivious-coin-toss}. We aim to simulate the following process: perform coin-tosses $C_{vu}$ with probability $p_{vu}$ for every $u \in B^{(i)}_v$ and update $\ADJ[v][u]$'s according to these coin-flips unless they are decided (i.e., $\ADJ[v][u] \neq \PHI$). We directly generate a sequence of $u$'s where the coins $C_{vu} = \ONE$, then add $u$ to $P_v$ and vice versa if $X_{vu}$ has not previously been decided. Thus, once $B^{(i)}_v$ is \filled, we will obtain $P_v^{(i)} = \Gamma^{(i)}(v)$ as desired.

As discussed in Section~\ref{sec:ER-rand}, while we have recorded all occurrences of $\ADJ[v][u]=\ONE$ in $P_v^{(i)}$,
we need an efficient way of checking whether $\ADJ[v][u] = \ZERO$ or $\PHI$. In Algorithm~\ref{alg:oblivious-coin-toss},
$\LAST$ serves this purpose by showing that $\ADJ[v][u]$ for all $u \leq \LAST[v]$ are decided as shown in Lemma~\ref{lem:cond-0}.
Here instead, we maintain a single bit marking whether each block is \filled~or \unfilled:
a \filled~block implies that $\ADJ[v][u]$ for all $u \in B^{(i)}_v$ are decided.
The block structure along with the mark bits, unlike $\LAST$, is capable of handling intermittent ranges of intervals,
which is sufficient for our purpose, as shown in the following lemma.
This yields the implementation of Algorithm~\ref{alg:fill} for the $\func{Fill}$ procedure
fulfilling the requirement previously given in Section~\ref{sec:block_partitioning_and_filling}.

\begin{restatable}{lemma}{condFill}\label{lem:cond-0-fill}
The data structures $P_v^{(i)}$'s and the block marking bits together provide a succinct representation of $\ADJ$ as long as modifications to $\ADJ$ are performed solely by the \func{Fill} operation in Algorithm~\ref{alg:fill}. In particular, let $u \in B^{(i)}_v$ and $v \in B_u^{(j)}$. Then, $\ADJ[v][u]=\ONE$ if and only if $u \in P_v^{(i)}$. Otherwise, $\ADJ[v][u]=\ZERO$ when at least one of $B^{(i)}_v$ or $B_u^{(j)}$ is marked as \filled. In all remaining cases, $\ADJ[v][u]=\PHI$.
\end{restatable}
\begin{proof}
The condition for $\ADJ[v][u]=\ONE$ still holds by construction. Otherwise, observe that $\ADJ[v][u]$ becomes decided precisely during a \func{Fill}$(v,i)$ or a \func{Fill}$(u,j)$ operation, which thereby marks one of the corresponding blocks as \filled.
\end{proof}

Note that $P_v^{(i)}$'s, maintained by our implementation, are initially empty but may not still be empty at the beginning of the \func{Fill} function call. These $P_v^{(i)}$'s are again instantiated and stored in a dictionary once they become non-empty.
Further, observe that the coin-flips are simulated independently of the state of $P_v^{(i)}$, so the number of iterations of Algorithm~\ref{alg:fill} is the same as the number of coins $C_{vu} = \ONE$ which is, in expectation, a constant (namely $\sum_{u\in B^{(i)}_v} \mathbb P[C_{vu}=\ONE] = \sum_{u\in B^{(i)}_v} p_{vu} \leq L+1$). 

By tracking the resource required by Algorithm~\ref{alg:fill} we obtain the following lemma; note that ``additional space'' refers to the enduring memory that the implementation must allocate and keep even after the execution, not its computation memory. The $\log n$ factors in our complexities are required to perform binary-search for the range of $B^{(i)}_v$, or for the value $u$ from the CDF of $\mathsf{F}(u,a,b)$, and to maintain the ordered sets $P_v^{(i)}$ and $P_u^{(j)}$.

\begin{restatable}{lemma}{fill_time}
\label{lem:fill_time}
Each execution of Algorithm~\ref{alg:fill} (the \func{Fill} operation) on an \unfilled~block $B^{(i)}_v$, in expectation:
\begin{itemize}
\item terminates within $\bo(1)$ iterations (of its \textup{\textbf{repeat}} loop);
\item computes $\bo(\log n)$ quantities of $\prod_{u \in [a,b]} (1-p_{vu})$ and $\sum_{u\in[a,b]} p_{vu}$ each;
\item uses additional $\mathcal O(\log n)$ time, $\mathcal O(1)$ random $\log n$-bit words, and $\mathcal O(1)$ additional space.
\end{itemize}
\end{restatable}

Observe that the number of iterations required by Algorithm~\ref{alg:fill} only depends on its random coin-flips and independent of the state of the algorithm.
Combining with Lemma~\ref{lem:rand_gen_fast}, we finally obtain polylogarithimc resource bound for our implementation of $\func{Random-Neighbor}$.

\begin{restatable}{corollary}{random_neighbor_time}
\label{cor:random_neighbor_time}
Each execution of Algorithm~\ref{alg:random} (the \func{Random-Neighbor} query), with high probability,
\begin{itemize}
\item terminates within $\bo(\log^2 n)$ iterations (of its \textup{\textbf{repeat}} loop);
\item computes $\bo(\log^3 n)$ quantities of $\prod_{u \in [a,b]} (1-p_{vu})$ and $\sum_{u\in[a,b]} p_{vu}$ each;
\item uses an additional $\mathcal O(\log^3 n)$ time, $\mathcal O(\log^2 n)$ random words, and $\mathcal O(1)$ additional space.
\end{itemize}
\end{restatable}

\paragraph*{Supporting Other Query Types along with \func{Random-Neighbor}}
\begin{itemize}
\item \func{Vertex-Pair}(u,v): We simply need to make sure that Lemma~\ref{lem:cond-0-fill} holds, so we first apply \func{Fill}$(u,j)$ on block $B_u^{(j)}$ containing $v$ (if needed), then answer accordingly.
\item \func{Next-Neighbor}(v): We maintain $\LAST$, and \func{Fill} repeatedly until we find a neighbor. Recall that by Lemma~\ref{lem:EmptyBlock}, the probability that a particular block is empty is $\le 2/3$ Then with high probability, there exists no $\omega(\log n)$ \emph{consecutive} empty blocks $B^{(i)}_v$'s for any vertex $v$, and thus \func{Next-Neighbor} only invokes up to $\bo(\log n)$ calls to \func{Fill}.
\end{itemize}

We summarize the results so far with through the following theorem.

\UndirectedGrand*

We have also been implicitly assuming perfect-precision arithmetic and we relax this assumption in Section~\ref{sec:remove-perfect}.
In the following Section~\ref{sec:undirected_applications}, we show applications of Theorem~\ref{thm:grand} to the $G(n,p)$ model,
and the Stochastic Block model under random community assignment,
by providing formulas and by constructing data structures for computing the quantities specified in Theorem~\ref{thm:grand}.

\subsection{Applications to Erd\"{o}s-R\'{e}nyi Model and Stochastic Block Model}
\label{sec:undirected_applications}
In this section we demonstrate the application of our techniques to
two well known, and widely studied models of random graphs. That is, as required by Theorem~\ref{thm:grand}, we must provide a method for computing the quantities $\prod_{u=a}^b (1-p_{vu})$ and $\sum_{u=a}^b p_{vu}$ of the desired random graph families in logarithmic time, space and random bits.
Our first implementation focuses on the well known Erd\"{o}s-R\'{e}nyi model -- $G(n,p)$: in this case, $p_{vu} = p$ is uniform and our quantities admit closed-form formulas.

Next, we focus on the Stochastic Block model with randomly assigned communities.
Our implementation assigns each vertex to a community in $\{C_1, \ldots, C_r\}$ identically and independently at random,
according to some given distribution $\mathsf{R}$ over the communities.
We formulate a method of sampling community assignments locally.
This essentially allows us to sample from the \emph{multivariate hypergeometric distribution},
using $\tilde{\mathcal O}(1)$ random bits, which may be of independent interest.
We remark that, as our first step, we sample the number of vertices of each community.
That is, our construction can alternatively support a community assignment where the number of vertices of each community is given,
under the assumption that the \emph{partition} of the vertex set into communities is chosen uniformly at random.

\subsubsection{Erd\"{o}s-R\'{e}nyi Model}
\label{sec:app_er}
As $p_{vu} = p$ for all edges $\{u,v\}$ in the Erd\"{o}s-R\'{e}nyi $G(n,p)$ model, we have the closed-form formulas $\prod_{u=a}^b (1-p_{vu}) = (1-p)^{b-a+1}$ and $\sum_{u=a}^b p_{vu} = (b-a+1)p$, which can be computed in constant time according to our assumption, yielding the following corollary.

\ERGrand*

We remark that there exists an alternative approach that picks $F\sim\mathsf{F}(v,a,b)$ directly via a closed-form formula $a+\lceil\frac{\log U}{\log (1-p)}\rceil$ where $U$ is drawn uniformly from $[0,1)$, rather than binary-searching for $U$ in its CDF. Such an approach may save some $\poly(\log n)$ factors in the resources, given the prefect-precision arithmetic assumption. This usage of the $\log$ function requires $\Omega(n)$-bit precision, which is not applicable to our computation model.

While we are able to generate our random graph on-the-fly supporting all three types of queries, our construction still only requires $\bo(m+n)$ space ($\log n$-bit words) in total at any state; that is, we keep $\bo(n)$ words for $\LAST$, $\bo(1)$ words per neighbor in $P_v$'s, and one marking bit for each block (where there can be up to $m+n$ blocks in total). Hence, our memory usage is nearly optimal for the $G(n,p)$ model:

\EROptimal*

\subsubsection{Stochastic Block model}
\label{sec:application_sbm}
In the Stochastic Block model, each vertex is assigned to some community $C_i$, $i \in [r]$.
By partitioning the product by communities, we may rewrite the desired formulas, for $v \in C_i$,
as $\prod_{u=a}^b (1-p_{vu}) = \prod_{j=1}^r (1-p_{ij})^{|[a,b]\cap C_j|}$ and $\sum_{u=a}^b p_{vu}=\sum_{j=1}^r |[a,b]\cap C_j|\cdot p_{ij}$.
Thus, it suffices to design a data structure that is able to efficiently count the number of occurrences of vertices of each community
in any contiguous range (namely the value $|[a,b]\cap C_j|$ for each $j \in [r]$),
where the vertices are assigned communities according to a given distribution $\mathsf{R}$.
To this end, we use the following lemma,
yielding an implementation for the Stochastic Block model using $O(r\, \poly(\log n))$ resources per query.

\begin{restatable}{theorem}{res:sbm-data}\label{thm:sbm-data}
There exists a data structure that samples a community for each vertex independently at random from $\mathsf{R}$
with $\frac{1}{\poly(n)}$ error in the $L_1$-distance, and supports queries that ask for the number of occurrences of vertices of each community
in any contiguous range, using $O(r\,\poly(\log n))$ time, random bits, and additional space per query.
Further, this data structure may be implemented in such a way that requires no overhead for initialization.
\end{restatable}
\SBMGrand*

We provide the full details of the construction in Section~\ref{sec:multivariate_hypergeometric_sampling}.
Our construction extends a similar implementation in the work of \cite{huge} which only supports $r = 2$.
The overall data structure is a balanced binary tree, where the root corresponds to the entire range of indices $[n]$,
and the children of each vertex correspond to the first and second half of the parent's range.
Each node\footnote{For clarity, ``vertex'' is only used in the sampled graph, and ``node'' is only used in the internal data structures.}
holds the number of vertices of each community in its range.
The tree initially contains only the root, with the number of vertices of each community $\langle |C_1|, |C_2|,\cdots, |C_r| \rangle$
sampled according to the multinomial distribution (for $n$ samples from the distribution $\mathsf{R}$).
The children are generated top-down on an as-needed basis according to the given queries.
The technical difficulties arise when generating the children,
where one needs to sample the counts assigned to either child from the correct marginal distribution.
We show how to sample such a count from the \emph{multivariate hypergeometric distribution},
below in Theorem~\ref{thm:SamplingManyColors} (proven in Section~\ref{sec:multivariate_hypergeometric_sampling}).

\begin{restatable}{theorem}{SamplingManyColors}
\label{thm:SamplingManyColors}
Given $B$ marbles of $r$ different colors, such that there are $C_i$ marbles of color $i$,
there exists an algorithm that samples $\langle s_1, s_2,\cdots, s_r \rangle$,
the number of marbles of each color appearing when drawing $l$ marbles from the urn without replacement,
in $O(r\cdot\poly(\log B))$ time and random words.
\end{restatable}
\begin{proof}[Proof of Theorem~\ref{thm:sbm-data}]
Recall that $\mathsf{R}$ denotes the given distribution over integers $[r]$ (namely, the random distribution of communities for each vertex).
Our algorithm generates and maintains random variables $X_1, \ldots, X_n$ (denoting the community assignment),
each of which is drawn independently from $\mathsf{R}$.
Given a pair $(a, b)$, it uses Theorem~\ref{thm:sbm-data} to sample the vector $\vec{C}(a, b) = \langle c_1, \ldots, c_r \rangle$,
where $c_k$ counts the number of variables in $\{X_a, \ldots, X_b\}$ that take on the value $k$.

We maintain a complete binary tree whose leaves corresponds to indices from $[n]$.
Each node represents a range and stores the vector $\vec{C}$ for the corresponding range.
The root represents the entire range $[n]$, which is then halved in each level.
Initially the root samples $\vec{C}(1, n)$ from the multinomial distribution according to $\mathsf{R}$
(see e.g., Section 3.4.1 of \cite{knuth}).
Then, the children are generated on-the-fly as described above.
Thus, each query can be processed within $O(r\,\poly(\log n))$ time, yielding Theorem~\ref{thm:sbm-data}.
\end{proof}

Then, by embedding the information stored by the data structure into the state (as in the proof of Lemma~\ref{lemma:transition}),
we obtain the desired Corollary~\ref{cor:sbm-construct}.

\section{Implementing Random Catalan Objects}%
\label{sec:catalan_objects}

In the previous Section~\ref{sec:application_sbm} on the Stochastic Block Model, we considered random sequences of colored marbles.
Next, we focus on an important variant of these sequences as Catalan objects, which impose a global constraint on the types of allowable sequences.
Specifically, consider a sequence of $n$ white and $n$ black marbles,
such that every \emph{prefix} of the sequence has at least as many white marbles as black ones.
Our goal will be to support queries to a uniformly random instance of such an object.

\begin{figure}[htbp]
    \centering
    \includegraphics[width=\textwidth]{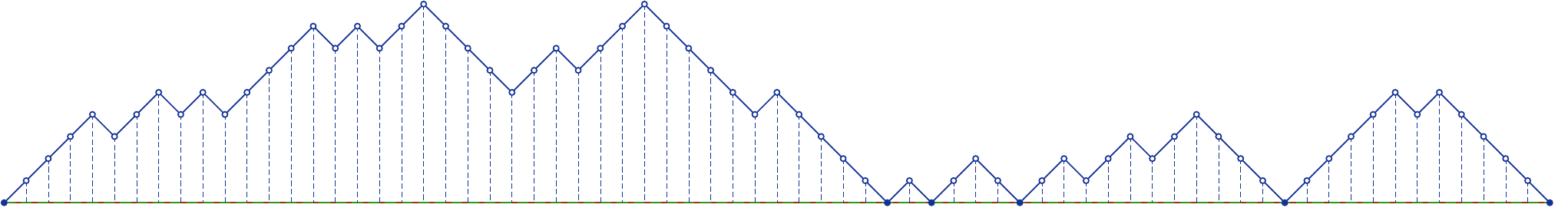}
    \caption{Simple Dyck path with $n = 35$ up and down steps.}
    \label{fig:basic_dyck}
\end{figure}
One interpretation of Catalan objects is given by Dyck paths (Figure~\ref{fig:basic_dyck}).
A Dyck path is essentially a $2n$ step \emph{balanced} one-dimensional walk with exactly $n$ up and down steps.
In Figure~\ref{fig:basic_dyck}, each step moves one unit along the positive $x$-axis (time) and one unit up or down the positive $y$-axis (position).
The prefix constraint implies that the $y$-coordinate of any point on the walk is $\ge 0$ i.e. the walk never crosses the $x$-axis.
The number of possible Dyck paths (see Theorem~\ref{thm:number_of_dyck_paths}) is the $n^{th}$ Catalan number $C_n=\frac{1}{n+1}\cdot{2n\choose n}$.
Many important combinatorial objects occur in Catalan families of which these are an example.

We will approach the problem of partially sampling Catalan objects through Dyck paths.
This, in turn, will allow us to implement access other random Catalan objects such as rooted trees, and bracketed expressions.
Specifically, we will want to answer the following queries:
\begin{itemize}
    \item \func{Direction}$(t)$: Returns the value of the $t^{th}$ step in the Dyck path (whether the step is up or down).
    \item \func{Height}$(t)$: Returns the $y$-position of the path after $t$ steps
    (the number of up steps minus the number of down steps among the first $t$ steps).
    Since a $\func{Direction}(t)$ query can be simulated using the queries $\func{Height}(t)$ and $\func{Height}(t-1)$,
    we will not explicitly discuss the \func{Direction} queries in what follows.
    \item \func{First-Return}$(t)$: If the $(t+1)^{th}$ step is upwards i.e. $\func{Height}(t+1) = \func{Height}(t)+1$,
    it returns the smallest index $t'>t$ such that $\func{Height}(t')=\func{Height}(t)$.
    While it may not be clear why this query is important, it will be useful for querying bracketed expressions and random trees.
    (see Section~\ref{sec:bijections_to_other_catalan_objects}).
\end{itemize}

\subsection{Bijections to other Catalan objects}%
The $\func{Height}$ query is natural for Dyck paths, but the $\func{First-Return}$ query is important in exploring other Catalan objects.
For instance, consider a random well bracketed expression; equivalently an uniform distribution over the Dyck language.
One can construct a trivial bijection between Dyck paths and words in this language
by replacing up and down steps with opening and closing brackets respectively (Figure~\ref{fig:dyck_bijection}).
The $\func{Height}$ query corresponds to asking for the nesting depth at a certain position in the word,
and $\func{First-Return}(x)$ returns the position after the matched closing bracket for the step $(x\rightarrow x+1)$.

\label{sec:bijections_to_other_catalan_objects}
\begin{figure}[htbp]
    \centering
    \includegraphics[width=\textwidth]{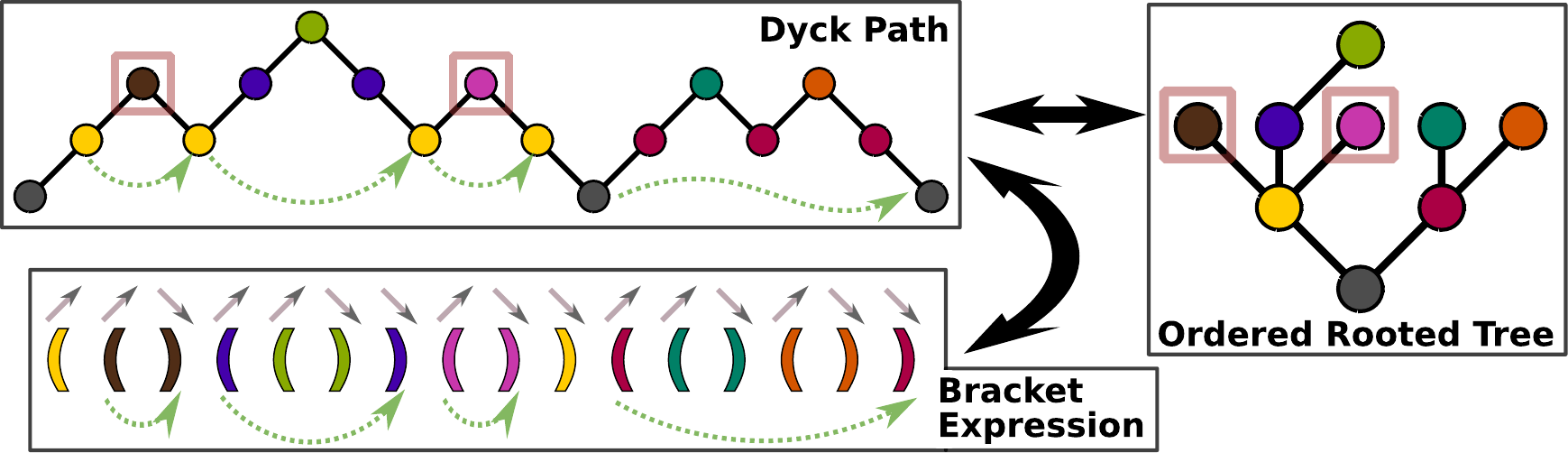}
    \caption{Bijections from Dyck paths to bracketed expressions and ordered rooted trees.
    Note that the color coded \emph{positions} on the path correspond to tree nodes, and the path itself is a DFS traversal of the tree.
    The bijection to bracketed expressions proceeds by directly replacing up and down steps with opening and closing brackets respectively
    (we color the brackets with the \emph{higher} endpoint of the corresponding step).
    Green dashed arrows show successive $\func{First-Return}$ queries on the path.
    Using the bijection, this is equivalent to revealing three child sub-trees in order from left to right for the corresponding tree node,
    and also to finding matching brackets in bracketed expressions.
    \func{First-Return} is not defined for the red boxed nodes in either forward or reverse direction,
    since this position corresponds to a leaf node in the tree, or to a terminal nesting level in the bracket expression.}
    \label{fig:dyck_bijection}
\end{figure}
There is also a natural bijection between Dyck paths and ordered rooted trees (Figure~\ref{fig:dyck_bijection}),
by viewing the Dyck path as a transcript of the tree's DFS traversal.
Starting with the root, for each ``up-step'' we move to a new child of the current node, and for each ``down-step'', we backtrack towards the root.
Thus, the $\func{Height}$ query returns the depth of a node.
Also, since the Dyck path is a DFS transcript of the tree, a $\func{First-Return}$ query on the path
can be used to find successive children of a tree node (each return produces the \emph{next child}).
For instance, in Figure~\ref{fig:dyck_bijection}, we can invoke \func{First-Return} thrice
starting at the first \emph{yellow path position} to reveal the corresponding three children of the \emph{yellow tree node}.

By definition, \func{First-Return}$(x)$ is meaningful only when the step from $x$ to $x+1$ is upwards,
i.e. when $\func{Height}(x+1) = \func{Height}(x) + 1$.
We can also implement a \func{Reverse-First-Return} query, which is just a standard \func{First-Return} query on the reversed Dyck path
(consider a reversal of the green dashed arrows in Figure~\ref{fig:dyck_bijection}).
The reversal implies that \func{Reverse-First-Return}$(x)$ is only meaningful when $\func{Height}(x-1) = \func{Height}(x) + 1$.
In terms of bijections, \func{Reverse-First-Return} is equivalent to finding a matching opening bracket in bracketed expressions,
and a \func{Previous-Child} query in rooted trees.
We show how to implement this query in Section~\ref{sec:reverse_first_return_queries}.
In the case where the height at $x$ is larger than both the heights at $x-1$ and $x+1$ (boxed nodes in Figure~\ref{fig:dyck_bijection}),
there is no meaningful ``first return'' from the context of the bijections.
Specifically, these nodes correspond to leaf nodes in rooted trees, or to a terminal nesting level in the bracket expression.

Moving forwards, we will focus on Dyck paths for the sake of simplicity.

\subsection{Catalan Trapezoids and Generalized Dyck Paths}
In order to implement local access to random Dyck paths, we will need to analyze more general Catalan objects.
Specifically, we consider a sequence of $U$ up-steps and $D$ down-steps,
such that any prefix of the sequence containing $U'$ up and $D'$ down steps satisfies $U'-D' \ge 1-k$.
This means that we start our Dyck path at a height of $k-1$, and we are never allowed to cross below zero (Figure~\ref{fig:complex_dyck}).
Note that the case $k=1$ corresponds to the standard description of Dyck paths, as mentioned previously (Figure~\ref{fig:basic_dyck}).
\begin{figure}[htbp]
    \centering
    \includegraphics[width=\textwidth]{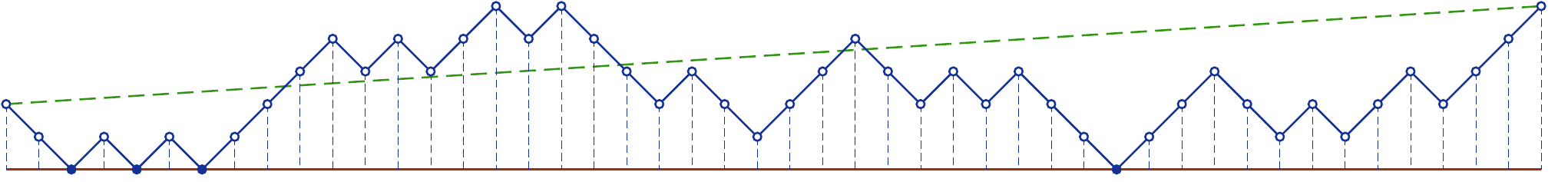}
    \caption{Generalized Dyck path with $U = 25$, $D = 22$ and $k = 3$.
             Note that the boundary is $k-1 = 2$ units below the starting height.} \label{fig:complex_dyck}
\end{figure}

We will denote the set of such \emph{generalized Dyck paths} as $\mathbb C_k(U,D)$ and the number of paths as $C_k(U,D) = |\mathbb C_k(U,D)|$,
which is an entry in the \textit{Catalan Trapezoid} of order $k$ \cite{trap}.
We also use $\mathsf C_k(U,D)$ to denote the uniform distribution over $\mathbb C_k(n,m)$.
Now, we state a result from \cite{trap} without proof:
\begin{align}
    \label{eq:catalan_trapezoid}
    C_k(U,D)=
    \begin{cases}
    \binom{U+D}{D} &0\le D<k\\
    \binom{U+D}{D} - \binom{U+D}{D-k} &k\le D\le U+k-1\\
    0 &D>U+k-1
    \end{cases}
\end{align}
For $k = 1$ and $n=m$, these represent the vanilla Catalan numbers i.e. $C_n = C_1(n,n)$ (number of simple Dyck paths).
Our goal is to sample from the distribution $\mathsf C_1(n,n)$.

Consider the situation after a sequence of \func{Height} queries to the Dyck path at various locations $\langle x_1, x_2,\cdots, x_m \rangle$,
such that the corresponding heights were sampled to be $ \langle y_1, y_2,\cdots, y_m \rangle$.
These revealed locations partition the path into disjoint \emph{intervals} $[x_i,x_{i+1}]$,
where the heights of the endpoints of each interval have been determined (as $y_i = \func{Height}(x_i)$).
We notice that these intervals can be generated independently of each other.
Specifically, the path within the interval $[x_i, x_{i+1}]$ will be sampled from $\mathsf C_k(U,D)$,
where $k - 1 = y_i$, $U + D = x_{i+1} - x_i$, and $U-D = y_{i+1} - y_i$.
Moreover, since the heights of the endpoints $y_i$ and $y_{i+1}$ are known, this choice is independent of any samples outside the interval.
Next, in Section~\ref{sec:implementing_height_queries}, we will show how one can determine heights within such an interval,
and in Section~\ref{sec:supporting_first_return_queries} we will move on to the more complicated $\func{First-Return}$ queries.


\subsection{Implementing \func{Height} queries}
\label{sec:implementing_height_queries}
We implement $\func{Height}(t)$ by showing how to efficiently determine the height of the path
at the midpoint of an existing interval $[x_i, x_{i+1}]$ (with corresponding endpoint heights $y_i, y_{i+1}$),
which results in two sub-intervals that are half the size.
Next, we extend this strategy to determine the heights of arbitrary positions by recursively sub-dividing the relevant interval (binary search).
If the interval in question has odd length, we sample a single step from an endpoint, and proceed with a shortened even length interval.
Sampling a single step is easy since there are only two outcomes (see proof of Theorem~\ref{thm:dyck_midpoint_sampling}).
\begin{figure}[htpb]
    \centering
    \includegraphics[width=\textwidth]{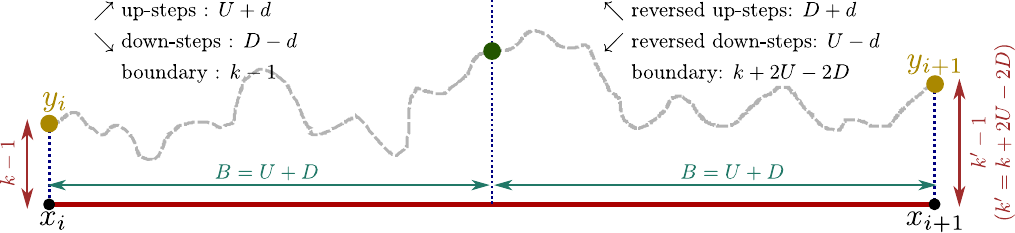}
    \caption{The $2B$-interval is split into two equal parts resulting in two separate Dyck problems.
             The green node (center) is the sampled height of the midpoint parameterized by the value of $d$.
             The path considered in both sub-intervals starts at a yellow node (left and right edges) and ends at the green node.
             From this perspective, the path on the right is reversed with up and down steps being swapped.
             A possible path is shown in gray.}
    \label{fig:dyck_height_sampling}
\end{figure}

Our general recursive step is as follows.
We consider an interval of length $2B$ comprising of $2U$ up-steps and $2D$ down-steps where the sum of any prefix cannot be less than $k-1$
i.e. the path within this interval should be sampled from $\mathsf C_k(2U,2D)$.
In order to make the analysis simpler, we have assumed that the number of up and down steps are both even.
The case of sampling according to $\mathsf C_k(2U+1, 2D+1)$ works similarly with slightly different formulae.
Without loss of generality, we assume that $D\le U$; if this were not the case, we could simply flip the interval,
swap the up and down steps, and modify the prefix constraint to $k'=k+2U-2D$ (Figure~\ref{fig:dyck_height_sampling}).
This ensures that the overall path in the interval is non-decreasing in height, which will simpify our analysis.

We determine the height of the path $B = U+D$ steps into the interval at the midpoint (Figure~\ref{fig:dyck_height_sampling}).
This is equivalent to finding the number of up/down steps that get assigned to the first half of the interval.
We parameterize the possibilities by $d$ and define $p_d$ to be the probability that exactly $U+d$ up-steps and $D-d$ down steps
get assigned to the first half (with the remaining $U-d$ up steps and $D+d$ down steps being assigned to the second half).
\begin{align}
\label{eq:height_sampling_probability}
p_d = \frac{S_{left}(d)\cdot S_{right}(d)}{S_{total}(d)}
\end{align}
Here, $S_{left}(d)$ denotes the number of possible paths in the first half (using $U+d$ up steps)
and $S_{right}(d)$ denotes the number of possible paths in the second half (using $U-d$ up steps).
Note that all of these paths have to respect the $k$-boundary constraint (cannot dip more than $k-1$ units below the starting height), where $k=y_i+1$.
Moving forwards, we will drop the $d$ when referring to the path counts.
We (conceptually) flip the second half of the interval,
such that the corresponding path begins from the end of the $2B$-interval and terminates at the midpoint (Figure~\ref{fig:dyck_height_sampling}).
This results in a different starting point, and the prefix/boundary constraint will also be different.
Hence, we define $k' = k + 2U - 2D$  to represent the new boundary constraint (since the final height of the $2B$-interval is $k'-1$).
Finally, $S_{total}$ is the total number of possible paths in the $2B$ interval.

We cannot directly sample from this complicated distribution $\{ p_d\}$.
Instead, we use the rejection sampling strategy from Lemma~\ref{lem:rejection_sampling}.
An important point to note is that in order to apply this lemma, we must be able to approximate the $p_d$ values.
However, we cannot naively use the formula from Equation~\ref{eq:height_sampling_probability},
since the values of $\{ S_{left},S_{right},S_{total}\}$ are too large to compute explicitly.
Lemma~\ref{lem:probability_approximation_oracle} in Section~\ref{sec:computing_probabilities}
shows how to indirectly compute the probabilty approximations.
We also use the following lemma to bound the deviation of the path with high probability.
A proof is presented in Section~\ref{sec:dyck_path_boundaries_and_deviations}.
\begin{restatable}{lemma}{DyckPathDeviationBound}
\label{lem:DyckPathDeviationBound}
Consider a contiguous \emph{sub-path} of a simple Dyck path of length $2n$
where the sub-path is of length $2B$ comprising of $U$ up-steps and $D$ down-steps (with $U + D = 2B$).
Then there exists a constant $c$ such that the quantities $|B-U|$, $|B-D|$, and $|U-D|$
are all $<c\sqrt{B\log n}$ with probability at least $1-1/n^2$ for every possible sub-path.
\end{restatable}

This lemma allows us to ignore potential midpoint heights that cause a deviation greater than $c \sqrt{B\log n}$.
A direct implication is that with high probability, the correctly sampled value for $d$ will be $\mathcal O(\sqrt{B\log n})$.
In other words, the height of the midpoint takes on one of only  $\mathcal O(\sqrt{B\log n})$ distinct values with high probability.
This immediately suggests a $\widetilde{\mathcal O}(\sqrt{B})$ time algorithm for determining the midpoint height,
by explicitly computing the probabilities of each of these potential heights, and directly sampling from the resulting distribution.
However, we can go further and obtain a $\mathcal O(poly(\log n))$ time algorithm.

\subsubsection{The Simple Case: Far Boundary}%
\label{sec:the_simple_case}
We first consider the case when the boundary constraint is far away from the starting point, i.e. $k$ is large.
The following lemma (proof in Section~\ref{sec:dyck_path_boundaries_and_deviations}) shows that in this case,
we can safely \emph{ignore} the constraint.  Intuitively, this is because the boundary is so far away,
that with high probability, we do not hit it even if we choose a random \emph{unconstrained} path.
\begin{restatable}{lemma}{DyckPathIrrelevantBoundary}
\label{lem:dyck_path_irrelevant_boundary}
Given a Dyck path sampling problem of length $B$ with $U$ up, $D$ down steps, and a boundary at $k$,
there exists a constant $c$ such that if $k > c \sqrt{B\log n}$, then the distribution of paths sampled without a boundary $\mathsf C_{\infty}(U,D)$
is $\mathcal O(1/n^2)$-close in $L_1$ distance to the distribution of Dyck paths $\mathsf C_k(U+D)$.
\end{restatable}
By Lemma~\ref{lem:dyck_path_irrelevant_boundary}, the problem of sampling from $C_k(2U,2D)$
reduces to sampling from the hypergeometric distribution $C_{\infty}(2U,2D)$ when $k>\mathcal{O}(\sqrt{B\log n})$
i.e. the probabilities $p_d$ can be approximated by:
\[
q_d = \frac{{{B}\choose{D-d}}\cdot{{B}\choose{D+d}}}{{{2B}\choose{2D}}}
\]
This problem of sampling from the hypergeometric distribution is implemented using $\mathcal O(poly(\log n))$ resources in \cite{huge}
(see Lemma~\ref{lem:ggn_interval_summable} in Section~\ref{sec:multivariate_hypergeometric_sampling}).
We also used this result earlier in the paper in order to find the community assignments in the Stochastic Block Model
(Section~\ref{sec:application_sbm}).

\subsubsection{The Difficult Case: Intervals Close to Zero}
\label{sec:the_difficult_case}
The difficult case is when $k = \mathcal{O}(\sqrt{B\log n})$,
and the previous approximation due to Lemma~\ref{lem:dyck_path_irrelevant_boundary} no longer works.
In this case, we cannot just ignore the boundary constraint, and instead we have to analyze the true probability distribution given by $p_d$.
We obtain an expression for $p_d$ by substituting the formula for generalized Catalan numbers as follows:
(Equation~\ref{eq:catalan_trapezoid}) into Equation~\ref{eq:height_sampling_probability}.
\begin{align}
    S_{left} = C_k(U+d,D-d)
    &&S_{right} = C_{k'}(U-d,D+d)
    &&S_{total} = C_k(2U,2D)
\end{align}
Since the right interval is flipped in our analysis, this changes the prefix/boundary constraint,
and hence, the expression for $S_{right}$ uses $k' = k+2U-2D$.
This also implies that $k' = \Bo(\sqrt{B\log n})$ (using Lemma~\ref{lem:DyckPathDeviationBound}).
We can now use Equation~\ref{eq:height_sampling_probability} to evaluate the probabilties $p_d = S_{left}\cdot S_{right}/S_{total}$.
Recall that $S_{left}$ and $S_{right}$ are the number of possible paths in the left and right half of the interval,
when exactly $U+d$ up steps are assigned to the first half, and $S_{total}$ is the total number of possible paths in the interval.

We will invoke the rejection sampling technique (Lemma~\ref{lem:rejection_sampling}), by constructing a different distribution $q_d$
that approximates $p_d$ up to logarithmic factors over the vast majority of its support
(we ignore all $|d|>\Theta(\sqrt{B\log n})$ since the associated probability mass is negligible by Lemma~\ref{lem:DyckPathDeviationBound}).
In order to perform rejection sampling, we also need good approximations of $p_d$,
which is acheved by Lemma~\ref{lem:probability_approximation_oracle} in Section~\ref{sec:computing_probabilities}.
Next, we define an appropriate $q_d$ that approximates $p_d$ and also has an \emph{efficiently computable} $CDF$.
Surprisingly, as in Section~\ref{sec:the_simple_case}, we will be able to use the hypergeometric distribution for $q_d$,
\[
q_d \equiv \frac{{B\choose D-d}\cdot{B\choose D+d}}{{2B\choose 2D}} = \frac{{B\choose D-d}\cdot{B\choose U-d}}{{2B\choose 2D}}
\]
However, the argument for why this $q_d$ is a good approximation to $p_d$ is far less straightforward.

First, we consider the case where $k\cdot k'\le 2U+1$.
In this case, we use loose bounds for $S_{left} < \binom{B}{D-d}$ and $S_{right} < \binom{B}{U-d}$.
These are true because $\binom{B}{D-d}$ and $\binom{B}{U-d}$ are the total number of \emph{unconstrained} paths
in the left and right half respectively, and adding the boundary constraint can only reduce the number of paths.
We also prove the following lemma in Section~\ref{sec:appendix_implementing_height_queries} to bound the value of $S_{total}$.
\begin{restatable}{lemma}{DTotalFarBoundary}
\label{lem:DTotalFarBoundary}
When $kk' > 2U + 1$, $S_{total} > \frac 12\cdot \binom{2B}{2D}$.
\end{restatable}

Combining the three bounds, we conclude that $p_d < \frac 12 q_d$.
Intuitively, in this case the Dyck boundary is far away, and therefore the number of possible paths
is only a constant factor away from the number of unconstrained paths (see Section~\ref{sec:the_simple_case}).
The case where the boundaries are closer (i.e. $k\cdot k' \le 2U+1$) is trickier,
since the individual counts need not be close to the corresponding binomial counts.
However, in this case we can still ensure that the sampling probability is within poly-logarithmic factors of the binomial sampling probability.
We use the following lemmas to bound $S_{left}$ and $S_{right}$ (proofs in Section~\ref{sec:appendix_implementing_height_queries}).
\begin{restatable}{lemma}{DLeftBound}
\label{lem:DLeftBound}
$S_{left} \le c_1 \frac{ k\cdot\sqrt{\log n}}{\sqrt{B}}\cdot{{B}\choose{D-d}}$ for some constant $c_1$.
\end{restatable}
\begin{restatable}{lemma}{DRightBound}
\label{lem:DRightBound}
$S_{right} \le c_2 \frac{k'\cdot \sqrt{log n}}{\sqrt{B}}\cdot{{B}\choose{U-d}}$ for some constant $c_2$.
\end{restatable}

Finally, we obtain a diferent bound on $S_{total}$ for the \emph{near boundary} case.
\begin{restatable}{lemma}{DTotalNearBoundary}
\label{lem:DTotalNearBoundary}
When $kk' \le 2U + 1$, $S_{total} \ge c_3 \frac{k\cdot k'}{B}\cdot{{2B}\choose{2D}}$ for some constant $c_3$.
\end{restatable}

Multiplying these inequalities together allows us to bound $p_d = S_{left}\cdot S_{right}/S_{total} \le \Theta(q_d\log n)$,
implying $p_d/q_d \le \Theta(\log n)$.
Consequently, we can leverage Lemma~\ref{lem:rejection_sampling} to obtain a sample from $\{ p_d\}$
using $\mathcal O(\log^{\mathcal O(1)}n)$ samples from $\{ q_d\}$,
thus determining the height of the Dyck path at the midpoint of the interval.

\begin{theorem}
\label{thm:dyck_midpoint_sampling}
Given an interval $[x_i,x_{i+1}]$ (and the associated endpoint heights $y_i, y_{i+1}$) in a Dyck path of length $2n$,
with the guarantee that no position between $x_i$ and $x_{i+1}$ has been sampled yet,
there is an algorithm that returns the height of the path at the midpoint of $x_i$ and $x_{i+1}$ (or next to the midpoint if $x_{i+1}-x_i$ is odd).
Moreover, this algorithm only uses $\mathcal O(poly(\log n))$ resources.
\end{theorem}
\begin{proof}
If $x_{i+1}-x_i$ is even, we can set $B = (x_{i+1}-x_i)/2$.
Otherwise, we first sample a single step from $x_i$ to $x_i+1$, and then set $B = (x_{i+1}-x_i-1)/2$.
Since there are only two possibilities for a single step, we can explicitly approximate the corresponding probabilities
using the strategy outlined in the proof of Lemma~\ref{lem:probability_approximation_oracle} (see Section~\ref{sec:computing_probabilities}),
and then sample accordingly.
This allows us to apply the rejection sampling from Lemma~\ref{lem:rejection_sampling}
using $\{ q_d\}$ to obtain samples from $\{ p_d\}$ as defined above.
\end{proof}

\begin{theorem}
\label{thm:dyck_height_sampling}
There is an algorithm that provides sample access to a Dyck path of length $2n$,
by answering queries of the form \func{Height}$(x)$ with the correctly sampled height of the Dyck path at position $x$
using only $\mathcal O(poly(\log n))$ resources per query.
\end{theorem}
\begin{proof}
The algorithm maintains a successor-predecessor data structure (e.g. Van Emde Boas tree) to store all positions $x$ that have already been determined.
Each newly determined position is added to this structure.
Given a query \func{Height}$(x)$, the algorithm first finds the successor and predecessor (say $a$ and $b$) of $x$ among the already queried positions.
This provides us the guarantee required to apply Theorem~\ref{thm:dyck_midpoint_sampling},
which allows us to query the height at the midpoint of $a$ and $b$.
We then binary search by updating either the successor or predecessor of $x$ and repeat until we determine the height of position $x$.
\end{proof}

\subsection{Supporting ``First Return'' Queries}%
\label{sec:supporting_first_return_queries}
In this section, we discuss more complex queries to a Dyck path.
Specifically, we implement $\func{First-Return}(x)$ to allow the user to query the next time the path returns to \func{Height}$(x)$.
The utility of this kind of query is illustrated through bijections to other random Catalan objects
in Section~\ref{sec:bijections_to_other_catalan_objects}.
An important detail here is that if the first step from $x$ to $x+1$ is a down-step, there is no well defined \func{First-Return}.



\subsubsection{Maintaining a Boundary Invariant}
\label{sec:maintaining_a_boundary_invariant}
Notice that after performing a set of $\func{Height}$ queries $\langle x_1, x_2,\cdots, x_m \rangle$ to the Dyck path,
many different positions are revealed (possibly in adversarial locations).
This partitions the path into at most $m+1$ disjoint and independent \emph{intervals} with known boundary conditions.
The first step towards finding the $\func{First-Return}$ from position $t$ would be to locate the \emph{interval} where the first return occurs.
Even if we had an efficient technique to filter intervals, we would want to avoid considering all $\Theta(m)$ intervals to find the correct one.
In addition, the fact that a specific interval \emph{does not} contain the first return implies dependencies for all subsequent samples.

We resolve these difficulties by maintaining a new invariant.
Consider all positions whose heights have have already been determined, either directly as a result of user given $\func{Height}$ queries,
or indirectly due to recursive $\func{Height}$ calls; $\langle x_1, x_2,\cdots, x_m \rangle$ in increasing order i.e. $x_i<x_{i+1}$.
Let the corresponding heights be $ \langle y_1, y_2,\cdots, y_m \rangle$ i.e. $\func{Height}(x_i) = y_i$.
\begin{invariant}
\label{inv:boundary_invariant}
For any interval $[x_i,x_{i+1}]$ where the heights  of the endpoints have been determined to be $y_i$ and $y_{i+1}$,
and every other height in the interval has yet to determined,
the section of the Dyck path between positions $x_i$ and $x_{i+1}$ is constrained to lie above $min(y_i, y_{i+1})$.
\end{invariant}
\begin{figure}[htpb]
    \centering
    \includegraphics[width=\textwidth]{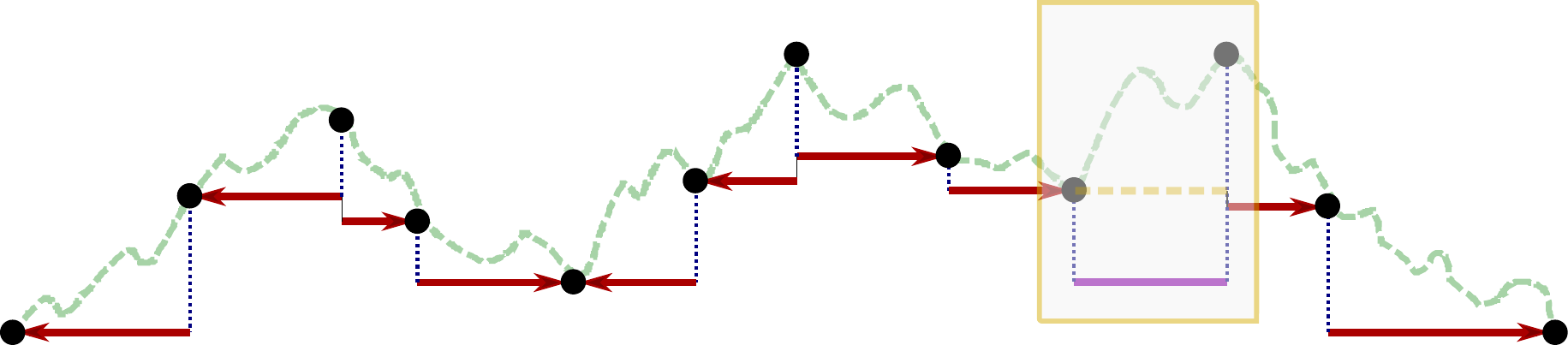}
    \caption{The set of intervals formed by a set of height samples.
        Each interval also has its own boundary constraint (red).
        Invariant~\ref{inv:boundary_invariant} implies that each boundary must coincide with one of the interval endpoints.
        Note that the only interval which violates the invariant is the third last one (shown in yellow box).}
    \label{fig:dyck_boundary_invariant}
\end{figure}

How can one maintain such an invariant?
After determining the height of a particular position $x_i$ as $y_i$ (with $x_{i-1} < x_i < x_{i+1}$),
the invariant is potentially broken on either side of $x_i$.
We re-establish the invariant by determining the height of an additional point on either side
(see Section~\ref{sec:maintaining_height_queries_under_invariant} for details).
Specifically, we use the following \emph{two step strategy} to find the additional point for some violating interval $[x_i, x_{i+1}]$
(example violation in Figure~\ref{fig:dyck_boundary_invariant}):
\begin{enumerate}
    \item Sample the lowest height $h$ achieved by the walk between $x_i$ and $x_{i+1}$ according to
    the uniform distribution over all possible paths that respect the current boundary constraint
    (see Section~\ref{sec:sampling_the_lowest_achievable_height}).
    \item Find an intermediate position $x$ such that $x_i < x < x_{i+1}$ and $\func{Height}(x) = h$
    (see Section~\ref{sec:sampling_first_position_touching_mandatory_boundary}).
\end{enumerate}
The newly determined position $x$ produces two sub-intervals $[x_i, x]$ and $[x, x_{i+1}]$.
Since $h = \func{Height}(x)$ has been determined to be the minimum height in the overall range $[x_i, x_{i+1}]$,
the invariant is preserved in the new intervals (see Figure~\ref{fig:dyck_invariant_preserve}).
Lemma~\ref{lem:first_return_interval} in Section~\ref{sec:finding_the_correct_interval}
shows how this invariant can help us efficiently search for the interval containing the first return.
\begin{figure}[htpb]
    \centering
    \includegraphics[width=\textwidth]{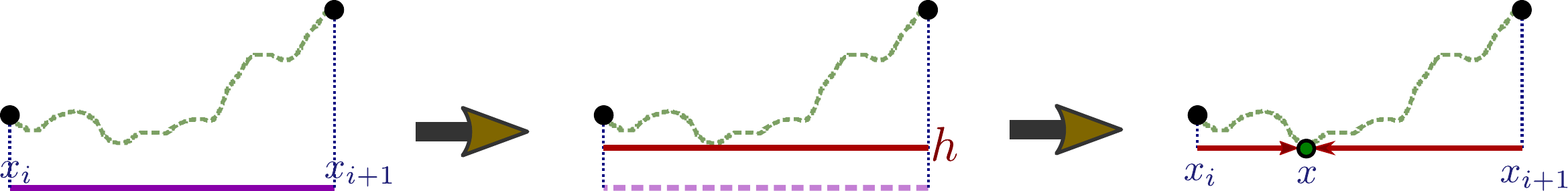}
    \caption{An interval $[x_i, x_{i+1}]$ that violates the invariant is ``fixed'' by first sampling the lowest height $h$ achieved within the interval,
    and then generating a position $x\in [x_i, x_{i+1}]$ such that $Height(x) = h$.}
    \label{fig:dyck_invariant_preserve}
\end{figure}

\subsubsection{Sampling the Lowest Achievable Height: Mandatory Boundary}
\label{sec:sampling_the_lowest_achievable_height}
First, we need to sample the lowest height $h$ of the walk between $x_i$ and $x_{i+1}$ (with corresponding heights $y_i$ and $y_{i+1}$).
We will refer to $h$ as the \emph{``mandatory boundary''} in this interval;
i.e. no height in the interval may be lower than the boundary, but at least one point \emph{must} touch the boundary (have height $h$).
We assume that $y_i\le y_{i+1}$ without loss of generality; if this is not the case, swap $x_i$ and $x_{i+1}$ and consider the reversed path.
Say this interval defines a generalized Dyck problem with $U$ up steps, $D$ down steps, and a boundary that is $k-1$ units below $y_i$.
\begin{figure}[htpb]
    \centering
    \includegraphics[width=0.9\textwidth]{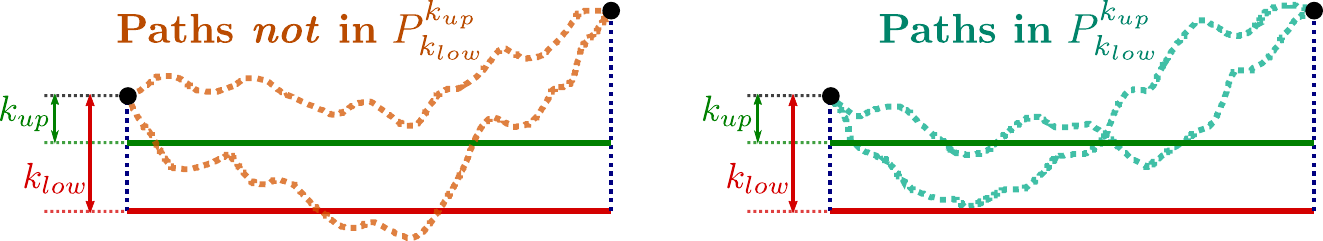}
    \caption{$P_{k_{low}}^{k^{up}}$ counts all the paths that dip at least $k_{up}$ units below the starting point,
    but \textbf{do not} dip $k_{low}$ units below the starting point. $k_{mid}$ would lie halfway between these boundaries.}
    \label{fig:dyck_mandatory_boundary}
\end{figure}

Given any two boundaries $k_{lower}$ and $k_{upper}$ on this interval (with $k_{lower} < k_{upper} \le y_i$),
we can count the number of possible generalized Dyck paths that violate the $k_{upper}$ boundary but \emph{not} the $k_{lower}$ boundary as follows
(see Figure~\ref{fig:dyck_mandatory_boundary}):
\[
P_{k_{lower}}^{k_{upper}} = C_{k_{lower}}(U,D) - C_{k_{upper}}(U,D)
\]
\begin{wrapfigure}[14]{r}{0.6\textwidth}
\vspace{-2.0em}
\begin{framed}
    \renewcommand\figurename{Algorithm}
    \caption{Finding the Mandatory boundary}
    \label{alg:mandatory_boundary}
    \begin{algorithmic}[1]
        \Function{Mandatory-Boundary}{$U, D, k$}
            \State {$k_{low}\gets k$}
            \State {$k_{up}\gets 0$}
            \While {$k_{low} < k_{up} - 1$}
                \vspace{.3em}
                \State {$k_{mid}\gets \floor{\frac{(k_{low} + k_{up})}{2}}$}
                \vspace{.3em}
                \State {$P_{total}\gets C_{k_{low}}(U,D) - C_{k_{up}}(U,D)$}
                \vspace{.3em}
                \State {$P_{k_{low}}^{k_{mid}}\gets C_{k_{low}}(U,D) - C_{k_{mid}}(U,D)$}
                \vspace{.3em}
                \State {\textbf{with probability} $P_{k_{low}}^{k_{mid}}/P_{total}$}
                    \State {\hspace{\algorithmicindent}$k_{up}\gets k_{mid}$}
                \State {\textbf{else}}
                    \State {\hspace{\algorithmicindent}$k_{low}\gets k_{mid}$}
            \EndWhile
            \State \Return $k_{low}$
        \EndFunction
    \end{algorithmic}
\end{framed}
\end{wrapfigure}
We define the current lower and upper boundaries as $k_{low} = k, k_{up} = 0$, and set $k_{mid} = (k_{low} + k_{up})/2$.
Since we can compute the quantities $P_{k_{mid}}^{k_{up}}$, $P_{k_{low}}^{k_{mid}}$, and $P_{total} = P_{k_{low}}^{k_{up}}$,
we can sample a single bit to decide if the \emph{``lower boundary''} should move up or if the \emph{``upper boundary''} should move down.
We then repeat this binary search until we find $k' = k_{low} = k_{up}-1$ and $k'$ becomes the \emph{``mandatory boundary''}
(i.e. the walk reaches the height exactly $k'-1$ units below $y_i$ but no lower.
\vspace{1.5em}

\subsubsection{Determine First Position Touching the \emph{``Mandatory Boundary''}}
\label{sec:sampling_first_position_touching_mandatory_boundary}

Now that we have a ``\emph{mandatory boundary}'' $k$, we just need to generate a position $x\in[x_i,x_{i+1}]$ with $\func{Height}(x) = y_i-k+1$,
according to the uniform distribution over all paths that touch, but do not violate the boundary at $k$.
In fact, we will do something stronger by determining the \emph{first} time the walk touches the boundary after $x_i$.
As before, we assume that this interval contains $U$ up steps and $D$ down steps.
\begin{figure}[htpb]
    \centering
    \includegraphics[width=\textwidth]{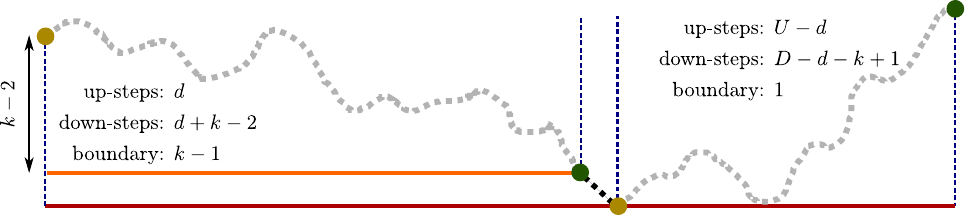}
    \caption{Zooming into the error in Figure~\ref{fig:dyck_boundary_invariant}.
        We generate a position $x$ (yellow) on the boundary (red),
        such that the section of the path to the left of $x$ never approaches the red boundary (it respects the orange boundary).}
    \label{fig:dyck_mandatory_boundary_sampling}
\end{figure}

We will parameterize the position $x$ the number of up-steps $d$ between $x_i$ and $x$ (See Figure~\ref{fig:dyck_mandatory_boundary_sampling}).
implying that $x = x_{i} + 2d + k - 1$.
Given a specific $d$, we want to compute the number of valid paths that result in
$d$ up-steps before the first approach to the boundary.
Note that unlike Section~\ref{sec:implementing_height_queries}, $d$ is used here to parameterize the (horizontal) $x$-position of the desired point.
We will calculate the probability $p_d$ associated with a particular position
by counting the total number of paths to the left and right of the first approach and multiplying them together.

As in Section~\ref{sec:the_difficult_case}, we define
$S_{left}$ to be the number of paths in the sub-interval before the first approach (left side of Figure~\ref{fig:dyck_mandatory_boundary_sampling}),
$S_{right}$ to be the number of paths following the first approach,
and $S_{total}$ to be the total number of paths that touch the ``mandatory boundary'' at $k$
(note that these quantities are functions of $U,D,k$ and $d$, but we drop the parameters for the sake of clarity):
{\small
\begin{align*}
    S_{left} = C_{k}(d, d+k-2)
    &&S_{right} = C_1(U-d, D-d-k+1)
    &&S_{total} = C_k(U,D) - C_{k-1}(U,D)
\end{align*}}
Our goal is to generate $d$ from the distribution $\{ p_d\}$ where $p_d = S_{left}\cdot S_{right}/S_{total}$.
The application of Lemma~\ref{lem:rejection_sampling} requires us to approximate $\{p_d\}$
with a ``well behaved'' $\{q_d\}$ (one whose CDF can be efficiently estimated).
Since $q_d$ only needs to approximate $p_d$ up to $\poly(\log n)$ factors,
we focus on obtaining $\poly(\log n)$ approximations to the path counts $\{ S_{left},S_{right},S_{total}\}$.
Using Equation~\ref{eq:catalan_trapezoid}, we obtain approximations for $S_{left}$ and $S_{right}$
(Lemma~\ref{lem:ReturnDLeftBound} and Lemma~\ref{lem:ReturnDRightBound} with proofs in Section~\ref{sec:omitted_supporting_first_return_queries}).
\begin{restatable}{lemma}{ReturnDLeftBound}
\label{lem:ReturnDLeftBound}
If $d > \log^4 n$, then $S_{left}(d)
= \Theta\left( \frac{2^{2d+k}}{\sqrt{d}}\mathlarger e^{-r_{left}(d)}\cdot \frac{k-1}{d+k-1}\right)$
where $r_{left}(d) = \frac{(k-2)^2}{2(2d+k-2)}$.
Furthermore, $r_{left}(d)=\mathcal O(\log^2 n)$.
\end{restatable}

\begin{restatable}{lemma}{ReturnDRightBound}
\label{lem:ReturnDRightBound}
If $U+D-2d-k > \log^4 n$, then $S_{right}(d)
= \Theta\left( \frac{2^{U+D-2d-k}}{\sqrt{U+d-2d-k}}\mathlarger e^{-r_{right}(d)}\cdot \frac{U-D+k}{U-d+1}\right)$
where $r_{right}(d) = \frac{(U-D-k-1)^2}{4(U+D-2d-k+1)}$.
Furthermore, $r_{right}(d)=\mathcal O(\log^2 n)$.
\end{restatable}

Our general strategy will be to integrate continuous versions of these approximations in order to obtain a CDF of some approximating distribution.
Unfortunately, the continuous approximation functions obtained do not admit closed form integrals.
The main culprit is the exponential term in both expressions.
We tackle this issue by noticing that the values of the exponents are bounded by $\mathcal O(\log^2 n)$ over the majority of the range of $d$.
Within this range of $d$ values, the exponential term may be further simplified by taking a piecewise constant approximation,
where each of the pieces corresponds to a fixed value of the \emph{floor of the corresponding exponent}.
This technique is elaborated in Section~\ref{sec:estimating_the_cdf}.

We start by considering the \emph{``problematic''} values of $d$ that are outside the range of the two preceding lemmas.
These values are the ones where $d < \log^4 n$ or $2d > U+D-k-\log^4 n$.
Since $d$ is the number of up steps in the left sub-interval, $d \ge 0$.
Further, since the length of the right sub-interval nust be non-negative (see Figure~\ref{fig:dyck_mandatory_boundary_sampling}),
we get $U+D-2d-k+1 \ge 0$.
Thus, we define the \emph{``problematic''} set:
\begin{align}
\label{eq:SetDefinition}
    \mathcal R = \left\{ d\ \mathlarger|\ 0\le d < \log^4 n \textrm{\textbf{\ or\ }} -1 < 2d-U-D+k < \log^4 n \right\}
\end{align}
Clearly, we can bound the size of this set as $|\mathcal R| = \mathcal O(\log^4 n)$.
An immediate consequence of Lemma~\ref{lem:ReturnDLeftBound} and Lemma~\ref{lem:ReturnDRightBound} is the following.

\begin{restatable}{corollary}{ReturnProbabilityBoundNotNormalized}
\label{cor:ReturnProbabilityBoundNotNormalized}
When $d\not \in \mathcal R$,
$S_{left}(d)\cdot S_{right}(d)
= \Theta\left( \frac{2^{U+D}}{\sqrt{d(U+D-2d-k)}}\cdot \mathlarger e^{-r(d)} \cdot \frac{k-1}{d+k-1}\cdot\frac{U-D+k}{U-d+1}\right)$
where $r(d)=\mathcal O(\log^2 n)$.
\end{restatable}

\subsubsection{Estimating the CDF}
\label{sec:estimating_the_cdf}
We use these observations to construct a suitable $\{q_d\}$ that can be used to invoke the rejection sampling lemma (Lemma~\ref{lem:rejection_sampling}).
We will achieve this by constructing a piecewise continuous function $\hat q$,
such that $\hat q(\delta)$ approximates $p_{\floor\delta}$, and then use the integral of $\hat q$ to define the discrete distribution $\{q_d\}$.
As stated in the previous section, we can leverage the fact that when $d \not \in \mathcal R$,
the floor of the exponent $\floor{r(d)}$ only takes $\mathcal O(\log^2 n)$ distinct values
(consequence of Corollary~\ref{cor:ReturnProbabilityBoundNotNormalized}).
Since the ``problematic'' set $\mathcal R$ only has $\mathcal O(\log^4 n)$ values,
we can also deal with these remaining values by simply creating $|\mathcal R|$ additional continuous pieces in the function $\hat q$.
We begin by rewriting $p_d = \Theta\left(\mathcal K \cdot f(d)\cdot e^{-r(d)}\right)$ where:
{\small
    \begin{align}
    \label{eq:dyck_integrable_function}
        \mathcal K = \frac{2^{U+D}}{S_{total}} = \frac{2^{U+D}}{C_k(U,D)-C_{k-1}(U,D)}\ \ \ \
        &&f(d) = \frac{(k-1)(U-D+k)}{\sqrt{d(U+D-2d-k)}(d+k-1)(U-d+1)}
    \end{align}}
Notice that $\mathcal K$ is a constant and $f(d)$ is a function whose integral has a closed form.
Using the fact that $r(d) = \mathcal O(\log^2 n)$ (Corollary~\ref{cor:ReturnProbabilityBoundNotNormalized}),
and $|\mathcal R| = \mathcal O(\log^4 n)$, we obtain the following lemma:

\begin{lemma}
\label{lem:ReturnProbabilityPiecewiseContinuous}
Given the piecewise continuous function
\begin{align*}
    \hat q(\delta) &=
    \begin{cases}
        p_{\floor\delta} &\mbox{if } \floor\delta \in \mathcal R \\ 
        \mathcal K\cdot f(\delta)\cdot exp\left(-\lfloor{r({\scriptstyle\floor\delta})}\rfloor\right)
        & \mbox{if } \floor\delta \not\in \mathcal R
    \end{cases}
    && \mathlarger \implies
    && p_d = \Theta\left( \int\limits_d^{d+1} \hat q(\delta)\right)
\end{align*}
Furthermore, $\hat q(\delta)$ has $\mathcal O(\log^4 n)$ continuous pieces.
\end{lemma}
\begin{proof}
For $d \in \mathcal R$, the integral trivially evaluates to exactly $p_d$.
For $d\not\in \mathcal R$, it suffices to show that $p_d = \Theta\left( \hat q(\delta)\right)$ for all $\delta\in [d,d+1)$.
We already know that $p_d = \Theta\left( \mathcal K\cdot f(d)\cdot e^{-r(d)}\right)$.
Moreover, for any $\delta\in [d,d+1)$, the exponential term $e^{-\floor{r(\floor\delta)}}$ in $\hat q(\delta)$
is within a factor of $e$ of the original $e^{-r(\floor\delta)}$ term.

For all $\mathcal O(\log^4 n)$ values $d\in \mathcal R$, $\hat q(\delta)$ is constant on the interval $[d,d+1]$.
Since $r(d) = \mathcal O(\log^2 n)$ by Corollary~\ref{cor:ReturnProbabilityBoundNotNormalized},
the exponential term $e^{-\floor{r(\floor\delta)}}$ in $\hat q(\delta)$ taken on at most $\mathcal O(\log^2 n)$ values.
Thus, $\hat q$ is continuous for a range of $\delta$ correspoding to a fixed value of $\floor{r(\floor\delta)}$,
and so, we conclude that $\hat q$ is piecewise continuous with $\mathcal O(\log^2 n)$ pieces.
\end{proof}

Now, we have everything in place to define the distribution $\{ q_d\}$ that we will be sampling from.
Specifically, we will define $q_d$ and it's CDF $Q_d$ as follows ($\mathcal N$ is the normalizing factor):
{\small
    \begin{align}
        q_d = \left(\int\limits_d^{d+1} \hat q(\delta)\right)\cdot \frac{1}{\mathcal N}
        && Q_d = \left(\int\limits_0^{d+1} \hat q(\delta)\right)\cdot \frac{1}{\mathcal N}
        && \textrm{where }\ \mathcal N = \int\limits_0^{d_{max}+1} \hat q(\delta)
    \end{align}}
To show that these can be computed efficiently, it suffices to show that any integral of $\hat q(\delta)$ can be efficiently evaluated.
This follows from the fact that $\hat q$ is piecewise continuous with $\mathcal O(\log^4)$ pieces (Lemma~\ref{lem:ReturnProbabilityPiecewiseContinuous}),
each of which has a closed form integral (since $f(d)$ defined in Equation~\ref{eq:dyck_integrable_function} has an integral).


\begin{lemma}
\label{lem:ReturnProbabilityPiecewiseContinuousIntegral}
Given the function $\hat q_d$ defined in Lemma~\ref{lem:ReturnProbabilityPiecewiseContinuous},
it is possible to approximate the integral $\int_{d_1}^{d_2+1} \hat q(\delta)$ to a multiplicative factor of $\left( 1\pm \frac{1}{n^2} \right)$,
in $\poly(\log n)$ time for any valid $d_1, d_2\in \mathbb Z$ (the bounds must be such that $d_i \ge 0$ and $U+D-2d_i-k+1 \ge 0$).
\end{lemma}
\begin{proof}
We will compute the integral by splitting it up into $\mathcal O(\log^4 n)$ continuous pieces and then approximating the integral over each piece.
The pieces corresponding to values of $d\in \mathcal R$ can be explicitly computed as $\int_d^{d+1} \hat q(\delta) = p_d$
(recall that we can compute $p_d$ using Lemma~\ref{lem:probability_approximation_oracle} from Section~\ref{sec:computing_probabilities}).

The more interesting pieces are over the range of $\delta$ where $\floor\delta\not\in \mathcal R$.
Such a piece is given by a continuous range of values $[\delta_{min}, \delta_{max}]\subseteq [d_1, d_2+1]$,
such that for any $\delta\in [\delta_{min}, \delta_{max}]$, the value of $\floor{r(\floor{\delta})}$ is a constant $E$.
We begin the calculation of $\hat q(\delta) = \mathcal K\cdot f(\delta)\cdot exp(-E)$,
by first computing a $\left( 1\pm 1/n^3\right)$ multiplicative approximation to $\ln \mathcal K = (U + D)\ln 2 - \ln S_{total}$,
using the strategy in Lemma~\ref{lem:probability_approximation_oracle} (Section~\ref{sec:computing_probabilities}).
Since $f(\delta)$ has a closed form integral (Equation~\ref{eq:dyck_integrable_function}),
we can also compute $F = \int_{\delta_{min}}^{\delta_{max}} f(\delta)$.
Using the fact that the exponent is constant over the range $[\delta_{min},\delta_{max}]$, we can write the logarithm of the integral as:
\[
\ln\left(\ \int\limits_{d_{min}}^{d_{max}} \hat q(\delta)\right) = \ln \left( \mathcal K\cdot e^{-E}\cdot F\cdot\right)
= (U + D)\cdot \ln 2 - \ln S_{total} + \ln F - E
\]
If this value is smaller than $-3\ln n$, we can safely ignore it since it contrbuted less than $1/n^3$ to the probability mass.
On the flipside, the logarithm is guaranteed to be bounded by $\mathcal O(1)$.
This upper bound is a result of the fact that $\int \hat q(\delta) = \mathcal O(\sum p_d) = \mathcal O(1)$,
where both the sum and the integral are taken over the entire \emph{valid} range of $d$.
This means that we can exponentiate to obtain the true value of the piecewise integral up to a multiplicative approximation of $(1\pm 1/n^3)$.
Adding the integrals of all the $\mathcal O(\log^4 n)$ pieces together produces the final desired value of the integral.
\end{proof}

Now, we are finally ready to use Lemma~\ref{lem:rejection_sampling} to sample $d$ from the distribution $\{p_d\}$,
using the efficient sampling procedure for $\{q_d\}$.
The only other requirement is the ability to approximate the $p_d$ values, which follows from Lemma~\ref{lem:probability_approximation_oracle}.
\begin{theorem}
\label{thm:first_return_in_interval}
Given a sub-interval $[x_{i},x_{i+1}]$ of a random Dyck path of length $2n$,
such that the only determined heights in the interval are $y_i=\func{Height}(x_i)$ and $y_{i+1}=\func{Height}(x_{i+1})$,
and a mandatory boundary constraint at $k$,
there exists an algorithm that generates a point $x$ within the interval such that $\func{Height}(x) = y_i-k+1$,
using $\poly(\log n)$ time, random bits, and additional space.
\end{theorem}

\subsubsection{Finding the Correct Interval: $\func{First-Return}$ Query}
\label{sec:finding_the_correct_interval}
As before, consider all positions that have been queried already $ \langle x_1, x_2,\cdots, x_m \rangle$ (in increasing order)
along with their corresponding heights $ \langle y_1, y_2,\cdots, y_m \rangle$.
We wish to find the first return to height $y_i$ after $x_i$ (where $y_i = \func{Height}(x_i)$).
Our strategy begins by using Invariant~\ref{inv:boundary_invariant} to find the interval $\mathcal I = [x_k, x_{k+1}]$
containing $\func{First-Return}(x_i)$.
\begin{figure}[htpb]
    \centering
    \includegraphics[width=\textwidth]{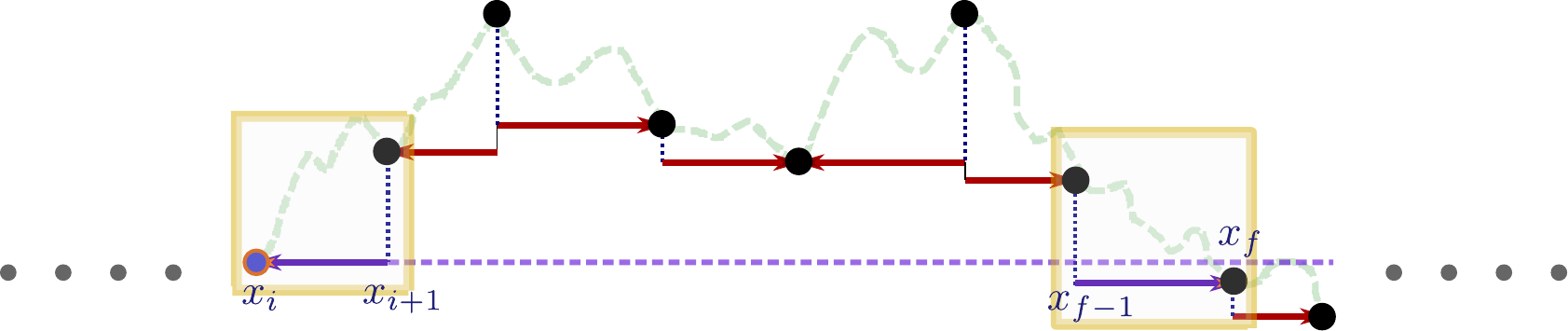}
    \caption{There are only two possible intervals (yellow boxes) that could contain $\func{First-Return}(x_i)$;
    either the interval adjacent to $x_i$, or the interval $[x_{f-1}, x_f]$, where $f$ is the smallest index such that $y_f < x_i$.
    }
    \label{fig:dyck_finding_correct_interval}
\end{figure}
\begin{restatable}{lemma}{FirstReturnInterval}
\label{lem:first_return_interval}
Assuming that Invariant~\ref{inv:boundary_invariant} holds, the interval $(x_{j-1},x_{j}]$ containing $\func{First-Return}(x_i)$
is obtained by setting $j$ to be either the smallest index $f>i$ such that $y_f\le y_i$ or setting $j-1=i$.
\end{restatable}
\begin{proof}
We assume the contrary i.e. there exists some $k\not=f$ and $k\not=i+1$ such that the correct interval is $(x_{k-1},x_k]$.
Since $y_f<y_i$, the position of first return to $y_i$ happens in the range $(x_i,x_f]$.
So, the only possibility is $i+1 < k \le f-1$.
By the definition of $y_f$, we know that both $y_k$ and $y_{k-1}$ are strictly larger than $y_i$.
Invariant~\ref{inv:boundary_invariant} implies that the boundary for this interval $(y_{k-1},y_k]$ is at $\min(y_{k-1},y_k) > y_i$.
So, it is not possible for the first return to be in this interval.
\end{proof}

The good news is that there are only two intervals that we need to worry about, one of which is just the adjacent one $[x_i, x_{i+1}]$.
The problem of finding the other interval that may contain the first return boils down to finding the smallest index $f>i$ such that $y_f\le y_i$.
To this end, we define $M_{[a,b]}$ as the minimum determined height in the range of positions $[a,b]$.

One solution is to maintain a range tree $\mathbf R$ \cite{comp_geo} over the range $[2n]$.
Assuming that $2n = 2^l$, we can view $\mathbf R$ as a complete binary tree with depth $r$.
Every non-leaf node is denoted by $\mathbf R_{[a,b]}$, and corresponds to a range $[a,b]\subseteq[2n]$ that is the union of the ranges of its children.
Each $\mathbf R_{[a,b]}$ stores the value $M_{[a,b]}$ which is is the minimum determined height
in the range of positions $[a,b]$, or $\infty$ if none of the heights have been revealed.
The leaf nodes are denoted as $\{\mathbf R_i\}_{i\in[2n]}$, and correspond to the singleton range corresponding to position $i\in [2n]$.
Note that a node at depth $l'$ will correspond to a range of size $2^{l-l'}$, with the root being associated with the entire range $[2n]$.

We say that the range $[a,b]$ is \emph{canonical} if it corresponds to a range of some $\mathbf{R}_{[a,b]}$ in $\mathbf R$.
By the property of range trees, any arbitrary range can be decomposed into a disjoint union of $O(\log n)$ canonical ranges.
We implement $\mathbf{R}$ to support the following operations:
\begin{itemize}
    \item \func{Update}$(x,y)$: Update the height of the position $x$ to $y$.\\
    This update affects all ranges $[a_i,b_i]$ containing $x$.
    So, for each $[a_i,b_i]$ we set $M_{[a_i,b_i]} = \min\left( M_{[a_i,b_i]}, y\right)$.
    \item \func{Query}$(a,b)$: Return the minimum boundary height in the range $[a,b]$.\\
    We decompose $[a,b]$ into $\mathcal O(\log n)$ \emph{canonical} ranges $ \langle r_1, r_2,\cdots\rangle$,
    and return the minimum of all the $M_{r_i}$ values as $M_{[a,b]}$ (since $[a,b]$ is the union of all $r_i$).
\end{itemize}

Now, we can binary search for $f$ by guessing a value $f'$ and checking if $\func{Query}(x_i,x_{f'}) \le y_i$.
Overall, this requires $\mathcal O(\log n)$ calls to $\func{Query}$, each of which makes $\mathcal O(\log n)$ probes to the range tree.
To avoid an initialization overhead, we only create the node $\mathbf{R}_{[a,b]}$ during the first $\func{Update}$ affecting a position $x\in[a,b]$.
Since a call to $\func{Update}$ can create at most $\mathcal O(\log n)$ new nodes in $\mathbf R$,
the additional space required for each $\func{Height}$ or $\func{First-Return}$ query is still bounded.

\begin{theorem}
\label{thm:dyck_first_return_sampling}
There is an algorithm using $\mathcal O(\log^{\mathcal O(1)} n)$ resources per query that provides access to a random Dyck path of length $2n$,
by answering queries of the form \func{First-Return}$(x_i)$ with the correctly sampled smallest position $x'>x_i$,
where the Dyck path first returns to $\func{Height}(x_i)$.
\end{theorem}
\begin{proof}
In order to make the presentation simpler, we ensure that the next determined position after $x_i$ is $x_i+1$ (in other words $x_{i+1} = x_i + 1$).
This can be done by invoking $\func{Height}(x_i+1)$, if needed.
If $\func{Height}(x_i+1) = y_{i+1} < y_i$, we can terminate because $\func{First-Return}(x_i)$ is not defined.
Otherwise, we notice that in this setting, the first return cannot lie in the adjacent interval $[x_i,x_{i+1}] = [x_i, x_i+1]$.

Hence, we proceed to finding the smallest value $f$ such that $y_f \le y_i$, by using the range tree data structure described above.
Since $\func{Height}(x_{f-1}) > y_i \le \func{Height}(x_f)$ by definition,
the interval $(x_{f-1},x_f]$ must contain at least one position at height $y_i$.
We determine the height at the midpoint of this interval,
and then fix the potential violation of Invariant~\ref{inv:boundary_invariant} by finding another point along with its height.
This essentially breaks up the interval $[x_{f-1},x_f]$ into $\mathcal O(1)$ sub-intervals, each at most half the size of the original.
Based on the newly revealed heights, we again find the (newly created) sub-interval containing the first return in $\mathcal O(1)$ time.
We repeat up to $\mathcal O(\log n)$ times, reducing the size of the intervals in consideration, until the position of the first return is revealed.
\end{proof}

\subsubsection{Maintaining \func{Height} Queries under Invariant~\ref{inv:boundary_invariant}}
\label{sec:maintaining_height_queries_under_invariant}
Finally, we show that the boundary constraints introduced in order to maintain Invariant~\ref{inv:boundary_invariant}
do not interfere with the implementation of \func{Height} queries.
As before, we consider the currently revealed heights $ \langle y_1, y_2,\cdots, y_m \rangle$,
along with the corresponding positions $ \langle x_1, x_2,\cdots, x_m \rangle$ (in increasing order).
Say that we are now presented with a query $\func{Height}(x)$, where $x_i < x < x_{i+1}$.
As in Section~\ref{sec:implementing_height_queries}, we swap $x_i$ and $x_{i+1}$ if necessary in order to ensure that $y_i < y_{i+1}$.
Due to Invariant~\ref{inv:boundary_invariant}, we know that the lowest achievable height in the interval $[x_i, x_{i+1}]$ is $y_i$,
i.e. the boundary constraint for the left half becomes $k = 1$ instead of $k = y_i + 1$, since the constrained boundary is at height $y_i$.
Similarly, the boundary constraint for the right half becomes $k' = 2U - 2D + 1$.

The algorithm for determining the height at the midpoint $x_{mid}$ of $[x_i,x_{i+1}]$
can proceed as described in Section~\ref{sec:implementing_height_queries}.
Of course, in this scenario, the boundary is never far away, and therefore we should always use the strategy in Section~\ref{sec:the_difficult_case}.
The second step is to re-establish Invariant~\ref{inv:boundary_invariant} in the newly created sub-interval $[x_{mid},x_{i+1}]$\footnote{
Note that the invariant is not broken for the $[x_i,x_{mid}]$ sub-interval since $y_i+1$ was the original boundary constraint.
In general, the invariant will be automatically preserved for one of the sub-intervals; the side corresponding to $\min(y_i,y_{i+1})$.},
using the two step strategy from Section~\ref{sec:maintaining_a_boundary_invariant}.
This involves determining the height of one additional point $x'\in [x_{mid},x_{i+1}]$.
Finally, we can continue with the height sampling algorithm from Theorem~\ref{thm:dyck_height_sampling},
by recursively halving one of the newly created intervals (the one containing $x$).
Figure~\ref{fig:dyck_invariant_height_queries} illustrates the aforementioned three steps.
\begin{figure}[htpb]
    \centering
    \includegraphics[width=\textwidth]{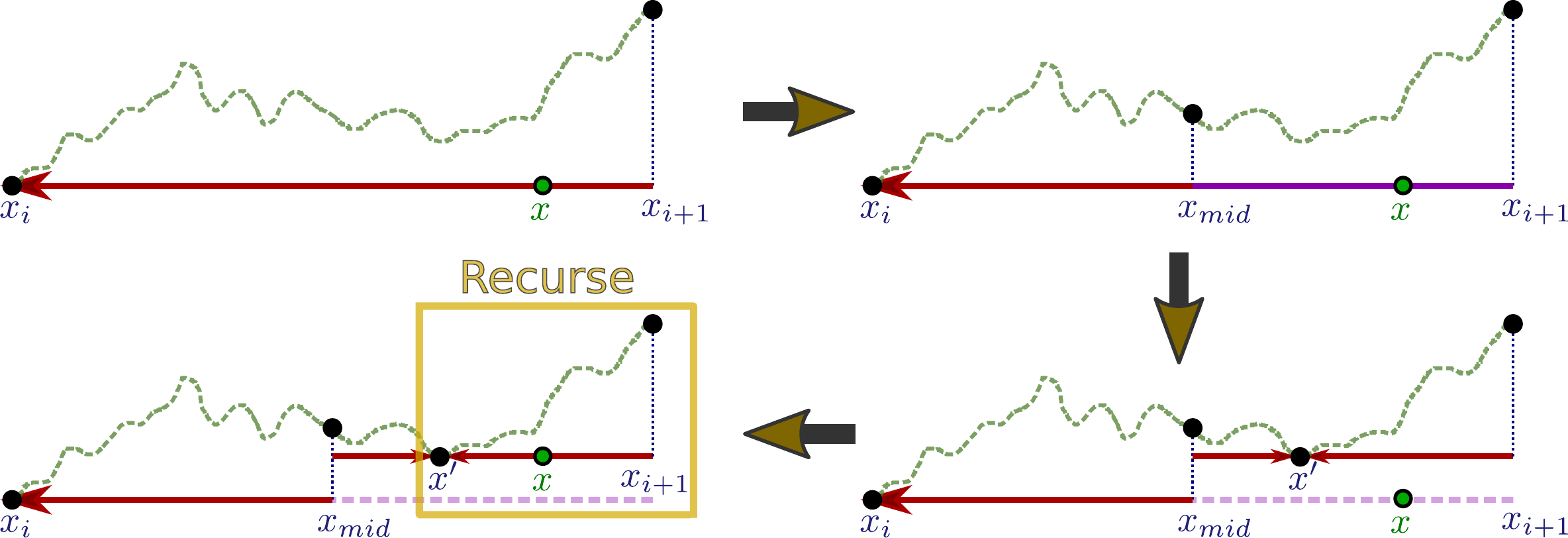}
    \caption{Performing \func{Height} queries while maintaining Invariant~\ref{inv:boundary_invariant}.
    The first step is to determine the height at the midpoint $x_{mid}$ of an existing interval $[x_i,x_{i+1}]$ that contains $x$.
    Next, we re-establish Invariant~\ref{inv:boundary_invariant} by sampling an additional point $x'$ within the violating interval
    Finally, we recurse on the appropriate sub-interval (yellow box).
    }
    \label{fig:dyck_invariant_height_queries}
\end{figure}

Hence, we can combine results from Theorem~\ref{thm:dyck_height_sampling} and Theorem~\ref{thm:dyck_first_return_sampling} to obtain the following:
\CatalanGrand*

\subsubsection{\func{Reverse-First-Return} Queries}
\label{sec:reverse_first_return_queries}
For the sake of completeness, we also sketch how to implement a \func{Reverse-First-Return} query,
which is just a standard \func{First-Return} query on the reversed Dyck path.
This type of query was motivated in Section~\ref{sec:bijections_to_other_catalan_objects},
and can be applied to positions $x$ where $\func{Height}(x-1) = \func{Height}(x)+1$.

Note that Invariant~\ref{inv:boundary_invariant} remains unchanged as it is symmetric to reversal of the Dyck path.
Specifically, the invariant on a particular interval does not change if we swap the endpoints of the interval.
First, we can find the correct interval containing $\func{Reverse-First-Return}(x)$, by using an analog of Lemma~\ref{lem:first_return_interval}.
\begin{restatable}{lemma}{ReverseFirstReturnInterval}
\label{lem:reverse_first_return_interval}
Assuming Invariant~\ref{inv:boundary_invariant}, the interval $[x_{j},x_{j+1})$ containing $\func{Reverse-First-Return}(x_i)$
is obtained by setting $j$ to be either the largest index $f<i$ such that $y_f\le y_i$ or setting $j+1=i$.
\end{restatable}
Since the range tree defined in Section~\ref{sec:finding_the_correct_interval} allows us to query the minimum in arbitrary contiguous ranges,
we can also use it to find the appropriate $f$ for Lemma~\ref{lem:reverse_first_return_interval} through binary search.
Once we have found an interval that is known to contain the \func{Reverse-First-Return},
we can follow the strategy from the proof of Theorem~\ref{thm:dyck_first_return_sampling};
recursively sub-dividing the relevant interval and finding the newly created sub-interval that contains the return,
until we converge on the correct return position.

\section{Random Coloring of a Graph}%
\label{sec:random_coloring_of_a_graph}
A \emph{valid} $q$-coloring of a graph $G=(V, E)$ is a vector of colors $\vec X \in [q]^V$, such that for all $(u,v)\in E$, $\vec X_u \not= \vec X_v$.
We present a sub-linear time algorithm to provide local access to a uniformly random valid $q$-coloring of an input graph.
Specifically, we implement $\func{Color}(v)$, which returns the color $\vec X_v$ of node $v$, where $\vec X$ is a uniformly random valid coloring.
The implementation can access the input graph $G$ through a sub-linear number of \emph{neighborhood queries}.
A neighborhood query of the form $\func{All-Neighbors}(v)$ returns a list of neighbors of $v$.
The implementation can also access a tape of public random bits $\vec R$.

Moreover, multiple independent instances of $\func{Color}$ that are given access to the same public tape of random bits $\vec R$,
should output color values consistent with a single $\vec X$, regardless of the order and content of the queries received.
Unlike our previous results, the choice of $\vec X$ only depends on $\vec R$,
and the \func{Color} implementations do not need any additional memory to maintain consistency.
For a formal description of this model, see Definition~\ref{def:local_access_LCA}.
This model is essentially a generalization of \emph{Local Computation Algorithms} \cite{LCA}.
Given a computation problem with input $x$ and a set of valid outputs $F(x)$,
LCAs provide query access to some $y\in F(x)$ using only a sub-linear number of probes to $x$.
This makes it feasible to consider sub-linear algorithms for problems where the output size is super-linear.
Our model has the additional restriction that the output must be drawn uniformly at random from $F(x)$, rather than just being an arbitrary member.
Since the number of possible outputs ($|F(x)|$) may be exponential, it is not possible to encode all the requisite random bits in sub-linear space.
Therefore, we use a source of public randomness $\vec R$.

\paragraph*{Sequential Algorithm for Random Coloring}
\label{par:sequential_algorithm_for_random_coloring}
We consider graphs with max degree $\Delta$, and $q = \Theta(\Delta)$, since this is the regime where this problem is feasible \cite{glauber_survey}.
In the sequential setting, \cite{glauber_survey} used the technique of path coupling to show that for $q > 2\Delta$,
one can sample an uniformly random coloring by using a simple Markov chain.
The Markov chain proceeds in $T$ steps. The state of the chain at time $t$ is given by $\vec X^t\in [q]^{|V|}$.
Specifically, the color of vertex $v$ at step $t$ is $\vec X^t_v$.
In each step of the Markov process, a vertex and a color are sampled uniformly at random i.e. a pair $(v, c)\thicksim_{\mathcal U} V\times [q]$.
Subsequently, if the recoloring of vertex $v$ with color $c$ does not result in a conflict with $v$'s neighbors,
i.e. $c\not\in \left\{ \vec X^t_u : u\in \Gamma(v)\right\}$, then the vertex is recolored i.e. $\vec X_v^{t+1}\leftarrow c$.
After running this chain for $T = \mathcal{O}(n\log (n/\epsilon))$ steps, the Markov chain is mixed,
implying that the distribution of resulting colors is $\epsilon$ close to the uniform distribution in $L_1$ distance.

\subsection{Modified Glauber Dynamics based on a Distributed Algorithm}%
\label{sec:modified_glauber_dynamics}

Now we define a modified Markov Chain that proceeds in epochs.
We denote the initial coloring of the graph by the vector $\vec X^0$ and the state of the coloring after the $k^{th}$ epoch by $\vec X^k$.
In the $k^{th}$ epoch, every node attempts to recolor itself simultaneously in a conservative manner, as described below:
\begin{itemize}
    \item Sample $|V|$ colors $ \langle c_1, c_2,\cdots, c_n \rangle$ from $[q]$, where $c_v$ is the color proposed by vertex $v$.
    \item For each vertex $v$, we set $\vec X^k_v$ to $c_v$, if and only if for all neighbors $w$ of $v$:
    \begin{enumerate}
        \item \fbox{$c_v \not= c_w$}
        The proposed color $c_v$ does not conflict with any of the neighbor proposals.
        \item \fbox{$c_v \not= \vec X^{k-1}_w$}
        Proposal $c_v$ does not conflict with any of the neighbor's current colors\footnote{Color of $w$ at the end of the previous epoch.}.
        \item \fbox{$c_w \not= \vec X^{k-1}_v$}
        None of the neighboring proposals $c_w$ conflict with current color of $v$.
    \end{enumerate}
\end{itemize}
This Markov Chain was introduced as the \emph{LocalMetropolis} algorithm by \cite{what_local}, in the context of distributed algorithms.
It produces a random valid coloring when $q > (2+\sqrt 2)\Delta$.
The bound on $q$ was further improved in \cite{yitong, ghaffari_fischer} by using a marking probability $\gamma$,
which indicates the likelihood of any vertex participating in a given round.
In the distributed setting, our epochs correspond to synchronous rounds, where many vertices re-color themselves simultaneously.

In order to bound the mixing time of this Markov chain, we use the standard technique of \emph{path coupling}, introduced in \cite{path_coupling}.
The argument begins by considering two initial states of the Markov Chains, say two colorings $\vec X^0$ and $\vec Y^0$,
that differ at only one vertex.
Formally, we can define the distance between two colorings $d(\vec X,\vec Y)$ as the number of vertices $v$ such that $\vec X_v\not= \vec Y_v$,
which results in the condition $d(\vec X^0, \vec Y^0) = 1$.
A \emph{coupling} is a joint evolution rule for a pair of states $(\vec X^0,\vec Y^0)\rightarrow(\vec X^1,\vec Y^1)$,
such that both of the individual evolutions $(\vec X^0\rightarrow \vec X^1)$ and $(\vec Y^0\rightarrow \vec Y^1)$
have the same transition probabilities as the original Markov Chain.
We can directly use the following main result from \cite{ghaffari_fischer} (setting $\gamma = 1$).
\begin{lemma}
\label{lem:ghaffari_fischer_single_epoch_distance}
If $q = 3\alpha\Delta$, then there exists a coupling $(\vec X^0,\vec Y^0)\rightarrow(\vec X^1,\vec Y^1)$, such that if $d(\vec X^0, \vec Y^0) = 1$,
then $\mathbb E[d(\vec X^1,\vec Y^1)] \le 1-\left( 1-\frac1{3\alpha}\right)e^{-2/\alpha} + \frac{1/3\alpha}{1-2/3\alpha}$
\end{lemma}
\begin{corollary}
\label{cor:single_epoch_distansce}
If $q \ge 9\Delta$ (or $\alpha > 3$) and $d(\vec X^0, \vec Y^0) = 1$, then $\mathbb E[d(\vec X^1,\vec Y^1)] < \frac1{e^{1/3}}$
\end{corollary}

The \emph{path coupling} lemma from \cite{path_coupling} uses a coupling on adjacent states to bound the mixing time.
\begin{lemma}
\label{lem:path_coupling}
\textbf{(Simplified Path Coupling from \cite{path_coupling})}
If there exists a coupling $(\vec X^0,\vec Y^0)\rightarrow(\vec X^1,\vec Y^1)$ defined for states where $d(\vec X^0, \vec Y^0) = 1$,
such that $\mathbb{E}[d(\vec X^1, \vec Y^1) \mid \vec X^0, \vec Y^0] \le \beta$ (for $\beta < 1$),
then, the mixing time $\tau_{mix}(\epsilon) = \mathcal O(\ln (n\epsilon^{-1})/\ln \beta^{-1})$.
\end{lemma}

\begin{corollary}
\label{cor:modified_mixing_time}
If $q\ge 9\Delta$, then the chain is mixed after $\tau_{mix}(\epsilon) = 3\left( \ln n + \ln(\frac1{\epsilon})\right)$ epochs.
\end{corollary}
\subsubsection{Naive Reduction from Distributed Algorithms}
\label{sec:naive_reduction_from_distributed_algorithms}
Using the technique of Parnas and Ron \cite{parnas_ron},
one can modify a distributed algorithm for a graph problem to construct a Local Computation Algorithm for the same problem.
Specifically, given a $k$-round distributed algorithm on a network of max degree $\Delta$,
we can simulate the behaviour and outcome of a single node $v$ by simulating the full algorithm on the $k$-neighborhood of $v$.
The simulation only requires us to probe this $k$-neighborhood, which contains $\mathcal O(\Delta^k)$ nodes.
However, the aforementioned distributed algorithm for Glauber Dynamics used $\mathcal O(\log n)$ rounds,
implying a probe complexity of $\Delta^{\mathcal O(\log n)}$, which is super-linear.
We show how to reduce the probe complexity by appropriately pruning the $k$-neighborhood.

\subsection{Local Coloring Algorithm}%
\label{sec:local_coloring_algortihm}
Given query access to the adjacency matrix of a graph $G$ with maximum degree $\Delta$ and a vertex $v$,
the algorithm has to output the color assigned to $v$ after running $t = \mathcal O(\ln n)$ epochs of \emph{Modified Glauber Dynamics}.
We want to be able to answer such queries in sub-linear time, without simulating the entire Markov Chain.
We will define the number of colors as $q = 3\alpha\Delta$ where $\alpha > 1$.

The proposals at each epoch are a vector of color samples $\vec C^{t} \thicksim_{\mathcal U} [q]^n$,
where $\vec C^t_v$ is the color proposed by $v$ in the $t^{th}$ epoch.
Since each of these $\vec C^t_v$ values are independent uniform samples from $[q]$,
instances of our algorithm will be able to access them in a consistent manner using the public random bits $\vec R$.

We also use $\vec X^t$ to denote the final vector of vertex colors at the end of the $t^{th}$ epoch.
Finally, we define indicator variables $\bm \chi^t_v$ to indicate whether the color $\vec C^t_v$ proposed by vertex $v$
was accepted at the $t^{th}$ epoch: $\bm \chi^t_v = \ONE$ if and only if for all neighbors $w\in \Gamma(v)$,
we satisfy the conditions $\vec C^t_w\not= \vec X^{t-1}_v$, $\vec C^t_v\not= \vec X^{t-1}_w$ and $\vec C^t_v\not= \vec C^t_w$
(see Section~\ref{sec:modified_glauber_dynamics}).

So, the color of a vertex $v$ after the $t^{th}$ epoch $\vec X^t_v$ will be $\vec C^i_v$
where $i\le t$ is the largest index such that $\bm \chi^i_v=1$.
While the proposals $\vec C^t_v$ can simply be read off the public random tape $\vec R$,
it is not clear how we can determine the $\bm \chi^t_v$ values efficiently.
Computing $\vec X^t_v$ is quite simple if we know the values $\bm\chi^i_v$ for all $i\le t$.
So, we focus our attention on the query $\func{Accepted}(v,t)$ that samples and returns the value of $\bm\chi^t_v$.

\subsubsection{Local Access to an Initial Valid Coloring}
\label{sec:local_access_to_an_initial_valid_coloring}
One caveat that we have not addressed is how we should initialize the Markov Chain.
The starting state can be any valid coloring of $G$, but we have to be able to access the initial colors of requisite vertices in sub-linear time.
One option is to simply assume that we have oracle access to an arbitrary valid coloring.
We can also invoke a result from \cite{coloring_initialize} that provides a \emph{Local Computation Algorithm} for $(\Delta+1)$-coloring of a graph.
Specifically, they provide local access to the color of any vertex using $\mathcal O(\Delta^{\mathcal O(1)}\log n)$ queries to the underlying graph,
such that the returned colors are consistent with some valid coloring.
Integrating their routine into our algorithm incurs a multiplicative $\poly(\Delta)$ overhead for each query.

Another option is to (uniformly) randomly intialize the colors of the vertex independently of its neighbors.
Although the initial coloring may be invalid, one can show that with high probability, the final coloring after running the Markov Chain is valid.
The intuitive reasoning is that each vertex attempts to recolor itself $\mathcal O(\log n)$ times,
and each attempt suceeds with constant probabiltity at least $1 - 1/\alpha$ (Lemma~\ref{lem:color_reject_probability}).
Furthermore, we can union bound to claim that \emph{all} the vertices get re-colored at least once with high probability.
After a vertex is re-colored once, it can no longer be in conflict with any of its neighbor's colors,
and therefore we obtain a final coloring that is valid.
For the sake of simplicity, we will assume that our algorithm can access the initial colors of any vertex $v$ through a public $\vec C^0_v$.

\subsubsection{Naive Coloring Implementations}%
\label{sec:naive_coloring_implementations}
Our general strategy to determine $\bm\chi^t_v$ will be to check for all neighbors $w$ of $v$,
whether $w$ causes a conflict with $v$'s proposed color in the $t^{th}$ epoch.
One naive way to achieve this, is to iterate backwards from epoch $t$, querying to find out whether $w$'s proposal was accepted,
until the most recent accepted proposal (latest epoch $t' < t$ such that $\bm\chi^{t'}_w=\ONE$) is found.
At this point, if $\vec C^{t'}_w =\vec C^t_v$, then the current color of $w$ conflicts with $v$'s proposal.
Otherwise there is no conflict, and we can proceed to the next neighbor.
However, this process potentially makes $\Delta$ recursive calls to a sub-problem that is only slightly smaller i.e. $T(t) \le \Delta\cdot T(t-1)$.
This leads to a running time upper bound of $\Delta^{t}$ which is superlinear for the desired number of epochs $t = \Omega(\log n)$ (the mixing time).

We can prune the number of recursive calls by only processing the neighbors $w$ which actually proposed the color $\vec C^t_v$ during \emph{some} epoch.
In this case, the expected number of neighbors that have to be probed recursively is less than $t\Delta/q$
(since the total number of neighbor proposals over $t$ epochs is at most $t\Delta$, and there are $q$ possible colors).
So, the overall runtime is upper bounded by $(t\Delta/q)^{t}$.
For this algorithm, if we allow $q > t\Delta = \Omega(\Delta\log n)$ colors, the runtime becomes sub-linear.
So, we can use this simple algorithm only when $q$ is sufficiently large.
However, we want a sub-linear time algorithm for $q = \mathcal O(\Delta)$.

\subsubsection{A Sub-linear Time Algorithm for $q = \mathcal O(\Delta)$}
\label{sec:jumping_back_to_past_epochs}
The expected number of neighbors that need to be checked recursively can always be $t\Delta/q$ in the worst case.
The crucial observation is that even though these recursive calls seem unavoidable,
we can aim to reduce the size of the recursive sub-problem, and thus bound the number of levels of recursion.

Algorithm~\ref{alg:coloring} shows our final procedure for generating $\bm\chi^t_v$,
where $\vec c = \vec C^t$ is used to represent the vector of proposals at epoch $t$.
We iterate through all neighbors $w$ of $v$, checking for conflicts (Line~\ref{alg:line:iterate_through_neighbors}).
The condition $\vec c_v\not=\vec c_w$ can be checked by reading $\vec c_w$ off the public random tape (Line~\ref{alg:line:check_immediate_conflict}).
If no conflict is seen, we proceed to check the remaining conditions: $\vec c_v\not= \vec X^{t-1}_w$ and $\vec c_w\not= \vec X^{t-1}_v$.

\begin{figure}[!ht]
    \renewcommand\figurename{Algorithm}
    \caption{Checking if proposal is accepted}
    \hrule
    \label{alg:coloring}
    \begin{algorithmic}[1]
        \Procedure{Accepted}{$v, t$}
            \State {$\vec c\gets\vec C^t$} \Comment{\parbox[t]{.6\linewidth}{\small Find current proposals using the public random bits}}
            \For{$w \gets \Gamma(v)$} \label{alg:line:iterate_through_neighbors}
                \Comment{\parbox[t]{.6\linewidth}{\small Iterate through the neighbors, checking for conflicts}}
                \If {$\vec C_w^t = \vec c_v$} \label{alg:line:check_immediate_conflict}
                    \Comment{\parbox[t]{.6\linewidth}{\small Check for conflict with neighbor's current proposal}}
                    \State \Return $\ZERO$
                        \Comment{\parbox[t]{.5\linewidth}{Conflict! This proposal is \textbf{not} accepted} \label{alg:line:conflict}}
                \EndIf
                \For{$t' \gets [t, t-1, t-2, \cdots, 1, 0]$} \label{alg:line:iterate_backwards}
                \Comment{\parbox[t]{.5\linewidth}{Check conflict with neighbor's prior color}}
                    \If {$\vec C^{t'}_w = \vec c_v$ \textbf{and} \func{Accepted}($w, t'$)}\label{alg:line:check_accepted}
                        \Comment{\parbox[t]{.40\linewidth}{\footnotesize Potential conflict: $w$ is colored $\vec c_v$}}
                        \State $overwritten\gets \FALSE$
                            \Comment{\parbox[t]{.40\linewidth}{\footnotesize Check if $\vec c_v$ overwritten in future epoch}}
                        \For{$\widetilde t \gets {\scriptstyle [t'+1, t'+2, t'+3, \cdots, t-1]}$} \label{alg:line:check_overwritten}
                            \If {\func{Accepted}($w, \widetilde t$)} \label{alg:line:check_overwritten_recursive}
                                \State $overwritten\gets \TRUE$
                                \State \textbf{break}
                            \EndIf
                        \EndFor
                        \If {\textbf{not} $overwritten$}
                        \State \Return $\ZERO$
                            \Comment{\parbox[t]{.5\linewidth}{Conflict! This proposal is \textbf{not} accepted} \label{alg:line:conflict1}}
                        \EndIf
                        \State \textbf{break}
                   \EndIf
                \EndFor
                \For{$t' \gets [t, t-1, t-2, \cdots, 1, 0]$} \label{alg:line:iterate_backwards2}
                    \Comment{\parbox[t]{.5\linewidth}{Check if $\vec c_w$ conflicts with $v$'s current color.}}
                    \If {$\vec C^{t'}_v = \vec c_w$ \textbf{and} \func{Accepted}($v, t'$)}\label{alg:line:check_accepted2}
                        \Comment{\parbox[t]{.40\linewidth}{\footnotesize Potential conflict: $v$ is colored $\vec c_w$}}
                        \State $overwritten\gets \FALSE$
                            \Comment{\parbox[t]{.40\linewidth}{\footnotesize Check if $\vec c_w$ overwritten in future epoch}}
                        \For{$\widetilde t \gets {\scriptstyle [t'+1, t'+2, t'+3, \cdots, t-1]}$} \label{alg:line:check_overwritten2}
                            \If {\func{Accepted}($v, \widetilde t$)} \label{alg:line:check_overwritten_recursive2}
                                \State $overwritten\gets \TRUE$
                                \State \textbf{break}
                            \EndIf
                        \EndFor
                        \If {\textbf{not} $overwritten$}
                        \State \Return $\ZERO$
                            \Comment{\parbox[t]{.5\linewidth}{Conflict! This proposal is \textbf{not} accepted} \label{alg:line:conflict2}}
                        \EndIf
                        \State \textbf{break}
                   \EndIf
                \EndFor
            \EndFor
            \State \Return $\ONE$ \Comment{\parbox[t]{.5\linewidth}{No conflicts! This proposal is accepted} \label{alg:line:no_conflict}}
        \EndProcedure
    \end{algorithmic}
    \hrule
\end{figure}

In order to check the first condition $\vec c_v\not= \vec X^{t-1}_w$,
we iterate through all the epochs in reverse order (line~\ref{alg:line:iterate_backwards})
to check whether the color $\vec c_v$ was ever proposed for vertex $w$ (if not, we can ignore $w$).
If this is the case, let's say that the most recent proposal for $\vec c_v$ was at epoch $t'$ i.e. $\vec C^{t'}_w = \vec c_v$.
Now, we directly ``jump'' to the $(t')^{th}$ epoch and recursively check if this proposal was accepted (line~\ref{alg:line:check_accepted}).
If the proposal $\vec C^{t'}_w$ was not accepted, we keep iterating back in time until we find the next most recent epoch
when $\vec c_v$ was proposed by $w$, or until we run out of epochs.
When we find the most recent epoch $t'$ in which $\vec c_v$ was accepted i.e. $\bm\chi^{t'}_w = \ONE$,
we successively consider epochs $t'+1, t'+2, t'+3, \cdots, t-1$, to see whether,
the color $\vec c_v$ was overwritten (line~\ref{alg:line:check_overwritten}) by an accepted proposal in a future epoch.
This is done by recursively invoking $\func{Accepted}(w,t'+i)$ in order to compute $\bm\chi^{t'+i}_w$
(line~\ref{alg:line:check_overwritten_recursive}), for $i\ge 1$.
If at any of these subsequent iterations, we see that a different proposal was accepted (thus overwriting the color $\vec c_v$),
then neighbor $w$ does not cause a conflict, and we can move on to the next neighbor.
Otherwise, we have seen that $\bm\chi^{t'}_w = \ONE$ (color $\vec c_v$ was accepted),
and every subsequent proposal until the current epoch $t$ was rejected;
this implies that color $\vec c_v$ \emph{survived} as the color of neighbor $w$, i.e. $\vec X^{t-1}_w = \vec c_v$.
This leads to a conflict with $v$'s current proposal for color $\vec c_v$ (line~\ref{alg:line:conflict}) and hence $\bm\chi^t_v = \ZERO$.

We can also check the second condition $\vec c_w\not= \vec X^{t-1}_v$ in a similar manner.
In this case, we would iterate back through epochs to find a $t'$ such that $\vec C^{t'}_v = \vec c_w$
i.e. vertex $v$ proposed color $\vec c_w$ (neighbor's current proposal) at epoch $t'$ (line~\ref{alg:line:check_accepted2}).
As before, if such a $t'$ is found, we step fowards through the epochs to check
whether the color of $v$ remains $\vec c_w$, or whether it was overwritten (line~\ref{alg:line:check_overwritten2}).
If we exhaust all the neighbors and don't find any conflicts of any kind (line~\ref{alg:line:no_conflict}), then $\bm\chi^t_v = \ONE$.

Now we analyze the runtime of $\func{Accepted}$ by constructing and solving a recurrence relation.
We will use the following lemma to evaluate the expectation of products of relevant random variables.

\begin{lemma}
\label{lem:color_reject_probability}
The probability that any given proposal is rejected $\mathbb P[\bm\chi^t_v=\ZERO]$ is at most $1/\alpha$.
Moreover, this upper bound holds even if we condition on all the values in $\vec C$ except $\vec C^t_v$.
\end{lemma}
\begin{proof}
Let us consider a process for generating the random vector $\vec C^t$ where we first

For $w\in \Gamma(v)$, the probability that proposal $\vec C^t_w$ does not conflict with the current color of $v$ is at least $(1-1/q)$.
Since there are at most $\Delta$ neighbors, the probability
that none of the neighbor proposals conflict with $\vec X^{t-1}_v$ is at least $(1-1/q)^{\Delta}$.
The only other way for a rejection to occur is if the proposal $\vec C^t_v$
conflicts with any of the colors in $\{\vec C^t_w, \vec X^{t-1}_w | w\in\Gamma(v)\}$.
Since this set contains at most $2\Delta$ possible values and there are $q=3\alpha\Delta$ colors, the acceptance probability is at least $1-2/3\alpha$.
Thus, overall acceptance probability is at least:
\[
\left(1-\frac{1}{q}\right)^{\Delta}\cdot\left( 1-\frac{2}{3\alpha}\right) \ge \left(1 - \frac{\Delta}{q}\right)\cdot\left( 1-\frac{2}{3\alpha}\right)
= \left(1 - \frac{1}{3\alpha}\right)\cdot\left( 1-\frac{2}{3\alpha}\right) \ge 1 - \frac{1}{\alpha}
\]
Thus, the rejection probability is at most $\frac{1}{\alpha}$.
\end{proof}

\begin{definition}
\label{def:coloring_recursions}
We define $T_t$ to be a random variable indicating the number of recursive calls performed during the execution of $\func{Accepted}(v,t)$
while computing a single $\bm \chi_v^t$.
\end{definition}

\begin{definition}
\label{def:blah}
We define $W^t_{t'}$ to be a random variable indicating the number of calls to \func{Accepted} that are required,
to check whether a color $\vec c_v$ assigned at epoch $t'$ was overwritten at some epoch before $t$.
\end{definition}
Using $\mathcal B(p)$ to denote the Bernoulli random variable with bias $p$, we obtain an expression for $W^t_{t'}$
(using $\le$ to denote stochastic dominance).
\begin{align}
\label{eq:color_overwrite}
W^t_{t'} \le
\Biggl[T_{t'+1} + \mathcal B\left(\frac{1}{\alpha}\right)\cdot T_{t'+2}
+ \mathcal B\left(\frac{1}{\alpha^2}\right)\cdot T_{t'+3} + \cdots
+ \mathcal B\left(\frac{1}{\alpha^{t-t'-2}}\right)\cdot T_{t-1} \Biggr]
\end{align}
The above Equation~\ref{eq:color_overwrite} conservatively assumes that the call to $\func{Accepted}(v, t'+1)$
in line~\ref{alg:line:check_overwritten_recursive} is always invoked (resulting in $T_{t'+1}$ invocations of \func{Accepted}).
However, the next call to $\func{Accepted}(v, t'+2)$ occurs only if the previous one was not accepted,
which happens with probability $\le 1/\alpha$ (Lemma~\ref{lem:color_reject_probability}).
This produces the $\mathcal B(1/\alpha)\cdot T_{t'+2}$ term in the expression.
In general, $\func{Accepted}(v, t'+i)$ is only invoked if the preceding $i-1$ calls to \func{Accepted} all returned $\ZERO$.
This event happens with probability at most $1/\alpha^{i-1}$.

Next, we prove our main lemma that bounds the number of recursive calls to \func{Accepted}.

\begin{lemma}
\label{lem:coloring_recurrence}
Given graph $G$ and $q=3\alpha\Delta$ colors, for $\alpha > 3$, the expected number of recursive calls to the procedure $\func{Accepted}$
while computing a single $\bm\chi^t_v = \func{Accepted}(v,t)$ is $\mathbb E[T_t] = \mathcal{O}\left(e^{4t/3\alpha}\right)$.
\end{lemma}
\begin{proof}
We start by constructing a recurrence for the expected number of calls to $\func{Accepted}$ used by the algorithm.
When checking for conflicts of the form $\vec X^{t-1}_w = \vec c_v$ (for a specific neighbor $w$),
the algorithm iterates through all the epochs $t'$ such that $\vec C^{t'}_w = \vec c_v$
(in reality, only the last occurence matters, but we are looking for an upper bound).
If such a $t'$ is found (which happens with probability $1/q$ independently for each trial), there is a recursive call to $\func{Accepted}(w,t')$,
which in turn results in $\textcolor{RawSienna}{T_{t'}}$ recursive calls to \func{Accepted}.
If we find that $\bm\chi^{t'}_w = \ONE$ (i.e. $w$ was colored to $\vec c_v$ at epoch $t'$),
we will need to proceed to check whether the color was subsequently overwritten,
which requires an additional $\textcolor{OliveGreen}{W^t_{t'}}$ calls to \func{Accepted}.
Thus, the recursive calls from within the first for loop (in Algorithm~\ref{alg:coloring})
result in the first term in Equation~\ref{eq:coloring_recurrence}.

Similarly, during the second for loop in Algorithm~\ref{alg:coloring}, which checks for conflicts of the form $\vec X^{t-1}_v = \vec c_w$,
the algorithm iterates through all the epochs $t'$ such that $\vec C^{t'}_v = \vec c_w$.
As before, such a $t'$ is found with probability $1/q$ independently for each pair $(w, t')$,
and once found it spawns $(\textcolor{RawSienna}{T_{t'}}+\textcolor{OliveGreen}{W^t_{t'}})$ recursive calls to \func{Accepted}, in expectation.

Summing up over all neighbors and epochs, we obtain the following bound:
\begin{align}
\label{eq:coloring_recurrence}
T_{t}
&\le \textcolor{Fuchsia}{\mathlarger\sum\limits_{w\in \Gamma(v)}}\ \
\mathlarger\sum\limits_{t'=1}^{t}\ \  \textcolor{Blue}{\mathbb P\left[\vec C^{t'}_w = \vec c_v\right]}\cdot
\left[ \textcolor{RawSienna}{T_{t'}} + \textcolor{OliveGreen!80!black}{W^t_{t'}} \right]
+ \textcolor{Fuchsia}{\mathlarger\sum\limits_{w\in \Gamma(v)}}\ \
\mathlarger\sum\limits_{t'=1}^{t}\ \  \textcolor{Blue}{\mathbb P\left[\vec C^{t'}_v = \vec c_w\right]}\cdot
\left[ \textcolor{RawSienna}{T_{t'}} + \textcolor{OliveGreen!80!black}{W^t_{t'}} \right] \\
&\le \textcolor{Fuchsia}{2\Delta} \cdot \mathlarger\sum\limits_{t'=1}^{t} \textcolor{Blue}{\mathcal B\left( \frac 1q\right)}\cdot
\Biggl[ \textcolor{RawSienna}{T_{t'}} +
\textcolor{OliveGreen!80!black}{T_{t'+1} + \mathcal B\left(\frac{1}{\alpha}\right)\cdot T_{t'+2} +
\mathcal B\left(\frac{1}{\alpha^2}\right)\cdot T_{t'+3} + \cdots} \notag\\
&\phantom{{}=3} \phantom{{}=3} \phantom{{}=3} \phantom{{}=3} \phantom{{}=3} \phantom{{}=3}
\phantom{{}=3} \phantom{{}=3} \phantom{{}=3} \phantom{{}=3} \phantom{{}=3} \phantom{{}=3}
\textcolor{OliveGreen!80!black}{\cdots + \mathcal B\left(\frac{1}{\alpha^{t-t'-2}}\right)\cdot T_{t-1}} \Biggr]\\
&\le \textcolor{Fuchsia}{2\Delta} \cdot \textcolor{Blue}{\mathcal B\left( \frac{1}{q}\right)} \cdot \Biggl[
\textcolor{RawSienna}{\mathlarger\sum\limits_{t'=1}^{t-1} T_{t'}} +
\textcolor{OliveGreen!80!black}{\mathlarger\sum\limits_{t'=1}^{t-1} T_{t'}\cdot
\left(1 + \mathcal B\left(\frac1\alpha\right) + \mathcal B\left(\frac1{\alpha^2}\right) + \cdots\right)}
\Biggr]
\end{align}
In the last step, we just grouped together all the terms corresponding to the same epoch
(note that we include additional terms since it's an upper bound).
Using Lemma~\ref{lem:color_reject_probability} and the fact that $\textcolor{Blue}{\mathbb P[\vec C^{t'}_x = \vec C^t_y]}$
is independent of all other events, we can write a recurrence for the expected number of probes.
\begin{align}
\mathbb E[T_t] \le \textcolor{Fuchsia}{2\Delta} \cdot \textcolor{Blue}{\frac{1}{3\alpha\Delta}}
\left[
\textcolor{RawSienna}{\mathlarger\sum\limits_{t'=1}^{t-1} T_{t'}} +
\textcolor{OliveGreen!80!black}{\mathlarger\sum\limits_{t'=1}^{t-1} T_{t'}\cdot \left(1 + \frac1\alpha + \frac1{\alpha^2} + \cdots\right)}
\right]
\le \frac{2}{3\alpha}\cdot \mathlarger\sum\limits_{t'=1}^{t-1} T_{t'}\cdot \left[1 + \frac{\alpha}{\alpha-1} \right]
\end{align}
Now, we make the assumption that $\mathbb E[T_{t'}]\le e^{k t/\alpha}$,
and show that this satisfies the expectation recurrence for the desired value of $k$.
First, we sum the geometric series:
\[
\mathlarger\sum\limits_{t'=1}^{t-1} \mathbb E[T_{t'}] = \mathlarger\sum\limits_{t'=1}^{t-1} e^{k t'/\alpha}
< \frac{e^{k t/\alpha}-1}{e^{k/\alpha}-1} < \frac{e^{k t/\alpha}}{e^{k/\alpha}-1}
\]
The expectation recurrence to be satisfied then becomes:
\[
\mathbb E[T_t]\le \frac 2{3\alpha}\cdot \frac{e^{k t/\alpha}}{e^{k/\alpha}-1}\cdot \left[ 1+ \frac{\alpha}{\alpha-1} \right]
= e^{k t/\alpha}\cdot \frac{2(2\alpha-1)}{3\alpha(\alpha-1)(e^{k/\alpha}-1)} = e^{k t/\alpha}\cdot f(\alpha, k)
\]
We notice that for $k=4/3$ and $\alpha > 3$, $f(\alpha) < 1$.
This can easily be verified by checking that $f(\alpha,4/3)$ decreases monotonically with $\alpha$ in the range $\alpha > 3$.
Thus, our recurrence is satisfied for $k=4/3$, and therefore the expected number of calls is $\mathcal O(e^{4t/3\alpha})$.
\end{proof}

\ColoringGrand*
\begin{proof}
Since $q\ge 9\Delta$, we can use Corollary~\ref{cor:modified_mixing_time} to obtain $\tau_{mix}(\epsilon) \le 3(\ln n + \ln 1/\epsilon)$.
Also, since $\alpha > 3$, we can invoke Lemma~\ref{lem:coloring_recurrence} to conclude that
the number of calls to $\func{Accepted}$ is $\mathcal O(n^{4/\alpha}\epsilon^{-4/\alpha})$.
Finally, we note that each call to \func{Accepted}$(v,t)$ potentially reads $\mathcal O(t\Delta)$ color proposals from the public random tape,
while iterating through all $\le \Delta$ neighbors of $v$ in all $t$ epochs.
Since $t \le 3\ln (n/\epsilon)$, this implies that the algorithm uses
$\mathcal O((n/\epsilon)^{4/\alpha}\Delta\log (n/\epsilon))$ time and random bits, which is sub-linear for $\alpha > 4$.
\end{proof}

\section{Open Problems}%
\label{sec:open_problems}
There are many interesting directions to pursue in this area.
Below, we provide a few examples of random objects that may admit local access implementations.

\subparagraph*{Small Description Size}
\label{par:small_description_size}
\begin{itemize}
    \item Provide a local access implementation of degree queries for undirected random graphs, even for $G(n,p)$.
    How about $i^{th}$ neighbor queries?
    \item For simple models such as $G(n,p)$, provide a local access implementation of a \func{Random-Triangle}$(v)$ query,
    that returns a uniformly random triangle containing vertex $v$.
    \item Provide a memory-less local access implementation of basic queries for undirected random graphs.
    \item Given an ordered graph such as a lattice, provide an implementation to locally access a random perfect matching.
    Interesting special cases of this problem include random domino and lozenge tilings.
\end{itemize}

\subparagraph*{Huge Description Size}
\label{par:huge_description_size}
\begin{itemize}
    \item Provide a faster local access implementation for sampling the color of a specified vertex
    in a random $q$-coloring of a bounded degree graph $G$.
    \item Improve the local access implementation for sampling the color of a specified vertex in a random $q$-coloring of $G$,
    by supporting smaller values of $q$ (smaller than $12\Delta$).
    We remark that this problem in particular should be feasible, by simulating a faster mixing Markov chain.
    The important question is whether you can you get down to $q = 2\Delta$?
    \item Given query access to an input graph $G$ and starting vertex $v$,
    provide a local access implementation for sampling the location of a random walk starting at $v$ after $t$ steps.
    This may be feasible in certain restricted classes of graphs.
    \item Given query access to an input DNF formula, provide an implementation to access the truth value of a single variable in a uniformly random satisfying assignment.
\end{itemize}

\bibliography{bob}

\clearpage
\appendix
\label{sec:appendix}

\section{Further Analysis and Extensions of Algorithm~\ref{alg:oblivious-coin-toss}: Sampling \func{Next-Neighbor} without Blocks}
\label{sec:reroll-cont}

\subsection{Performance Guarantee}
This section is devoted to showing the following lemma that bounds the required resources per query of Algorithm~\ref{alg:oblivious-coin-toss}. We note that we only require efficient computation of $\prod_{u \in [a,b]} (1-p_{vu})$ (and not $\sum_{u \in [a,b]} p_{vu}$), and that for the $G(n,p)$ model, the resources required for such computation is asymptotically negligible.

\begin{restatable}{theorem}{res:ER-rand-iterations}\label{thm:ER-rand-iterations}
Each execution of Algorithm~\ref{alg:oblivious-coin-toss} (the \func{Next-Neighbor} query), with high probability,
\begin{itemize}
\item terminates within $\bo(\log n)$ iterations (of its \textup{\textbf{repeat}} loop);
\item computes $\bo(\log^2 n)$ quantities of $\prod_{u \in [a,b]} (1-p_{vu})$;
\item aside from the above computations, uses $\bo(\log^2 n)$ time, $\bo(\log n)$ random $N$-bit words, and $\bo(\log n)$ additional space.
\end{itemize}
\end{restatable}

\begin{proof}
We focus on the number of iterations as the remaining results follow trivially. This proof is rather involved and thus is divided into several steps.

\paragraph*{Specifying random choices} The performance of the algorithm depends on not only the random variables $X_{vu}$'s, but also the unused coins $C_{vu}$'s. We characterize the two collections of Bernoulli variables $\{X_{vu}\}$ and $\{Y_{vu}\}$ that cover all random choices made by Algorithm~\ref{alg:oblivious-coin-toss} as follows.

\begin{itemize}
\item Each $X_{vu}$ (same as $X_{uv}$) represents the result for the \emph{first} coin-toss corresponding to cells $\ADJ[v][u]$ and $\ADJ[u][v]$, which is the coin-toss obtained when $X_{vu}$ becomes decided: either $C_{vu}$ during a \func{Next-Neighbor}$(v)$ call when $\ADJ[v][u] = \PHI$, or $C_{vu}$ during a \func{Next-Neighbor}$(u)$ call when $\ADJ[u][v] = \PHI$, whichever occurs first.
This description of $X_{vu}$ respects our invariant that, if the generation process is executed to completion, we will have $\ADJ[v][u]=X_{vu}$ in all entries.
\item Each $Y_{vu}$ represents the result for the \emph{second} coin-toss corresponding to cell $\ADJ[v][u]$, which is the coin-toss $C_{vu}$ obtained during a \func{Next-Neighbor}$(v)$ call when $X_{vu}$ is already decided. In other words, $\{Y_{vu}\}$'s are the coin-tosses that should have been skipped but still performed in Algorithm~\ref{alg:oblivious-coin-toss} (if they have indeed been generated). Unlike the previous case, $Y_{vu}$ and $Y_{uv}$ are two independent random variables: they may be generated during a \func{Next-Neighbor}$(v)$ call and a \func{Next-Neighbor}$(u)$ call, respectively.
\end{itemize}
As mentioned earlier, we allow any sequence of probabilities $p_{vu}$ in our proof. The success probabilities of these indicators are therefore given by $\mathbb P[X_{vu}=\ONE] = \mathbb P[Y_{vu}=\ONE] = p_{vu}$.

\paragraph*{Characterizing iterations}
Suppose that we compute \func{Next-Neighbor}$(v)$ and obtain an answer $u$. Then $X_{v,\LAST[v]+1} = \cdots = X_{v, u-1} = \ZERO$ as none of $u' \in (\LAST[v], u)$ is a neighbor of $v$. The vertices considered in the loop of Algorithm~\ref{alg:oblivious-coin-toss} that do not result in the answer $u$, are $u' \in (\LAST[v], u)$ satisfying $\ADJ[v][u'] = \ZERO$ and $Y_{v,u'} = \ONE$; we call the iteration corresponding to such a $u'$ a \emph{failed iteration}. Observe that if $X_{v,u'} = \ZERO$ but is undecided ($\ADJ[v][u'] = \PHI$), then the iteration is not failed, even if $Y_{v,u'} = \ONE$ (in which case, $X_{v,u'}$ takes the value of $C_{v,u'}$ while $Y_{v,u'}$ is never used). Thus we assume the worst-case scenario where all $X_{v,u'}$ are revealed: $\ADJ[v][u']=X_{v,u'}=\ZERO$ for all $u'\in(\LAST[v], u)$. The number of failed iterations in this case stochastically dominates those in all other cases.\footnote{There exists an adversary who can enforce this worst case. Namely, an adversary that first makes \func{Next-Neighbor} queries to learn all neighbors of every vertex except for $v$, thereby filling out the whole $\ADJ$ in the process. The claimed worst case then occurs as this adversary now repeatedly makes \func{Next-Neighbor} queries on $v$. In particular, a committee of $n$ adversaries, each of which is tasked to perform this series of calls corresponding to each $v$, can always expose this worst case.}

Then, the upper bound on the number of failed iterations of a call \func{Next-Neighbor}$(v)$ is given by the maximum number of cells $Y_{v, u'} = 1$ of $u' \in(\LAST[v], u)$, over any $u \in(\LAST[v], n]$ satisfying $X_{v,\LAST[v]+1} = \cdots = X_{vu} = \ZERO$. Informally, we are asking ''of all consecutive cells of $\ZERO$'s in a single row of $\{X_{vu}\}$-table, what is the largest number of cells of $\ONE$'s in the corresponding cells of $\{Y_{vu}\}$-table?''

\paragraph*{Bounding the number of iterations required for a fixed pair $(v, \LAST[v])$}
We now proceed to bounding the number of iterations required over a sampled pair of $\{X_{vu}\}$ and $\{Y_{vu}\}$, from any probability distribution. For simplicity we renumber our indices and drop the index $(v,\LAST[v])$ as follows. Let $p_1, \ldots, p_L \in [0, 1]$ denote the probabilities corresponding to the cells $\ADJ[v][\LAST[v]+1 \ldots n]$ (where $L = n-\LAST[v]$), then let $X_1, \ldots, X_L$ and $Y_1, \ldots, Y_L$ be the random variables corresponding to the same cells on $\ADJ$.

For $i=1, \ldots, L$, define the random variable $Z_i$ in terms of $X_i$ and $Y_i$ so that
\begin{itemize}
\item $Z_i = 2$ if $X_i = 0$ and $Y_i = 1$, which occurs with probability $p_i(1-p_i)$. \\
This represents the event where $i$ is not a neighbor, and the iteration fails.
\item $Z_i = 1$ if $X_i = Y_i = 0$, which occurs with probability $(1-p_i)^2$.\\
 This represents the event where $i$ is not a neighbor, and the iteration does not fail.
\item $Z_i = 0$ if $X_i = 1$, which occurs with probability $p_i$. \\
This represents the event where $i$ is a neighbor.
\end{itemize}

For $\ell \in [L]$, define the random variable $M_\ell := \prod_{i=1}^\ell Z_i$, and $M_0 = 1$ for convenience. If $X_i = 1$ for some $i \in [1, \ell]$, then $Z_i = 0$ and $M_\ell = 0$. Otherwise, $\log M_\ell$ counts the number of indices $i \in [\ell]$ with $Y_i = 1$, the number of failed iterations. Therefore, $\log(\max_{\ell \in \{0, \ldots, L\}} M_\ell)$ gives the number of failed iterations this \func{Next-Neighbor}$(v)$ call.

To bound $M_\ell$, observe that for any $\ell\in[L]$, $\mathbb{E}[Z_\ell] = 2p_\ell(1-p_\ell) + (1-p_\ell)^2 = 1 - p_\ell^2 \leq 1$ regardless of the probability $p_\ell \in [0, 1]$. Then, $\mathbb{E}[M_\ell] = \mathbb{E}[\prod_{i=1}^\ell Z_i] = \prod_{i=1}^\ell \mathbb{E}[Z_i] \leq 1$ because $Z_\ell$'s are all independent. By Markov's inequality, for any (integer) $r \geq 0$, $\Pr[\log M_\ell > r] = \Pr[M_\ell > 2^r] < 2^{-r}$. By the union bound, the probability that more than $r$ failed iterations are encountered is $\Pr[\log(\max_{\ell \in \{0, \ldots, L\}} M_\ell) > r] < L\cdot 2^{-r} \leq n\cdot 2^{-r}$.

\paragraph*{Establishing the overall performance guarantee}
So far we have deduced that, for each pair of a vertex $v$ and its $\LAST[v]$, the probability that the call \func{Next-Neighbor}$(v)$ encounters more than $r$ failed iterations is less that $n \cdot 2^{-r}$, which is at most $n^{-c-2}$ for any desired constant $c$ by choosing a sufficiently large $r = \Theta(\log n)$. As Algorithm~\ref{alg:oblivious-coin-toss} may need to support up to $\Theta(n^2)$ \func{Next-Neighbor} calls, one corresponding to each pair $(v, \LAST[v])$, the probability that it ever encounters more than $O(\log n)$  failed iterations to answer a single \func{Next-Neighbor} query is at most $n^{-c}$. That is, with high probability, $O(\log n)$ iterations are required per \func{Next-Neighbor} call, which concludes the proof of Theorem~\ref{thm:ER-rand-iterations}.
\end{proof}

\subsection{Supporting \func{Vertex-Pair} Queries} \label{sec:ER-pair}

We extend our implementation (Algorithm~\ref{alg:oblivious-coin-toss}) to support the \func{Vertex-Pair} queries: given a pair of vertices $(u, v)$, decide whether there exists an edge $\{u, v\}$ in the generated graph. To answer a \func{Vertex-Pair} query, we must first check whether the value $X_{uv}$ for $\{u, v\}$ has already been assigned, in which case we answer accordingly. Otherwise, we must make a coin-flip with the corresponding bias $p_{uv}$ to assign $X_{uv}$, deciding whether $\{u, v\}$ exists in the generated graph. If we maintained the full $\ADJ$, we would have been able to simply set $\ADJ[u][v]$ and $\ADJ[v][u]$ to this new value. However, our more efficient Algorithm~\ref{alg:oblivious-coin-toss} that represents $\ADJ$ compactly via $\LAST$ and $P_v$'s cannot record arbitrary modifications to $\ADJ$.

Observe that if we were to apply the trivial implementation of \func{Vertex-Pair}, then by Lemma~\ref{lem:cond-0}, $\LAST$ and $P_v$'s will only fail capture the state $\ADJ[v][u] = \ZERO$ when $u > \LAST[v]$ and $v > \LAST[u]$. Fortunately, unlike \func{Next-Neighbor} queries, a \func{Vertex-Pair} query can only set one cell $\ADJ[v][u]$ to $\ZERO$ per query, and thus we may afford to store these changes explicitly.\footnote{The disadvantage of this approach is that the implementation may allocate more than $\Theta(m)$ space over the entire graph generation process, if \func{Vertex-Pair} queries generate many of these $\ZERO$'s.} To this end, we define the set $Q = \{\{u,v\}: X_{uv}\textrm{ is assigned to }\ZERO \textrm{ during a \func{Vertex-Pair} query}\}$, maintained as a hash table. Updating $Q$ during \func{Vertex-Pair} queries is trivial: we simply add $\{u,v\}$ to $Q$ before we finish processing the query if we set $\ADJ[u][v]=\ZERO$. Conversely, we need to add $u$ to $P_v$ and add $v$ to $P_u$ if the \func{Vertex-Pair} query sets $\ADJ[u][v]=\ONE$ as usual, yielding the following observation. It is straightforward to verify that each \func{Vertex-Pair} query requires $O(\log n)$ time, $O(1)$ random $N$-bit word, and $O(1)$ additional space per query.

\begin{restatable}{lemma}{cond-q}\label{lem:cond-0-q}
The data structures $\LAST$, $P_v$'s and $Q$ together provide a succinct representation of $\ADJ$ when \func{Next-Neighbor} queries (modified Algorithm~\ref{alg:oblivious-coin-toss}) and \func{Vertex-Pair} queries are allowed. In particular, $\ADJ[v][u]=\ONE$ if and only if $u \in P_v$. Otherwise, $\ADJ[v][u]=\ZERO$ if $u <\LAST[v]$, $v < \LAST[u]$, or $\{v,u\} \in Q$. In all remaining cases, $\ADJ[v][u]=\PHI$.
\end{restatable}

We now explain other necessary changes to Algorithm~\ref{alg:oblivious-coin-toss}. In the implementation of \func{Next-Neighbor}, an iteration is not failed when the chosen $X_{vu}$ is still undecided: $\ADJ[v][u]$ must still be $\phi$. Since $X_{vu}$ may also be assigned to $\ZERO$ via a \func{Vertex-Pair}$(v,u)$ query, we must also consider an iteration where $\{v,u\} \in Q$ failed. That is, we now require one additional condition $\{v,u\} \notin Q$ for termination (which only takes $O(1)$ time to verify per iteration). As for the analysis, aside from handling the fact that $X_{vu}$ may also become decided during a \func{Vertex-Pair} call, and allowing the states of the algorithm to support \func{Vertex-Pair} queries, all of the remaining analysis for correctness and performance guarantee still holds.

Therefore, we have established that our augmentation to Algorithm~\ref{alg:oblivious-coin-toss} still maintains all of its (asymptotic) performance guarantees for \func{Next-Neighbor} queries, and supports \func{Vertex-Pair} queries with complexities as specified above, concluding the following corollary.
We remark that, as we do not aim to support \func{Random-Neighbor} queries, this simple algorithm here provides significant improvement over the performance of \func{Random-Neighbor} queries (given in Corollary~\ref{cor:random_neighbor_time}).

\begin{restatable}{corollary}{res:oblivious-thm}\label{cor:oblivious-alg}
Algorithm~\ref{alg:oblivious-coin-toss} can be modified to allow an implementation of \func{Vertex-Pair} query as explained above, such that the resource usages per query still asymptotically follow those of Theorem~\ref{thm:ER-rand-iterations}.
\end{restatable}

\clearpage
\section{Omitted Details from Section~\ref{sec:undirected}: Undirected Random Graph Implementations}
\label{sec:undirected_omitted}

\subsection{Removing the Perfect-Precision Arithmetic Assumption}
\label{sec:remove-perfect}

In this section we remove the prefect-precision arithmetic assumption. Instead, we only assume that it is possible to compute $\prod_{u=a}^b (1-p_{vu})$ and $\sum_{u=a}^b p_{vu}$ to $N$-bit precision, as well as drawing a random $N$-bit word, using polylogarithmic resources. Here we will focus on proving that the family of the random graph we generate via our procedures is statistically close to that of the desired distribution. The main technicality of this lemma arises from the fact that, not only the implementation is randomized, but the agent interacting with the implementation may choose their queries arbitrarily (or adversarially): our proof must handle any sequence of random choices the implementation makes, and any sequence of queries the agent may make.

Observe that the distribution of the graphs constructed by our implementation is governed entirely by the samples $u$ drawn from $\mathsf{F}(v,a,b)$ in Algorithm~\ref{alg:fill}. By our assumption, the CDF of any $\mathsf{F}(v,a,b)$ can be efficiently computed from $\prod_{u=a}^{u'} (1-p_{vu})$, and thus sampling with $\frac{1}{\poly(n)}$ error in the $L_1$-distance requires a random $N$-bit word and a binary-search in $\bo(\log (b-a+1)) = \bo(\log n)$ iterations. Using this crucial fact, we prove our lemma that removes the perfect-precision arithmetic assumption.


\begin{restatable}{lemma}{transition}\label{lemma:transition}
If Algorithm~\ref{alg:fill} (the \func{Fill} operation) is repeatedly invoked to construct a graph $G$ by drawing the value $u$ for at most $S$ times in total, each of which comes from some distribution $\mathsf{F}'(v,a,b)$ that is $\epsilon$-close in $L_1$-distance to the correct distribution $\mathsf{F}(v,a,b)$ that perfectly generates the desired distribution $\mathsf{G}$ over all graphs, then the distribution $\mathsf{G}'$ of the generated graph $G$ is $(\epsilon S)$-close to $\mathsf{G}$ in the $L_1$-distance.
\end{restatable}
\begin{proof}
\label{proof:transition}
For simplicity, assume that the algorithm generates the graph to completion according to a sequence of up to $n^2$ distinct blocks $\mathcal{B} = \langle B^{(u_1)}_{v_1}, B^{(u_2)}_{v_2}, \ldots \rangle$, where each $B^{(u_i)}_{v_i}$ specifies the \unfilled~block in which any query instigates a \func{Fill} function call. Define an \emph{internal state} of our implementation as the triplet $s = (k, u, \ADJ)$, representing that the algorithm is currently processing the $k^\textrm{th}$ \func{Fill}, in the iteration (the \textbf{repeat} loop of Algorithm~\ref{alg:fill}) with value $u$, and have generated $\ADJ$ so far. Let $t_{\ADJ}$ denote the \emph{terminal state} after processing all queries and having generated the graph $G_\ADJ$ represented by $\ADJ$. We note that $\ADJ$ is used here in the analysis but not explicitly maintained; further, it reflects the changes in every iteration: as $u$ is updated during each iteration of \func{Fill}, the cells $\ADJ[v][u'] = \PHI$ for $u' < u$ (within that block) that has been skipped are also updated to $\ZERO$.

Let $\mathcal{S}$ denote the set of all (internal and terminal) states. For each state $s$, the implementation samples $u$ from the corresponding $\mathsf{F}'(v,a,b)$ where $\|\mathsf{F}(v,a,b)-\mathsf{F}'(v,a,b)\|_1 \leq \epsilon = \frac{1}{\poly(n)}$, then moves to a new state according to $u$. In other words, there is an induced pair of collection of distributions over the states: $(\mathcal{T},\mathcal{T}')$ where $\mathcal{T}=\{\mathsf{T}_s\}_{s\in\mathcal{S}}, \mathcal{T}'=\{\mathsf{T}'_s\}_{s\in\mathcal{S}}$, such that $\mathsf{T}_s(s')$ and $\mathsf{T}'_s(s')$ denote the probability that the algorithm advances from $s$ to $s'$ by using a sample from the correct $\mathsf{F}(v,a,b)$ and from the approximated $\mathsf{F}'(v,a,b)$, respectively. Consequently, $\|\mathsf{T}_s-\mathsf{T}'_s\|_1 \leq \epsilon$ for every $s\in\mathcal{S}$.

The implementation begins with the initial (internal) state $s_0 = (1, 0, \ADJ_\PHI)$ where all cells of $\ADJ_\PHI$ are $\PHI$'s, goes through at most $S=O(n^3)$ other states (as there are up to $n^2$ values of $k$ and $O(n)$ values of $u$), and reach some terminal state $t_\ADJ$, generating the entire graph in the process. Let $\pi = \langle s^\pi_0 = s_0, s^\pi_1, \ldots, s^\pi_{\ell(\pi)} = t_\ADJ \rangle$ for some $\ADJ$ denote a sequence (``path'') of up to $S+1$ states the algorithm proceeds through, where $\ell(\pi)$ denote the number of transitions it undergoes. For simplicity, let $T_{t_\ADJ}(t_\ADJ)=1$, and $T_{t_\ADJ}(s)=0$ for all state $s \neq t_\ADJ$, so that the terminal state can be repeated and we may assume $\ell(\pi) = S$ for every $\pi$. Then, for the correct transition probabilities described as $\mathcal{T}$, each $\pi$ occurs with probability $q(\pi) = \prod_{i=1}^{S} \mathsf{T}_{s_{i-1}}(s_i)$, and thus $\mathsf{G}(G_\ADJ) = \sum_{\pi:s^\pi_{S} = t_\ADJ} q(\pi)$.

Let $\mathcal{T}^{\min}=\{\mathsf{T}^{\min}_s\}_{s\in\mathcal{S}}$ where $\mathsf{T}^{\min}_s(s') = \min\{\mathsf{T}_s(s'),\mathsf{T}'_s(s')\}$, and note that each $\mathsf{T}^{\min}_s$ is not necessarily a probability distribution. Then, $\sum_{s'} \mathsf{T}^{\min}_s(s') = 1 - \|\mathsf{T}_s-\mathsf{T}'_s\|_1 \geq 1-\epsilon$. Define $q', q^{\min}, \mathsf{G}'(G_\ADJ),\mathsf{G}^{\min}(G_\ADJ)$ analogously, and observe that $q^{\min}(\pi) \leq \min\{q(\pi), q'(\pi)\}$ for every $\pi$, so $\mathsf{G}^{\min}(G_\ADJ) \leq \min\{\mathsf{G}(G_\ADJ),\mathsf{G}'(G_\ADJ)\}$ for every $G_\ADJ$ as well. In other words, $q^{\min}(\pi)$ lower bounds the probability that the algorithm, drawing samples from the correct distributions or the approximated distributions, proceeds through states of $\pi$; consequently, $\mathsf{G}^{\min}(G_\ADJ)$ lower bounds the probability that the algorithm generates the graph $G_\ADJ$.

Next, consider the probability that the algorithm proceeds through the prefix $\pi_i = \langle s^\pi_0, \ldots, s^\pi_{i}\rangle$ of $\pi$. Observe that for $i \geq 1$,
\begin{align*}\sum_{\pi} q^{\min}(\pi_i) &=\sum_{\pi} q^{\min}(\pi_{i-1})\cdot \mathsf{T}^{\min}_{s^\pi_{i-1}}(s^\pi_{i}) 
= \sum_{s,s'} \sum_{\pi:s^\pi_{i-1} = s,s^\pi_{i} = s'} q^{\min}(\pi_{i-1})\cdot \mathsf{T}^{\min}_{s}(s') \\
&= \sum_{s'} \mathsf{T}^{\min}_s(s')\cdot\sum_{s} \sum_{\pi:s^\pi_{i-1} = s} q^{\min}(\pi_{i-1})
\geq (1-\epsilon) \sum_{\pi} q^{\min}(\pi_{i-1}).\end{align*}
Roughly speaking, at least a factor of $1-\epsilon$ of the ``agreement'' between the distributions over states according to $\mathcal{T}$ and $\mathcal{T}'$ is necessarily conserved after a single sampling process. As $\sum_{\pi} q^{\min}(\pi_0)=1$ because the algorithm begins with $s_0 = (1, 0, \ADJ_\PHI)$, by an inductive argument we have $\sum_{\pi} q^{\min}(\pi)=\sum_{\pi} q^{\min}(\pi_S) \geq (1-\epsilon)^S \geq 1-\epsilon S$. Hence, $\sum_{G_\ADJ} \min\{\mathsf{G}(G_\ADJ),\mathsf{G}'(G_\ADJ)\} \geq \sum_{G_\ADJ} \mathsf{G}^{\min}(G_\ADJ) \geq 1-\epsilon S$, implying that $\|\mathsf{G}-\mathsf{G}'\|_1 \leq \epsilon S$, as desired. In particular,  by substituting $\epsilon = \frac{1}{\poly(n)}$ and $S = O(n^3)$, we have shown that Algorithm~\ref{alg:fill} only creates a $\frac{1}{\poly(n)}$ error in the $L_1$-distance. 
\end{proof}

We remark that \func{Random-Neighbor} queries also require that the returned edge is drawn from a distribution that is close to a uniform one, but this requirement applies only \emph{per query} rather then over the entire execution of the generator. Hence, the error due to the selection of a random neighbor may be handled separately from the error for generating the random graph; its guarantee follows straightforwardly from a similar analysis.

\subsection{Bounding Block Sizes}\label{sec:bounding_block_sizes}
\MaxBlockSize*
\begin{proof}
Fix a block $B_v^{(i)}$, and consider the Bernoulli RVs $\left\{ X_{vu}\right\}_{u\in B_v^{(i)}}$.
The expected number of neighbors in this block is
$ \textstyle\mathbb{E} \left[ |\Gamma^{(i)}(v)| \right] =\mathbb{E} \left[ \sum_{u\in B_v^{(i)}} X_{vu} \right] < L+1$.
Via the Chernoff bound,
\[
\mathbb{P} \left[ |\Gamma^{(i)}(v)|> (1+3c\log n)\cdot L \right]
\le e^{-\frac{3c\log n\cdot L}{3}} = n^{-\Theta(c)}
\]
for any constant $c > 0$.
\end{proof}

\EmptyBlock*
\begin{proof}
For $i < |B_v|$, since $ \mathbb{E} \left[ |\Gamma^{(i)}(v)| \right] =\mathbb{E} \left[ \sum_{u\in B_v^{(i)}} X_{vu} \right] > L-1$, we bound the probability that $B_v^{(i)}$ is empty:
\[
\mathbb P[B_v^{(i)}\textrm{ is empty}] = \prod_{u\in B_v^{(i)}} (1-p_{vu}) \leq e^{-\sum_{u\in B_v^{(i)}} p_{vu}} \leq e^{1-L}=c
\]
for any arbitrary small constant $c$ given sufficienty large constant $L$. Let $T_{i}$ be the indicator for the event that $B_v^{(i)}$ is \emph{not} empty, so $\mathbb E 1-c$. By the Chernoff bound, the probability that less than $|B_v|/3$ blocks are non-empty is 
\[
\textstyle
\mathbb P\left[\sum_{i\in[|B_v|]} T_i<\frac{|B_v|}{3}\right]<\mathbb P\left[\sum_{i\in[|B_v|-1]} T_i<\frac{|B_v-1|}{2}\right]\leq e^{-\Theta(|B_v|-1)} = n^{-\Omega(1)}
\] as $|B_v| = \Omega(\log n)$ by assumption.
\end{proof}

\clearpage
\section{\func{Next-Neighbor} Implementation with Deterministic Performance Guarantee}
\label{sec:ER-det}
In this section, we construct data structures that allow us to sample for the next neighbor directly by considering only the cells $\ADJ[v][u]=\PHI$ in the Erd\"{o}s-R\'{e}nyi model and the Stochastic Block model. This provides $\poly(\log n)$ \emph{worst-case} performance guarantee for implementations supporting only the \func{Next-Neighbor} queries. We may again extend this data structure to support \func{Vertex-Pair} queries, however, at the cost of providing $\poly(\log n)$ \emph{amortized} performance guarantee instead.

In what follows, we first focus on the $G(n,p)$ model, starting with \func{Next-Neighbor} queries (Section~\ref{sec:det-er}) then extend to \func{Vertex-Pair} queries (Section~\ref{sec:det-er-pair}. We then explain how this result may be generalized to support the Stochastic Block model with random community assignment in Section~\ref{sec:det-sbm}.

\subsection{Data structure for next-neighbor queries in the Erd\"{o}s-R\'{e}nyi model}\label{sec:det-er}

\begin{wrapfigure}[15]{r}{0.48\textwidth}
\vspace{-2.5em}
\begin{framed}
    \renewcommand\figurename{Algorithm}
    \caption{Alternate implementation}
    \label{alg:exact-coin-toss}
    \begin{algorithmic}
        \Procedure{Next-Neighbor}{$v$}
            \State{$w \gets \min K_v$, or $n+1$ if $K_v = \emptyset$}
            \State{$t \gets$ \func{count}$(v)$}
            \State{\textbf{sample} $F\sim\mathsf{ExactF}(p,t)$}
            \If{$F \leq t$}
                \State{$u \gets$ \func{pick}$(v,F)$}
                \State{$K_u \gets K_u \cup \{v\}$}
            \Else
                \State{$u \gets w$}
                \If{$u \neq n+1$}
                    \State{$K_v \gets K_v \setminus \{u\}$}
                \EndIf
            \EndIf
            \State{\func{update}$(v,u)$}
            \State{$\LAST[v] \gets u$}
            \State \Return $u$
        \EndProcedure
    \end{algorithmic}
\end{framed}
\end{wrapfigure}

Recall that \func{Next-Neighbor}$(v)$ is given by $\min\{u > \LAST[v]: X_{vu} = 1\}$ (or $n+1$ if no satisfying $u$ exists). To aid in computing this quantity, we define:
\begin{align*}
K_v &= \{u \in (\LAST[v], n]: \ADJ[v][u]=1\},\\
w_v &= \min K_v \textrm{, or $n+1$ if $K_v = \emptyset$,} \\
T_v &= \{u \in (\LAST[v], w_v): \ADJ[v][u] = \PHI\}.
\end{align*}
The ordered set $K_v$ is only defined for ease of presentation: it is equivalent to $(\LAST[v],n] \cap P_v$, recording the known neighbors of $v$ after $\LAST[v]$ (i.e., those that have not been returned as an answer by any \func{Next-Neighbor}$(v)$ query yet). The quantity $w_v$ remains unchanged but is simply restated in terms of $K_v$. $T_v$ specifies the list of candidates $u$ for \func{Next-Neighbor}$(v)$ with $\ADJ[v][u] = \PHI$; in particular, all candidates $u$'s, such that the corresponding RVs $X_{vu} = \ZERO$ are decided, are explicitly excluded from $T_v$.

Unlike the approach of Algorithm~\ref{alg:oblivious-coin-toss} that simulates coin-flips even for decided $X_{vu}$'s, here we only flip undecided coins for the indices in $T_v$: we have $|T_v|$ Bernoulli trials to simulate. Let $F$ be the random variable denoting the first index of a successful trial out of $|T_v|$ coin-flips, or $|T_v|+1$ if all fail; denote the distribution of $F$ by $\mathsf{ExactF}(p,|T|)$. The CDF of $F$ is given by $\mathbb P[F = f] = 1-(1-p)^f$ for $f \leq |T_v|$ (i.e., there is some success trial in the first $f$ trials), and $\mathbb P[F = |T_v|+1] = 1$. Thus, we must design a data structure that can compute $w_v$, compute $|T_v|$, find the $F^\textrm{th}$ minimum value in $T_v$, and update $\ADJ[v][u]$ for the $F$ lowest values $u \in T_v$ accordingly.

Let $k = \lceil \log n \rceil$. We create a range tree, where each node itself contains a balanced binary search tree (BBST), storing $\LAST$ values of its corresponding range. Formally, for $i \in [0, n/2^j)$ and $j \in [0, k]$, the $i^\textrm{th}$ node of the $j^\textrm{th}$ level of the range tree, stores $\LAST[v]$ for every $v \in (i \cdot 2^{k-j}, (i+1)\cdot 2^{k-j}]$. Denote the range tree by $\mathbf{R}$, and each BBST corresponding to the range $[a, b]$ by $\mathbf{B}_{[a,b]}$. We say that the range $[a,b]$  is \emph{canonical} if it corresponds to a range of some $\mathbf{B}_{[a,b]}$ in $\mathbf{R}$.

Again, to allow fast initialization, we make the following adjustments from the given formalization above: (1) values $\LAST[v] = 0$ are never stored in any $\mathbf{B}_{[a,b]}$, and (2) each $\mathbf{B}_{[a,b]}$ is created on-the-fly during the first occasion it becomes non-empty. Further, we augment each $\mathbf{B}_{[a,b]}$ so that each of its node maintains the size of the subtree rooted at that node: this allows us to count, in $O(\log n)$ time, the number of entries in $\mathbf{B}_{[a,b]}$ that is no smaller than a given threshold.

Observe that each $v$ is included in exactly one $\mathbf{B}_{[a,b]}$ per level in $\mathbf{R}$, so $k+1=O(\log n)$ copies of $\LAST[v]$ are stored throughout $\mathbf{R}$. Moreover, by the property of range trees, any interval can be decomposed into a disjoint union of $O(\log n)$ canonical ranges. From these properties we implement the data structure $\mathbf{R}$ to support the following operations. (Note that $\mathbf{R}$ is initially an empty tree, so initialization is trivial.)
\begin{itemize}
\item \func{count}$(v)$: compute $|T_v|$. \\
We break $(\LAST[v],w_v)$ into $O(\log n)$ disjoint canonical ranges $[a_i, b_i]$'s each corresponding to some $\mathbf{B}_{[a_i,b_i]}$, then compute $t_{[a_i,b_i]} =|\{u \in [a_i, b_i]: \LAST[u] < v\}|$, and return $\sum_i t_{[a_i,b_i]}$. The value $t_{[a_i, b_i]}$ is obtained by counting the entries of $\mathbf{B}_{[a_i, b_i]}$ that is at least $v$, then subtract it from $b_i-a_i+1$; we cannot count entries less than $v$ because $\LAST[u]=0$ are not stored.
\item \func{pick}$(v,F)$: find the $F^\textrm{th}$ minimum value in $T_v$ (assuming $F \leq |T_v|$). \\
We again break $(\LAST[v],w_v)$ into $O(\log n)$ canonical ranges $[a_i, b_i]$'s, compute $t_{[a_i, b_i]}$'s, and identify the canonical range $[a^*,b^*]$ containing the $i^\textrm{th}$ smallest element (i.e., $[a_i, b_i]$ with the smallest $b$ satisfying $\sum_{j \leq i} t_{[a_j,b_j]} \geq F$ assuming ranges are sorted). Binary-search in $[a^*,b^*]$ to find exactly the $i^\textrm{th}$ smallest element of $T$. This is accomplished by traversing $\mathbf{R}$ starting from the range $[a^*,b^*]$ down to a leaf, at each step computing the children's $T_{[a,b]}$'s and deciding which child's range contains the desired element.
\item \func{update}$(v,u)$: simulate coin-flips, assigning $X_{vu} \leftarrow 1$, and $X_{v,u'} \leftarrow 0$ for $u' \in (\LAST[v], u) \cap T_v$. \\
This is done implicitly by handling the change $\LAST[v] \leftarrow u$: for each BBST $\mathbf{B}_{[a,b]}$ where $v \in [a, b]$, remove the old value of $\LAST[v]$ and insert $u$ instead.
\end{itemize}
It is straightforward to verify that all operations require at most $O(\log^2 n)$ time and $O(\log n)$ additional space per call. The overall implementation is given in Algorithm~\ref{alg:exact-coin-toss}, using the same asymptotic time and additional space. Recall also that sampling $F\sim\mathsf{ExactF}(p,t)$ requires $O(\log n)$ time and one $N$-bit random word for the $G(n,p)$ model.

\subsection{Data structure for \func{Vertex-Pair} queries in the Erd\"{o}s-R\'{e}nyi model}\label{sec:det-er-pair}
Recall that we define $Q$ in Algorithm~\ref{alg:oblivious-coin-toss} as the  set of pairs $(u,v)$ where $X_{uv}$ is assigned to $\ZERO$ during a \func{Vertex-Pair} query, allowing us to check for modifications of $\ADJ$ not captured by $\LAST[v]$ and $K_v$. Here in Algorithm~\ref{alg:exact-coin-toss}, rather than checking, we need to be able to count such entries. Thus, we instead create a BBST $Q'_v$ for each $v$ defined as:
\[Q'_v = \{u: u > \LAST[v], v > \LAST[u], \textrm{ and } X_{uv}\textrm{ is assigned to }\ZERO \textrm{ during a \func{Vertex-Pair} query}\}.\]
This definition differs from that of $Q$ in Section~\ref{sec:ER-pair} in two aspects. First, we ensure that each $\ADJ[v][u] = \ZERO$ is recorded by either $\LAST$ (via Lemma~\ref{lem:cond-0}) or $Q'_v$ (explicitly), but \emph{not both}. In particular, if $u$ were to stay in $Q'_v$ when $\LAST[v]$ increases beyond $u$, we would have double-counted these entries $\ZERO$ not only recorded by $Q'_v$ but also implied by $\LAST[v]$ and $K_v$. By having a BBST for each $Q'_v$, we can compute the number of $\ZERO$'s that must be excluded from $T_v$, which cannot be determined via $\LAST[v]$ and $K_v$ alone: we subtract these from any counting process done in the data structure $\mathbf{R}$.

Second, we maintain $Q'_v$ separately for each $v$ as an ordered set, so that we may identify non-neighbors of $v$ within a specific range -- this allows us to remove non-neighbors in specific range, ensuring that the first aspect holds. More specifically, when we increase $\LAST[v]$, we must go through the data structure $Q'_v$ and remove all $u < \LAST[v]$, and for each such $u$, also remove $v$ from $Q'_u$. There can be as many as linear number of such $u$, but the number of removals is trivially bounded by the number of insertions, yielding an amortized time performance guarantee in the following theorem. Aside from the deterministic guarantee, unsurprisingly, the required amount of random words for this algorithm is lower than that of the algorithm from Section~\ref{sec:reroll-cont} (given in Theorem~\ref{thm:ER-rand-iterations} and Corollary~\ref{cor:oblivious-alg}).

\begin{theorem}
Consider the Erd\"{o}s-R\'{e}nyi $G(n,p)$ model. For \func{Next-Neighbor} queries only, Algorithm~\ref{alg:exact-coin-toss} is an implementation that answers each query using $O(\log^2 n)$ time, $O(\log n)$ additional space, and one $N$-bit random word. For \func{Next-Neighbor} and \func{vertex pair} queries,an extension of Algorithm~\ref{alg:exact-coin-toss} answers each query using $O(\log^2 n)$ amortized time, $O(\log n)$ additional space, and one $N$-bit random word.
\end{theorem}

\subsection{Data structure for the Stochastic Block model}\label{sec:det-sbm}

We employ the data structure for generating and counting the number of vertices of each community in a specified range from Section~\ref{sec:application_sbm}. We create $r$ different copies of the data structure $\mathbf{R}$ and $Q'_v$, one for each community, so that we may implement the required operations separately for each color, including using the \func{count} subroutine to sample $F\sim\mathsf{ExactF}$ via the corresponding CDF, and picking the next neighbor according to $F$. Recall that since we do not store $\LAST[v] = 0$ in $\mathbf{R}$, and we only add an entry to $K_v$, $P_v$ or $Q'_v$ after drawing the corresponding $X_{uv}$, the communities of the endpoints, which cover all elements stored in these data structures, must have already been determined. Thus, we obtain the following corollary for the Stochastic Block model.

\begin{restatable}{corollary}{res:sbm-corol}
Consider the Stochastic Block model with randomly-assigned communities. For \func{Next-Neighbor} queries only, Algorithm~\ref{alg:exact-coin-toss} is an implementation that answers each query using $O(r\,\poly(\log n))$ time, random words, and additional space per query. For \func{Next-Neighbor} and \func{Vertex-Pair} queries, Algorithm~\ref{alg:exact-coin-toss} answers each query using $O(r\,\poly(\log n))$ amortized time, $O(r\,\poly(\log n))$ random words, and $O(r\,\poly(\log n))$ additional space per query additional space, and one $N$-bit random word.
\end{restatable}

\clearpage
\section{Sampling from the Multivariate Hypergeometric Distribution}
\label{sec:multivariate_hypergeometric_sampling}

Consider the following random experiment. Suppose that we have an urn containing $B \leq n$ marbles (representing vertices), each occupies one of the $r$ possible colors (representing communities) represented by an integer from $[r]$. The number of marbles of each color in the urn is known: there are $C_k$ indistinguishable marbles of color $k \in [r]$, where $C_1 + \cdots + C_r = B$. Consider the process of drawing $\ell \leq B$ marbles from this urn \emph{without replacement}. We would like to sample how many marbles of each color we draw.

More formally, let $\vec{C} = \langle c_1, \ldots, c_r \rangle$, then we would like to (approximately) sample a vector $\mathbf{S}^\vec{C}_\ell$ of $r$ non-negative integers such that
\[\Pr[\mathbf{S}^\vec{C}_\ell = \langle s_1, \ldots, s_r \rangle]
= \frac{{C_1\choose s_1}\cdot{C_2\choose s_2}\cdots{C_r\choose s_r}}{{B \choose C_1+C_2+ \cdots +C_r}}\]

where the distribution is supported by all vectors satisfying $s_k \in \{0, \ldots, C_k\}$ for all $k \in [r]$ and $\sum_{k=1}^{r} s_k = \ell$. This distribution is referred to as the \emph{multivariate hypergeometric distribution}.

The sample $\mathbf{S}^\vec{C}_\ell$ above may be generated easily by simulating the drawing process, but this may take $\Omega(\ell)$ iterations, which have linear dependency in $n$ in the worst case: $\ell = \Theta(B) = \Theta(n)$. Instead, we aim to generate such a sample in $O(r\,\poly(\log n))$ time with high probability. We first make use of the following procedure from \cite{huge}.

\begin{restatable}{lemma}{res:ggn-marble}\label{lem:ggn_interval_summable}
Suppose that there are $T$ marbles of color $1$ and $B-T$ marbles of color $2$ in an urn,
where $B \leq n$ is even. There exists an algorithm that samples $\langle s_1, s_2 \rangle$,
the number of marbles of each color appearing when drawing $B/2$ marbles from the urn without replacement,
in $O(\poly(\log n))$ time and random words.
Specifically, the probability of sampling a specific pair $\langle s_1, s_2 \rangle$ where $s_1 + s_2 = T$
is approximately ${B/2 \choose s_1}{B/2 \choose T-s_1}/{B \choose T}$ with error of at most $n^{-c}$ for any constant $c>0$.
\end{restatable}

In other words, the claim here only applies to the two-color case,
where we sample the number of marbles when drawing exactly half of the marbles from the entire urn ($r=2$ and $\ell = B/2$).
First we generalize this claim to handle any desired number of drawn marbles $\ell$ (while keeping $r=2$).

\begin{restatable}{lemma}{res:new-marble2}\label{thm:sampling_two_colors}
Given $C_1$ marbles of color $1$ and $C_2 = B-C_1$ marbles of color $2$,
there exists an algorithm that samples $\langle s_1, s_2 \rangle$,
the number of marbles of each color appearing when drawing $\ell$ marbles from the urn without replacement,
in $O(\poly(\log B))$ time and random words.
\end{restatable}
\begin{proof}
For the base case where $B=1$, we trivially have $\mathbf{S}^\vec{C}_1=\vec{C}_1$ and $\mathbf{S}^\vec{C}_0=\vec{C}_2$.
Otherwise, for even $B$, we apply the following procedure.
\begin{itemize}
\item If $\ell \leq B/2$, generate $\vec{C}'=\mathbf{S}^\vec{C}_{B/2}$ using Lemma~\ref{lem:ggn_interval_summable}.
\begin{itemize}
\item If $\ell = B/2$ then we are done.
\item Else, for $\ell < B/2$ we recursively generate $\mathbf{S}^\vec{C'}_{\ell}$.
\end{itemize}
\item Else, for $\ell > B/2$, we generate $\mathbf{S}^\vec{C'}_{B-\ell}$ as above, then output $\vec{C}-\mathbf{S}^\vec{C'}_{B-\ell}$.
\end{itemize}
On the other hand, for odd $B$, we simply simulate drawing a single random marble
from the urn before applying the above procedure on the remaining $B-1$ marbles in the urn.
That is, this process halves the domain size $B$ in each step, requiring $\log B$ iterations to sample $\mathbf{S}^\vec{C}_\ell$.
\end{proof}

Lastly we generalize to support larger $r$.
\SamplingManyColors*
\begin{proof}
Observe that we may reduce $r>2$ to the two-color case by sampling the number of marbles of the first color,
collapsing the rest of the colors together.
Namely, define a pair $\vec D=\langle C_1, C_2+\cdots+C_r \rangle$,
then generate $\mathbf{S}^{\vec D}_{\ell}=\langle s_1, s_2+\ldots+s_r\rangle$ via the above procedure.
At this point we have obtained the first entry $s_1$ of the desired $\mathbf{S}^{\vec{C}}_{\ell}$.
So it remains to generate the number of marbles of each color from the remaining $r-1$ colors in $\ell-s_1$ remaining draws.
In total, we may generate $\mathbf{S}^{\vec{C}}_{\ell}$ by performing $r$ iterations of the two-colored case.
The error in the $L_1$-distance may be established similarly to the proof of Lemma~\ref{lemma:transition}.
\end{proof}
\begin{theorem}
\label{thm:sampling_many_colors_contiguous}
Given $B$ marbles of $r$ different colors in $[r]$, such that there are $C_i$ marbles of color $i$ and a parameter $k\le r$,
there exists an algorithm that samples $s_1 + s_2 +\cdots + s_k$,
the number of marbles among the first $k$ colors appearing when drawing $\ell$ marbles from the urn without replacement,
in $O(\poly(\log B))$ time and random words.
\end{theorem}
\begin{proof}
Since we don't have to find the individual counts, we can be more efficient by grouping half the colors together at each step.
Formally, we define a pair $\vec  D= \langle D_1,D_2 \rangle$ where $D_1=C_1 + C_2 +\cdots + C_{r/2}$ and $D_2=C_{r/2+1}+\cdots+C_{r-1}+C_r$.
We then generate $ \langle D_1', D_2'\rangle = \mathbf S^{\vec D}_\ell$.
\begin{itemize}
    \item If $k < r/2$, we recursively solve the problem with the first $r/2$ colors, $B\gets D_1'$, and the original value of $k$.
    \item If $k > r/2$, we recurse on the last $r/2$ colors, $B$ set to $D_2'$, and $k$ set to $k-r/2$.
    In this case, we add $D_1'$ to the returned value.
    \item Otherwise, $k=r/2$ and we can return $D_1'$.
\end{itemize}
The number of recursive calls is $\mathcal O(\log r) = \mathcal O(\log B)$ (since $r\le B$).
So, the overall runtime is $\mathcal O(\poly(\log B))$.
\end{proof}

\clearpage
\section{Local-Access Implementations for Random Directed Graphs}
\label{sec:small_world}

In this section, we consider Kleinberg's Small-World model \cite{kleinberg, klein}
where the probability that a \emph{directed} edge $(u,v)$ exists is $\min\{c/(\func{dist}(u,v))^2, 1\}$.
Here, $\func{dist}(u,v)$ is the Manhattan distance between $u$ and $v$ on a $\sqrt n\times\sqrt n$ grid.
We begin with the case where $c = 1$, then generalize to different values of $c = \log^{\pm\Theta(1)}(n)$. 
We aim to support $\func{All-Neighbors}$ queries using $\poly(\log n)$ resources. 
This returns the entire list of out-neighbors of $v$.

\subsection{Implementation for $c=1$}

Observe that since the graphs we consider here are directed, the answers to the $\func{All-Neighbor}$ queries are all independent: each vertex may determine its out-neighbors independently.
Given a vertex $v$, we consider a partition of all the other vertices of the graph into sets $\{\Gamma^v_1, \Gamma^v_2,\ldots\}$ by distance: $\Gamma^v_k = \{u: \func{dist}(v,u) = k\}$ contains all vertices at a distance $k$ from vertex $v$. Observe that $|\Gamma^v_k|\leq 4k = O(k)$. Then, the expected number of edges from $v$ to vertices in $\Gamma^v_k$ is therefore $|\Gamma^v_k|\cdot 1/k^2 = O(1/k)$.
Hence, the expected degree of $v$ is at most $\sum_{k=1}^{2(\sqrt{n}-1)}O(1/k) = O(\log n)$.
It is straightforward to verify that this bound holds with high probability (use Hoeffding's inequality).
Since the degree of $v$ is small, in this model we can afford to perform \func{All-Neighbors} queries instead of \func{Next-Neighbor} queries using an additional $\poly(\log n)$ resources.

Nonetheless, internally in our implementation, we generate our neighbors one-by-one similarly to how we process \func{Next-Neighbor} queries.
We perform our sampling in two phases.
In the first phase, we compute a distance $d$, such that the next neighbor closest to $v$ is at distance $d$.
We maintain $\LAST[v]$ to be the last computed distance.
In the second phase, we generate all neighbors of $v$ at distance $d$, under the assumption that there must be at least one such neighbor.
For simplicity, we generate these neighbors as if there are \emph{full} $4d$ vertices at distance $d$ from $v$:
some generated neighbors may lie outside our $\sqrt n\times\sqrt n$ grid, which are simply discarded.
As the running time of our implementation is proportional to the number of implementation neighbors,
then by the bound on the number of neighbors, this assumption does not asymptotically worsen the performance of the implementation.

\subsubsection{Phase 1: Generate the distance $D$}
Let $a = \LAST[v] + 1$, and let $\mathsf{D}(a)$ to denote the probability distribution of the distance where the next closest neighbor of $v$ is located, or $\bot$ if there is no neighbor at distance at most $2(\sqrt{n}-1)$.
That is, if $D\sim\mathsf{D}(a)$ is drawn, then we proceed to Phase 2 to generate all neighbors at distance $D$.
We repeat the process by sampling the next distance from $\mathsf{D}(a+D)$ and so on until we obtain $\bot$, at which point we return our answers and terminate.

To generate the next distance, we perform a binary search: we must evaluate the CDF of $\mathsf{D}(a)$.
The CDF is given by $\mathbb P[D\leq d]$ where $D\sim\mathsf{D}(a)$, the probability that there is \emph{some} neighbor at distance at most $d$.
As usual, we compute the probability of the negation: there is \emph{no} neighbor at distance at most $d$.
Recall that each distance $i$ has exactly $|\Gamma_i^v| = 4i$ vertices, and the probability of a vertex $u \in \Gamma_i^v$ is not a neighbor is exactly $1-1/i^2$.
So, the probability that there is no neighbor at distance $i$ is $(1-1/i^2)^{4i}$.
Thus, for $D\sim\mathsf{D}(a)$ and $d \leq 2(\sqrt{n}-1)$,
\begin{align*}
\mathbb P[D\leq d] &= 1-\prod_{i=a}^{d} \left(1-\frac{1}{i^2}\right) = 1-\prod_{i=a}^{d} \left(\frac{(i-1)(i+1)}{i^2}\right)^{4i}
=1-\left(\frac{(a-1)^{a}}{a^{a-1}}\cdot\frac{(d+1)^{d}}{d^{d+1}}\right)^4
\end{align*}
where the product enjoys telescoping as the denominator $(i^2)^{4i}$ cancels with $(i^2)^{4(i-1)}$ and $(i^2)^{4(i+1)}$ in the numerators of the previous and the next term, respectively.
This gives us a closed form for the CDF, which we can compute with $2^{-N}$ additive error in constant time (by our computation model assumption).
Thus, we may generate the distance $D\sim\mathsf{D}(a)$ using $O(\log n)$ time and one random $N$-bit word.

\subsubsection{Phase 2: Sampling neighbors at distance $D$}
After sampling a distance $D$, we now have to generate all the neighbors at distance $D$.
We label the vertices in $\Gamma_D^v$ with unique indices in $\{1, \ldots, 4D\}$.
Note that now each of the $4D$ vertices in $\Gamma_D^v$ is a neighbor with probability $1/D^2$.
However, by Phase 1, this is conditioned on the fact that there is at least one neighbor among the vertices in $\Gamma_D^v$,
which may be difficult to generate when $1/D^2$ is very small.
We can emulate this na\"{i}vely by repeatedly sampling a ``block'', composing of the $4D$ vertices in $\Gamma_D^v$, by deciding whether each vertex is a neighbor of $v$ with uniform probability $1/D^2$ (i.e., $4D$ identical independent Bernoulli trials), and then discarding the entire block if it contains no neighbor. We repeat this process until we finally generate one block that contains at least one neighbor, and use this block as our output.

For the purpose of making the sampling process more efficient, we view this process differently. Let us imagine that we are given an infinite sequence of independent Bernoulli variables, each with bias $1/D^2$.
We then divide the sequence into contiguous blocks of length $4D$ each.
Our task is to find the \emph{first} occurrence of success (a neighbor), then report the whole block hosting this variable.

This first occurrence of a successful Bernoulli trial is given by sampling from the geometric distribution, $X\sim\mathsf{Geo}(1/D^2)$.
Since the vertices in each block are labeled by $1, \ldots, 4D$, then this first occurrence has label $X' = {X\mathrm{\,mod\,}4D}$.
By sampling $X\sim\mathsf{Geo}(1/D^2)$, the first $X'$ Bernoulli variables of this block is also implicitly determined. Namely, the vertices of labels $1, \ldots, X'-1$ are non-neighbors, and that of label $X'$ is a neighbor.
The sampling for the remaining $4D-X'$ vertices can then be performed in the same fashion we generate next neighbors in the $G(n,p)$ case: 
repeatedly find the next neighbor by sampling from $\mathsf{Geo}(1/D^2)$, until the index of the next neighbor falls beyond this block.

Thus at this point, we have generated all neighbors in $\Gamma_D^v$. We can then update $\LAST[v] \leftarrow D$ and continue the process of larger distances.
Sampling each neighbor takes $O(\log n)$ time and one random $N$-bit word; the resources spent sampling the distances is also bounded by that of the neighbors.
As there are $O(\log n)$ neighbors with high probability, we obtain the following theorem.

\begin{theorem}\label{thm-swm}
There exists an algorithm that generates a random graph from Kleinberg's Small World model,
where probability of including each directed edge $(u,v)$ in the graph is $1/(\func{dist}(u,v))^2$ where $\func{dist}$ denote the Manhattan distance,
using $O(\log^2 n)$ time and random $\log n$-bit words per \func{All-Neighbors} query with high probability.
\end{theorem}

\subsection{Implementation for $c \neq 1$}

Observe that to support different values of $c$ in the probability function $c/(\func{dist}(u,v))^2$, we do not have a closed-form formula for computing the CDF for Phase 1, whereas the process for Phase 2 remains unchanged. To handle the change in the probability distribution Phase 1, we consider the following, more general problem. Suppose that we have a process $P$ that, one-by-one, provide occurrences of successes from the sequence of independent Bernoulli trials with success probabilities $\langle p_1, p_2, \ldots \rangle$. We show how to construct a process $\mathcal{P}^c$ that provide occurrences of successes from Bernoulli trials with success probabilities $\langle c\cdot p_1, c\cdot p_2, \ldots\rangle$ (truncated down to $1$ as needed). For our application, we assume that $c$ is given in $N$-bit precision, there are $O(n)$ Bernoulli trials, and we aim for an error of $\frac{1}{\poly(n)}$ in the $L_1$-distance.

\subsubsection{Case $c < 1$}
We use rejection sampling in order to construct a new Bernoulli process.

\begin{restatable}{lemma}{res:small-c}
Given a process $\mathcal{P}$ outputting the indices of successful Bernoulli trials with bias $\langle p_i\rangle$, there exists a process $\mathcal{P}^c$ outputting the indices of successful Bernoulli trials with bias $\langle c\cdot p_i\rangle$ where $c<1$,
using one additional $N$-bit word overhead for each answer of $\mathcal{P}$.
\end{restatable}
\begin{proof}
Consider the following rejection sampling process to simulate the Bernoulli trials.
In addition to each Bernoulli variable $X_i$ with bias $p_i$, we generate another coin-flip $C_i$ with bias $c$.
Set $Y_i = X_i \cdot C_i$, then $\mathbb P[Y_i = 1] = \mathbb P[X_i = 1]\cdot\mathbb P[C_i] = c\cdot p_i$, as desired.
That is, we keep a success of a Bernoulli trial with probability $c$, or reject it with probability $1-c$.

Now, we are already given the process $\mathcal{P}$ that ``handles'' $X_i$'s, generating a sequence of indices $i$ with $X_i = 1$.
The new process $\mathcal{P}^c$ then only needs to handle the $C_i$'s. Namely, for each $i$ reported as success by $\mathcal{P}$, $\mathcal{P}^c$ flips a coin $C_i$ to see if it should also report $i$, or discard it.
As a result, $\mathcal{P}^c$ can generate the indices of successful Bernoulli trials using only one random $N$-bit word overhead for each answer from $\mathcal{P}$.
\end{proof}
Applying this reduction to the distance sampling in Phase 1, we obtain the following corollary.
\begin{restatable}{corollary}{res:small-c-corol}
There exists an algorithm that generates a random graph from Kleinberg's Small World model with edge probabilities $c/(\func{dist}(u,v))^2$ where $c<1$,
using $O(\log^2 n)$ time and random $\log n$-bit words per \func{All-Neighbors} query with high probability.
\end{restatable}

\subsubsection{Case $c > 1$}

Since we aim to sample with larger probabilities, we instead consider making $k\cdot c$ independent copies of each process $\mathcal{P}$, where $k>1$ is a positive integer.
Intuitively, we hope that the probability that one of these process returns an index $i$ will be at least $c\cdot p_i$, so that we may perform rejection sampling to decide whether to keep $i$ or not.
Unfortunately such a process cannot handle the case where $c\cdot p_i$ is large, notably when $c\cdot p_i > 1$ is truncated down to $1$, while there is always a possibility that none of the processes return $i$.

\begin{restatable}{lemma}{rem:large-c}
Let $k>1$ be a constant integer. Given a process $\mathcal{P}$ outputting the indices of successful Bernoulli trials with bias $\langle p_i\rangle$, there exists a process $\mathcal{P}^c$ outputting the indices of successful Bernoulli trials with bias $\langle \min\{c\cdot p_i, 1\}\rangle$ where $c>1$ and $c \cdot p_i \leq 1-\frac{1}{k}$ for every $i$,
using one additional $N$-bit word overhead for each answer of $k\cdot c$ independent copies of $\mathcal{P}$.
\end{restatable}
\begin{proof}
By applying the following form of Bernoulli's inequality, we have
\begin{align*}
(1-p_{i})^{k\cdot c} \leq 1-\frac{k\cdot c\cdot p_{i}}{1+(k\cdot c-1)\cdot p_{i}}
= 1-\frac{k\cdot c\cdot p_{i}}{1+k\cdot c\cdot p_{i}-p_{i}}
\leq 1-\frac{k\cdot c\cdot p_{i}}{1+(k-1)} = 1-c\cdot p_{i}
\end{align*}
That is, the probability that at least one of the implementations report an index $i$ is $1-(1-p_{i})^{k\cdot c} \geq c\cdot p_i$, as required.
Then, the process $\mathcal{P}^c$ simply reports $i$ with probability $(c\cdot p_i) / (1-(1-p_{i})^{k\cdot c})$ or discard $i$ otherwise.
Again, we only require $N$-bit of precision for each computation, and thus one random $N$-bit word suffices.
\end{proof}

In Phase 1, we may apply this reduction only when the condition $c \cdot p_i \leq 1-\frac{1}{k}$ is satisfied.
For lower value of $p_i = 1/D^2$, namely for distance $D < \sqrt{c/(1-1/k)} = O(\sqrt{c})$, we may afford to generate the Bernoulli trials one-by-one as $c$ is $\poly(\log n)$.
We also note that the degree of each vertex is clearly bounded by $O(\log n)$ with high probability, as its expectation is scaled up by at most a factor of $c$.
Thus, we obtain the following corollary.
\begin{restatable}{corollary}{res:large-c-corol}
There exists an algorithm that generates a random graph from Kleinberg's Small World model with edge probabilities $c/(\func{dist}(u,v))^2$ where $c = \poly(\log n)$,
using $O(\log^2 n)$ time and random words per \func{All-Neighbors} query with high probability.
\end{restatable}

\clearpage
\section{Omitted Proofs for the Dyck Path Implementation}%
\label{sec:dyck_appendix}

\begin{theorem}
\label{thm:number_of_dyck_paths}
There are $\frac{1}{n+1}\binom{2n}{n}$ Dyck paths for length $2n$ (construction from \cite{catalan_book}).
\end{theorem}
\begin{proof}[Proof from \cite{catalan_book}]
Consider all possible sequences containing $n+1$ up-steps and $n$ down-steps with the restriction that the first step is an up-step.
We say that two sequences belong to the same \emph{class} if they are cyclic shifts of each other.
Because of the restriction, the total number of sequences is $\binom{2n}{n}$ and each class is of size $n+1$.
Now, within each class, exactly one of the sequences is such that the prefix sums are \emph{strictly greater} than zero.
From such a sequence, we can obtain a Dyck sequence by deleting the first up-step.
Similarly, we can start with a Dyck sequence, add an initial up-step and consider all $n+1$ cyclic shifts to obtain a \emph{class}.
This bijection shows that the number of Dyck paths is $ \frac{1}{n+1} \binom{2n}{n}$.
\end{proof}

\subsection{Approximating Close-to-Central Binomial Coefficients}%
\label{sec:approximating_close_to_central_binomial_coefficients}
We start with Stirling's approximation which states that
\[
m! = \sqrt{2\pi m}\left( \frac me\right)^m\left( 1 + \mathcal O\left( \frac 1m\right)\right)\\
\]
We will also use the logarithm approximation when a better approximation is required:
\begin{align}
    \label{eq:log_factorial_approximation}
\log (m!) = m\log m -m + \frac 12 \log(2\pi m) + \frac{1}{12 m} - \frac{1}{360 m^3} + \frac{1}{1260 m^5} - \cdots
\end{align}

This immediately gives us an asymptotic formula for the central binomial coefficient as:
\begin{lemma}
\label{lem:central_binomial_coefficient}
The central binomial coefficient can be approximated as:
\[
\binom{n}{n/2} = \sqrt{\frac{2}{\pi n}}2^n\left( 1 + \mathcal O\left( \frac 1n\right)\right)
\]
\end{lemma}

Now, we consider a ``off-center'' Binomial coefficient $\binom{n}{k}$ where $k = \frac{n+c\sqrt n}{2}$.
\begin{lemma}
\label{lem:close_to_central_binomial_coefficient}
\[
\binom{n}{k} = \binom{n}{n/2} e^{-c^2/2}exp\left( \mathcal O(c^3/\sqrt n)\right)
\]
\end{lemma}
\begin{proof}[Proof from \cite{asymptopia}]
We consider the ratio: $R = \binom{n}{k}/\binom{n}{n/2}$:
\begin{align}
R &= \frac{\binom{n}{k}}{\binom{n}{n/2}}
= \frac{(n/2)!(n/2)!}{k!(n-k)!} = \mathlarger\prod\limits_{i=1}^{c\sqrt n/2}\frac{n/2-i+1}{n/2+i}\\
\implies \log R &= \mathlarger\sum\limits_{i=1}^{c\sqrt n/2}\log\left(\frac{n/2-i+1}{n/2+i} \right)\\
&= \mathlarger\sum\limits_{i=1}^{c\sqrt n/2} - \frac{4i}{n} + \mathcal O\left( \frac{i^2}{n^2}\right)
= - \frac{c^2n}{2n} + \mathcal O\left( \frac{(c\sqrt n)^3}{n^2}\right)
= - \frac{c^2}{2} + \mathcal O\left( \frac{c^3}{\sqrt n}\right)\\
\implies \binom{n}{k} &= \binom{n}{n/2} e^{-c^2/2}exp\left( \mathcal O(c^3/\sqrt n)\right)
\end{align}
\end{proof}

\subsection{Dyck Path Boundaries and Deviations}%
\label{sec:dyck_path_boundaries_and_deviations}
\begin{lemma}
\label{lem:random_walk_deviation_bound}
Given a random walk of length $2n$ with exactly $n$ up and down steps,
consider a contiguous \emph{sub-path} of length $2B$ that comprises of $U$ up-steps and $D$ down-steps i.e. $U + D = 2B$.
Both $|B-U|$ and $|B-D|$ are $\mathcal O(\sqrt{B\log n})$ with probability at least $1-1/n^4$.
\end{lemma}
\begin{proof}
We consider the random walk as a sequence of unbiased random variables $\{X_i\}_{i=1}^{2n}\in \{0,1\}^{2n}$
with the constraint $\sum\limits_{i=1}^{2n}X_i = n$.
Here, $1$ corresponds to an up-step and $0$ corresponds to a down step.
Because of the constraint, $X_i, X_j$ are negatively correlated for $i \not= j$ which allows us to apply Chernoff bounds.
Now we consider a sub-path of length $2B$ and let $U$ denote the sum of the $X_i$s associated with this subpath.
Using Chernoff bound with $\mathbb E[X] = B$, we get:
\[
\mathbb P\left[ |U-B| < 3\sqrt{B \log n}\right]
= \mathbb P\left[ |U-B| < 3\frac{\sqrt{\log n}}{\sqrt B}B\right] < e^{-\frac{9\log n}{3}} = \frac{1}{n^3}
\]
Since $U$ and $D$ are symmetric, the same argument applies.
\end{proof}

\begin{corollary}
\label{cor:random_walk_deviation_bound}
With high probability, every contiguous sub-path of length $2B$ (with $U$ up and $D$ down steps such that $U+D=2B$) in the random walk
satisfies the property that $|B-U|$ and $|B-D|$ are upper bounded by $c\sqrt{B\log n}$ w.h.p. $1-1/n^2$ (for some constant $c$).
\end{corollary}
\begin{proof}
We can simply apply Lemma~\ref{lem:random_walk_deviation_bound} and union bound over all $n^2$ possible contiguous sub-paths.
\end{proof}

\DyckPathDeviationBound*
\begin{proof}
As a consequence of Theorem~\ref{thm:number_of_dyck_paths}, we can sample a Dyck path
by first sampling a \emph{balanced} random walk with $n$ up steps and $n$ down steps and adding an initial up step.
We can then find the corresponding Dyck path by taking the unique cyclic shift that satisfies the Dyck constraint (after removing the initial up-step).
Any interval in a cyclic shift is the union of at most two intervals in the original sequence.
This affects the bound only by a constant factor.
So, we can simply use Corollary~\ref{cor:random_walk_deviation_bound} to finish the proof.
Notice that since $|U-D| \le |B-U|+|B-D|$, $|U-D| = \mathcal O(\sqrt{B\log n})$ comes for free.
\end{proof}

\DyckPathIrrelevantBoundary*
\begin{proof}
We use $\mathcal D$ and $\mathcal R$ to denote the set of all valid Dyck paths and all random sequences respectively.
Clearly, $\mathcal D\subseteq \mathcal R$.
Since the random walk/sequence distribution is uniform on $\mathcal R$, and by Corollary~\ref{cor:random_walk_deviation_bound}
we see that at least $1-1/n^2$ fraction of the elements of $\mathcal R$ do not violate the boundary constraint.
Therefore, $|\mathcal D|\ge (1-1/n^2)|\mathcal R|$,
and so the $L_1$ distance between the uniform distributions on $\mathcal D$ and $\mathcal R$ is $\mathcal O(1/n^2)$.
\end{proof}

\subsection{Estimating the Sampling Probabilities}%
\label{sec:computing_probabilities}
\begin{lemma}
\label{lem:probability_approximation_oracle}
Given a Dyck sub-path problem within a global Dyck path of size $2n$ and a probability expression of the form
$p_d = \frac{S_{left}\cdot S_{right}}{S_{total}}$, there exists a $\poly(\log n)$ time oracle that returns a
$\left( 1\pm 1/n^2\right)$ multiplicative approximation to $p_d$ if $p_d = \Omega(1/n^2)$ and returns $0$ otherwise.
\end{lemma}
\begin{proof}
We first compute a $1+1/n^3$ multiplicative approximation to $\ln p_d$.
Using $\mathcal O(\log n)$ terms of the series in Equation~\ref{eq:log_factorial_approximation},
it is possible to estimate the logarithm of a factorial up to $1/n^c$ additive error.
So, we can use the series expansion from Equation~\ref{eq:log_factorial_approximation} up to $\mathcal O(\log n)$ terms.
The additive error can also be cast as multiplicative since factorials are large positive integers.

The probability $p_d$ can be written as an arithmetic expression involving sums and products of a constant number of factorial terms.
Given a $1\pm1/n^c$ multiplicative approximation to $l_a = \ln a$ and $l_a = \ln b$, we wish to approximate $\ln(ab)$ and $\ln(a+b)$.
The former is trivial since $\ln(ab) = l_a + l_b$.
For the latter, we assume $a>b$ and use the identity $\ln(a+b) = \ln a + \ln(1+b/a)$ to note that it suffices to approximate $\ln(1+b/a)$.
We define $\widetilde l_a = l_a\cdot(1\pm \mathcal O(1/n^c))$ and $\widetilde l_b = l_b\cdot(1\pm \mathcal O(1/n^c))$.
In case $\widetilde l_b-\widetilde l_a < -c\ln n\implies b/a < 1/n^c$,
we approximate $\ln(a+b)$ by $\ln a$ since $\ln(1+b/a) = \mathcal O(1/n^c)$ in this case.

Otherwise, we notice that $\max(a,b) = \mathcal O\left((n!)^2\right)\implies \max(l_a, l_b) = o(n^3)\implies l_a-l_b = o(n^3)$.
This is true because if we write out the expression for $p_d = S_{left}\cdot S_{right}/S_{total}$ in terms of sums and products of factorials,
the largest possible value will be a product of at most two terms that are present in the numerator/denominator of a binomial term
(see the expression for generalized Catalan numbers in Equation~\ref{eq:catalan_trapezoid}).
Hence, the claim $\max(a,b) = \mathcal O\left((n!)^2\right)$ follows from the fact that
the numerators/denominators of the relevant binomial coefficient in Equation~\ref{eq:catalan_trapezoid} are at most $n!$.
Using this fact, we obtain the following:
\[
1+e^{\widetilde l_b-\widetilde l_a} = 1+\frac ba\cdot \mathlarger e^{\mathcal O\left( \frac{l_b-l_a}{n^c}\right)}
= 1+\frac ba\cdot\left( 1 \pm \mathcal O\left( \frac{1}{n^{c-3}} \right)\right)
= \left( 1+\frac ba\right)\cdot\left( 1 \pm \mathcal O\left( \frac{1}{n^{c-3}} \right)\right)
\]
In other words, the value of $c$ decreases every time we have a sum operation.
Since there are only a constant number of such arithmetic operations in the expression for $p_d$,
we can set $c$ to be a high enough constant when approximating the factorials,
cewhich allows us to obtain the desired $1\pm1/n^3$ multiplicative approximation to $\ln p_d$.
If $\ln p_d < -3\ln n$, we approximate $p_d = 0$.
Otherwise, we can exponentiate the approximation to obtain $\widetilde p_d=p_d\cdot e^{-\mathcal O(\ln n/n^3)} = p_d\left( 1\pm\mathcal O(1/n^2)\right)$.
\end{proof}

\subsection{Omitted Proofs from Section~\ref{sec:implementing_height_queries}: Sampling the Height}%
\label{sec:appendix_implementing_height_queries}
\begin{lemma}
\label{lem:taylor_bound}
For $x < 1$ and $k\ge 1$, we claim that $1-kx < (1-x)^k < 1 - kx + \frac{k(k-1)}{2}x^2$
\end{lemma}

\DLeftBound*
\begin{proof}
This involves some simple manipulations.
\begin{align}
S_{left} &= {{B}\choose{D-d}}-{{B}\choose{D-d-k}}\\
&= {{B}\choose{D-d}}\cdot \left[1-\frac{(D-d)(D-d-1)\cdots(D-d-k+1)}{(B-D-d+k)(B-D-d+k-1)\cdots(B-D-d+1)}\right]\\
&\le {{B}\choose{D-d}}\cdot \left[1-\left(\frac{D-d-k+1}{B-D+d+k}\right)^k\right]\\
&\le {{B}\choose{D-d}}\cdot \left[1-\left(\frac{U+d+k-(U-D+d+k-1)}{U+d+k}\right)^k\right]\\
&\le {{B}\choose{D-d}}\cdot \left[1-\left(\frac{U+d+k-\Bo(\sqrt{B\log n})}{U+d+k}\right)^k\right]\\
&\le \Theta \left( \frac{k\sqrt{\log n}}{\sqrt{B}} \right) \cdot{{B}\choose{D-d}}
\end{align}
\end{proof}

\DRightBound*
\begin{proof}
\begin{align}
S_{right} &= {{B}\choose{U-d}}-{{B}\choose{U-d-k'}}\\
&= {{B}\choose{U-d}}\cdot \left[1-\frac{(U-d)(U-d-1)\cdots(U-d-k'+1)}{(B-U-d+k')(B-U-d+k'-1)\cdots(B-U-d+1)}\right]\\
&\le {{B}\choose{U-d}}\cdot \left[1-\left(\frac{U-d-k'+1}{B-U+d+k'}\right)^{k'}\right]\\
&\le {{B}\choose{U-d}}\cdot \left[1-\left(\frac{2D-U-d-k+1}{2U-D+k+d}\right)^{k'}\right]\\
&\le {{B}\choose{U-d}}\cdot \left[1-\left(\frac{U+k+d - (2U-2D+2d+2k-1)}{U+k+d}\right)^{k'}\right]\\
&\le {{B}\choose{U-d}}\cdot \left[1-\left(\frac{U+k+d - \Bo(\sqrt{B\log n})}{U+k+d}\right)^{k'}\right]\\
&\le \Theta\left(\frac{k'\sqrt{\log n}}{\sqrt{B}}\right)\cdot{{B}\choose{U-d}}
\end{align}
\end{proof}

\begin{restatable}{lemma}{DTotalBasicBound}
\label{lem:DTotalBasicBound}
$S_{tot} \ge {{2B}\choose{2D}}\cdot \left[1-\left(1 - \frac{k'}{2U+1}\right)^k\right]$.
\end{restatable}
\begin{proof}
\begin{align}
S_{tot} &= {{2B}\choose{2D}}-{{2B}\choose{2D-k}}\\
&= {{2B}\choose{2D}}\cdot \left[1-\frac{(2D)(2D-1)\cdots(2D-k+1)}{(2B-2D+k)(2B-2D+k-1)\cdots(2B-2D+1)}\right]\\
&\ge {{2B}\choose{2D}}\cdot \left[1-\left(\frac{2D-k+1}{2B-2D+1}\right)^k\right]\\
&\ge {{2B}\choose{2D}}\cdot \left[1-\left(\frac{2U-(2U-2D+k-1)}{2U+1}\right)^k\right]\\
&\ge {{2B}\choose{2D}}\cdot \left[1-\left(\frac{(2U+1)-k'}{2U+1}\right)^k\right]\\
&\ge {{2B}\choose{2D}}\cdot \left[1-\left(1 - \frac{k'}{2U+1}\right)^k\right]\\
\end{align}
\end{proof}

\DTotalFarBoundary*
\begin{proof}
When $kk' > 2U + 1 \implies k > \frac{2U+1}{k'}$,
we will show that the above expression is greater than $\frac 12 \binom{2B}{2D}$.
Defining $\nu = \frac{2U+1}{k'} > 1$, we see that $(1-\frac 1\nu)^k \le (1-\frac 1\nu)^\nu$.
Since this is an increasing function of $\nu$ and since the limit of this function is $\frac 1e$,
we conclude that
\[
1-\left(1 - \frac{k'}{2U+1}\right)^k > \frac 12
\]
\end{proof}

\DTotalNearBoundary*
\begin{proof}
Now we bound the term $1-\left(1-\frac{k'}{2U+1}\right)^k$, given that $kk'\le 2U+1\implies\frac{kk'}{2U+1} \le 1$.
Using Bernoulli's inequality, we know that $(1 - x)^n \le 1/(1+nx)$ for $x\in [0, 1]$ and $n\in \mathbb N$.
Since the term $k'/(2U+1)$ is positive and $\le 1$ (since $kk' \le 2U+1$), we can apply this as follows:
\begin{align}
&1 - \left(1 - \frac{k'}{2U+1}\right)^k\\
&\ge 1 - \frac{1}{1 + \frac{kk'}{2U+1} } = \frac{\frac{kk'}{2U+1}}{1 + \frac{kk'}{2U+1} } = \frac{kk'}{2U+1}\cdot \frac{1}{1 + \frac{kk'}{2U+1}} \\
&\ge \frac{kk'}{2U+1}\cdot \frac{1}{2} &\left( \textrm{Since }\frac{kk'}{2U+1}\le 1\right) \\
&\ge \frac{kk'}{\Theta(B)}
\end{align}
\end{proof}

\subsection{Omitted Proofs from Section~\ref{sec:supporting_first_return_queries}: \func{First-Return} Queries}
\label{sec:omitted_supporting_first_return_queries}

\ReturnDLeftBound*
\begin{proof}
In what follows, we will drop constant factors:
Refer to Figure~\ref{fig:dyck_mandatory_boundary_sampling} for the setup.
The left section of the path reaches one unit above the boundary (the next step would make it touch the boundary).
The number of up-steps on the left side is $d$ and therefore the number of down steps must be $d + k - 2$.
This inclues $d$ down steps to cancel out the upwards movement, and $k-2$ more to get to one unit above the boundary.
The boundary for this section is $k' = k-1$. This gives us:
\begin{align}
S_{left}(d) &= \binom{2d+k-2}{d} - \binom{2d+k-2}{d-1}\\
&= \binom{2d+k-2}{d}\left[ 1-\frac{d}{d+k-1}\right] = \binom{2d+k-2}{d}\frac{k-1}{d+k-1}
\end{align}
Now, letting $z = 2d+k-2$,  we can write $d = \frac{z-(k-2)}{2} = \frac{z-\frac{k-2}{\sqrt z}\sqrt z}{2}$.
Using Lemma~\ref{lem:DyckPathDeviationBound}, we see that $\frac{k-2}{\sqrt z}$ should be $\mathcal O(\sqrt{\log n})$.
If this is not the case, we can simply return $0$ because the probability associated with this value of $d$ is negligible.
Since $z > \log^4 n$, we can apply Lemma~\ref{lem:close_to_central_binomial_coefficient} to get:
\[
S_{left}(d) = \Theta\left( \binom{z}{z/2} e^{\frac{(k-2)^2}{2z}} \frac{k-1}{d+k-1} \right)
= \Theta\left( \frac{2^{2d+k}}{\sqrt d} e^{\frac{(k-2)^2}{2(2d+k-2)}} \frac{k-1}{d+k-1} \right)
\]
\end{proof}

\ReturnDRightBound*
\begin{proof}
The right section of the path starts from the original boundary.
Consequently, the boundary for this section is at $k' = 1$.
The number of up-steps on the right side is $U-d$ and the number of down steps is $D-d-k+1$.
This gives us:
\begin{align}
S_{right}(d) &= \binom{U+D-2d-k+1}{U-d} - \binom{U+D-2d-k+1}{U-d+1}\\
&= \binom{U+D-2d-k+1}{U-d}\left[ 1-\frac{D-d-k-1}{U-d+1}\right]\\
&= \binom{U+D-2d-k+1}{U-d}\frac{U-D+k}{U-d+1}
\end{align}
Now, letting $z = U+D-2d-k+1$,  we can write $U-d = \frac{z+(U-D+k-1)}{2} = \frac{z+\frac{U-D+k-1}{\sqrt z}\sqrt z}{2}$.
Using Lemma~\ref{lem:DyckPathDeviationBound}, we see that $\frac{k-2}{\sqrt z}$ should be $\mathcal O(\sqrt{\log n})$.
If this is not the case, we can simply return $0$ because the probability associated with this value of $d$ is negligible.
Since $z > \log^4 n$, we can apply Lemma~\ref{lem:close_to_central_binomial_coefficient} to get:
\begin{align}
S_{right}(d) &= \Theta\left( \binom{z}{z/2} e^{\frac{(U-D+k-1)^2}{2z}} \frac{U-D+k}{U-d+1} \right)\\
&= \Theta\left( \frac{2^{U+D-2d-k}}{\sqrt{U+D-2d-k}} e^{\frac{(U-D+k-1)^2}{2(U+D-2d-k+1)}} \frac{U-D+k}{U-d+1} \right)
\end{align}
\end{proof}

\clearpage
\section{Additional related work}
\label{sec:additional_related_work}

\paragraph*{Random graph models}
The Erd\"{o}s-R\'{e}nyi model, given in \cite{er}, is one of the most simple theoretical random graph model,
yet more specialized models are required to capture properties of real-world data.
The Stochastic Block model (or the planted partition model) was proposed in \cite{holland} originally for modeling social networks;
nonetheless, it has proven to be an useful general statistical model in numerous fields,
including recommender systems \cite{rec0,rec1}, medicine \cite{med0}, social networks \cite{social0,social1},
molecular biology \cite{bio0,bio1}, genetics \cite{gene0,gene1,gene2}, and image segmentation \cite{img0}.
Canonical problems for this model are the community detection and community recovery problems:
some recent works include \cite{chin2015stochastic,mossel2015reconstruction,abbe2015community,abbe2016exact};
see e.g., \cite{abbe} for survey of recent results.
The study of Small-World networks is originated in \cite{watts1998collective} has frequently been observed,
and proven to be important for the modeling of many real world graphs such as social networks \cite{small0, small1},
brain neurons \cite{bassett2006small}, among many others.
Kleinberg's model on the simple lattice topology (as considered in this paper) imposes a geographical that allows navigations,
yielding important results such as routing algorithms (decentralized search) \cite{kleinberg, klein}.
See also e.g., \cite{newman2000models} and Chapter 20 of \cite{easley2010networks}.

\paragraph*{Generation of random graphs}
The problem of local-access implementation of random graphs has been considered in the aforementioned work \cite{huge_old,sparse,reut},
as well as in \cite{mansour} that locally generates out-going edges on bipartite graphs while minimizing the maximum in-degree.
The problem of generating full graph instances for random graph models have been frequently considered in many models of computations,
such as sequential algorithms \cite{milo2003uniform,er_gen,nobari2011fast,miller2011efficient},
and the parallel computation model \cite{alam2017parallel}.

\paragraph*{Query models}
In the study of sub-linear time graph algorithms where reading the entire input is infeasible,
it is necessary to specify how the algorithm may access the input graph,
normally by defining the type of queries that the algorithm may ask about the input graph;
the allowed types of queries 
can greatly affect the performance of the algorithms.
While \func{Next-Neighbor} query is only recently considered in \cite{reut},
there are other query models providing a neighbor of a vertex,
such as asking for an entry in the adjacency-list representation \cite{goldreich1997property},
or traversing to a random neighbor. On the other hand,
the
\func{Vertex-Pair} query is common in the study of dense graphs as accessing the adjacency matrix representation \cite{goldreich1998property}.
The \func{All-Neighbors} query has recently been explicitly considered in local algorithms \cite{feige2017probe}.

Other constructions of huge pseudorandom functions that are permutations or random hash functions were given in \cite{luby_rackoff, naor, mansour}.

\end{document}